\documentclass[sigconf]{acmart} 
\def\usingIEEEtran{0}
\def\usingACM{1}

\def\isanonymous{0}
\def\isfullversion{1}
\def\isacmccsreferences{0}

\def\istoappearfootnote{1}

\usepackage{ifthen}
\newcommand{\anonymous}[2]{%
	\ifthenelse{\equal{\isanonymous}{1}}%
	{#1}%
	{#2}%
}%

\newcommand{\fullversion}[2]{%
	\ifthenelse{\equal{\isfullversion}{1}}%
	{#1}%
	{#2}%
}%

\newcommand{\IEEEtrans}[2]{%
	\ifthenelse{\equal{\usingIEEEtran}{1}}%
	{#1}%
	{#2}%
}%

\newcommand{\ACM}[2]{%
	\ifthenelse{\equal{\usingACM}{1}}%
	{#1}%
	{#2}%
}%

\newcommand{\acmccsreferences}[2]{%
	\ifthenelse{\equal{\isacmccsreferences}{1}}%
	{#1}%
	{#2}%
}%

\newcommand{\toappearfootnote}[2]{%
	\ifthenelse{\equal{\istoappearfootnote}{1}}%
	{#1}%
	{#2}%
}%

\acmccsreferences{
\pagenumbering{gobble}
}{

}

\usepackage{subcaption}
\usepackage{adjustbox}
\usepackage[inline]{enumitem}

\ACM{
	\definecolor{Gray}{RGB}{50,50,50}
	\definecolor{oxygenorange}{HTML}{FFDD00}
	\usepackage[color=oxygenorange]{todonotes}
	\hypersetup{bookmarksdepth=2,colorlinks=true,citecolor=ACMBlue,linkcolor=ACMBlue}
}{
	\usepackage{xcolor}	
	\definecolor{Gray}{RGB}{50,50,50}
	\definecolor{NavyBlue}{RGB}{0,0,128}
	\definecolor{ForestGreen}{RGB}{34, 139, 34}
	\usepackage[bookmarksdepth=2,colorlinks=true,citecolor=NavyBlue,linkcolor=ForestGreen]{hyperref}
	
	\usepackage{natbib}
}

\newcommand\Wider[2][3em]{%
	\makebox[\linewidth][c]{%
		\begin{minipage}{\dimexpr\textwidth+#1\relax}
			\raggedright#2
		\end{minipage}%
	}%
}

\IEEEtrans{
	\usepackage{amsthm}
	
	\newtheorem{definition}{Definition}
	\newtheorem{lemma}{Lemma}
	
	\newtheorem{remark}{Remark}
	\newtheorem{corollary}{Corollary}
}{
	
}

\ACM{
\newtheorem{remark}{Remark}
}
{}

\newcommand{\queedee}{
	\IEEEtrans{\tag*{\qedhere}}{\tag*{\ensuremath{\Box}}}
}
 
\usepackage{mathtools}

\usepackage{xspace}
\usepackage[
	lambda,
	landau,
	operators,
	probability,
	sets,
	logic,
	complexity,
	asymptotics
]{cryptocode}
\usepackage{thmtools, thm-restate}

\newcommand{\pcmycomment}[1]{\textcolor{Gray}{\texttt{//~#1}}}

\usepackage{tikz,pgfplots}
\usetikzlibrary{shapes, arrows, calc, arrows.meta, fit, positioning, fit, backgrounds} 

\definecolor{DarkPurple}{HTML}{332288}
\definecolor{DarkBlue}{HTML}{6699CC}
\definecolor{LightBlue}{HTML}{88CCEE}
\definecolor{DarkGreen}{HTML}{117733}
\definecolor{DarkRed}{HTML}{661100}
\definecolor{LightRed}{HTML}{CC6677}
\definecolor{LightPink}{HTML}{AA4466}
\definecolor{DarkPink}{HTML}{882255}
\definecolor{LightPurple}{HTML}{AA4499}
\definecolor{DarkBrown}{HTML}{604c38}
\definecolor{DarkTeal}{HTML}{23373b}
\definecolor{LightBrown}{HTML}{EB811B}
\definecolor{LightGreen}{HTML}{14B03D}

\pgfkeys{
  /foo/bar/.initial=1,
}

\usepackage[capitalize,noabbrev]{cleveref}

\ACM{
\fancyhf{} 
\fancyfoot[C]{\thepage}

\setcopyright{none} 

\settopmatter{printacmref=false, printccs=true, printfolios=true} 
}{}

\DeclareFontFamily{OT1}{pzc}{}
\DeclareFontShape{OT1}{pzc}{m}{it}%
{<-> s * [1.150] pzcmi7t}{}
\DeclareMathAlphabet{\mathscript}{OT1}{pzc}{m}{it}
\DeclareMathAlphabet{\mathsc}{OT1}{cmr}{m}{sc}
\DeclareMathAlphabet{\mathsl}{OT1}{cmr}{m}{sl}

\newcommand*{\getsr}{{\:{\leftarrow{\hspace*{-3pt}\raisebox{.75pt}{$\scriptscriptstyle\$$}}}\:}} 
\mathchardef\mhyphen="2D 

\renewcommand{\nn}{\nonumber}

\newcommand{\up}{{\textsf{up}}}
\newcommand{\qry}{{\textsf{qry}}}

\newcommand{\amq}{\text{AMQ-PDS}\xspace}

\newcommand{\amqs}{\text{AMQ-PDS}\xspace}

\newcommand{\repsimO}{\textbf{RepSim}}
\newcommand{\PermuteO}{\textbf{Per}}
\newcommand{\upsimO}{\textbf{UpSim}}

\newcommand{\qrysimO}{\textbf{QrySim}}

\newcommand{\revealsimO}{\textbf{RevSim}}
\newcommand{\simS}{\mathcal{S}}

\newcommand{\setupS}{\textsf{setup}}
\newcommand{\upS}{\textsf{up}}

\newcommand{\qryS}{\textsf{qry}}

\newcommand{\somealgS}{\textsf{alg}}
\newcommand{\nai}{\text{NAI}}
\newcommand{\naigen}{\text{\nai-gen}}
\newcommand{\statdist}{SD}

\newcommand{\pfpsym}{P_{\Pi,pp}}
\newcommand{\pfp}[1]{\pfpsym(FP\,|\,{#1})}
\newcommand{\opfp}[1]{\overline{\pfpsym}(FP \mid {#1})}

\newcommand{\Load}[1]{\text{load}(\ensuremath{#1})}

\newcommand{\Values}[1]{\text{vals}(\ensuremath{#1})}

\newcommand{\setV}{{V}}

\newcommand{\E}{\mathbf{E}}

\newcommand{\tinit}{\text{init}}
\newcommand{\tupenabled}{\text{UpIsEnabled}}
\newcommand{\tinserted}{\text{inserted}}
\newcommand{\tFPlist}{\text{FPlist}}
\newcommand{\tCALQ}{\text{CALQ}}

\newcommand{\repO}{\textbf{Rep}}
\newcommand{\repleakO}{\ensuremath{\textbf{RepLeak}}\xspace}
\newcommand{\elemleakO}{\ensuremath{\textbf{ElemLeak}}\xspace}

\newcommand{\upO}{\textbf{Up}}

\newcommand{\qryO}{\textbf{Qry}}
\newcommand{\revO}{\textbf{Reveal}}

\newcommand{\RRO}{\textbf{RoR}}

\newcommand{\advA}{{\mathcal{A}}}

\newcommand{\advB}{{\mathcal{B}}}

\newcommand{\dD}{\mathcal{D}}

\newcommand*{\Kspace}{\mathcal{K}}

\newcommand{\bits}{\{0,1\}}

\newcommand*{\Perms}{\mathsf{Perms}}
\newcommand*{\Funcs}{\mathsf{Funcs}}
\newcommand*{\setS}{S}

\newcommand*{\statespace}{\Sigma}

\newcommand*{\identity}{\text{Id}}
\newcommand*{\plist}{\mathcal{P}_\text{lists}}

\newcommand*{\gameRoI}{\text{Real-or-Ideal}}

\newcommand*{\gameRoIERP}{\text{R-or-I-ElemRepPriv}}
\newcommand*{\gameRoIRP}{\text{R-or-I-RepPriv}}
\newcommand*{\gamePI}{Exp^{PI}_\Pi}

\newcommand*{\Ideal}{\textit{Ideal}}
\newcommand*{\Real}{\textit{Real}}

\newcommand{\pr}[1]{{\Pr}\left[\,#1\,\right]}

\newlength{\saveparindent}
\newlength{\saveparskip}
\newcounter{ctr}

\renewcommand{\eqref}[1]{\mbox{Eq.~(\ref{#1})}}

\def\next{\:;\:}

\renewcommand{\tag}{\mathrm{tag}}

\newcommand{\domain}{\mathfrak{D}}
\newcommand{\range}{\mathfrak{R}}

\newcommand{\spaceK}{\mathcal{K}}
\newcommand{\coinspace}{\mathcal{R}}

\ACM{
\acmccsreferences{
	\copyrightyear{2022}
	\acmYear{2022}
	\setcopyright{acmlicensed}
	\acmConference[CCS '22] {Proceedings of the 2022 ACM SIGSAC Conference on Computer and Communications Security}{November 7--11, 2022}{Los Angeles, CA, USA}
	\acmBooktitle{Proceedings of the 2022 ACM SIGSAC Conference on Computer and Communications Security (CCS '22), November 7--11, 2022, Los Angeles, CA, USA}
	\acmPrice{15.00}
	\acmDOI{10.1145/3548606.3560621}
	\acmISBN{978-1-4503-9450-5/22/11}
	
	\settopmatter{printacmref=true}
}{
	\settopmatter{printccs=false, printfolios=false} 
	\renewcommand\footnotetextcopyrightpermission[1]{} 
	\thispagestyle{plain}
	\pagestyle{plain} 
}
}{}

\begin{document}

\title{Adversarial Correctness and Privacy for Probabilistic Data Structures}
\ACM{
\acmccsreferences{
	}{
		\toappearfootnote{
			\ifthenelse{\equal{\isfullversion}{1}}%
			{\titlenote{This is the full version of a paper to appear at ACM CCS '22.}}%
			{\titlenote{This is a version of a paper to appear at ACM CCS '22.}}%
		}{}
	}
}{}

\ACM{
\author{Mia Fili\'{c}}
\affiliation{%
	\institution{ETH Z\"urich}
	\city{Z\"urich}
	\country{Switzerland}}
\email{mia.filic@inf.ethz.ch}

\author{Kenneth G. Paterson}
\affiliation{%
	\institution{ETH Z\"urich}
	\city{Z\"urich}
	\country{Switzerland}}
\email{kenny.paterson@inf.ethz.ch}

\author{Anupama Unnikrishnan}
\affiliation{%
	\institution{ETH Z\"urich}
	\city{Z\"urich}
	\country{Switzerland}}
\email{anupama.unnikrishnan@inf.ethz.ch}

\author{Fernando Virdia}
\authornote{The work of Virdia was carried out while employed by ETH Z\"urich.}
\affiliation{%
	\institution{Intel Labs}
	\city{Z\"urich}
	\country{Switzerland}}
\email{fernando.virdia@intel.com}
}{
	\maketitle
	\thispagestyle{plain}
	\pagestyle{plain}
}

\begin{abstract}
	We study the security of Probabilistic Data Structures (PDS) for handling Approximate Membership Queries (AMQ); prominent examples of AMQ-PDS are Bloom and Cuckoo filters. 
	AMQ-PDS are increasingly being deployed in environments where adversaries can gain benefit from carefully selecting inputs, for example to increase the false positive rate of an AMQ-PDS. 
	They are also being used in settings where the inputs are sensitive and should remain private in the face of adversaries who can access an AMQ-PDS through an API or who can learn its internal state by compromising the system running the AMQ-PDS. 
	
	We develop simulation-based security definitions that speak to correctness and privacy of AMQ-PDS. 
	Our definitions are general and apply to a broad range of adversarial settings. 
	We use our definitions to analyse the behaviour of both Bloom filters and insertion-only Cuckoo filters. 
	We show that these AMQ-PDS can be provably protected through replacement or composition of hash functions with keyed pseudorandom functions in their construction.
	We also examine the practical impact on storage size and computation of providing secure instances of Bloom and insertion-only Cuckoo filters.
\end{abstract}




\ACM{
\maketitle
\acmccsreferences{
	}{
	\thispagestyle{plain}
	\pagestyle{plain}
	}
}{}

\section{Introduction}\label{sec:introduction}
Probabilistic data structures (PDS) are becoming ubiquitous in practice. They enjoy improved efficiency over exact data structures but this comes at the cost of them only giving approximate answers. Still, this is sufficient in many use-cases, for example when computing statistics on large data sets (e.g.\ finding the number of distinct items in the set~\cite{flajolet2007hyperloglog}, or moments of the frequency distribution of the set elements~\cite{DBLP:conf/stoc/AlonMS96}), in answering set membership queries (e.g.\ to decide, in a storage-efficient manner, whether a particular data item has been encountered before~\cite{bloom1970space}), or in identifying so-called \emph{heavy hitters} in a set, that is, high frequency elements~\cite{Cormode2008b}. Often, these problems appear in the streaming setting, where memory is limited and items have to be processed in an online manner. 

We refer to PDS that provide approximate answers to membership queries as AMQ-PDS; they are the focus of this paper. Prominent examples of AMQ-PDS include Bloom filters~\cite{bloom1970space}, Cuckoo filters~\cite{_CoNEXT:FAKM14}, and Morton filters~\cite{DBLP:journals/pvldb/BreslowJ18}. \amq find use in applications ranging from certificate revocation systems (eg. CRLite)~\cite{_SP:LCLMMW17}, database query speedup~\cite{E-S:PMBS} and DNA sequence analysis~\cite{BMC:MelPri}. 

As early as the 1990s it was recognised that simple data structures like hash tables could be manipulated into poor performance or leak information about their inputs~\cite{LiptonN93}. Increasingly, PDS are being used in environments where the input may be adversarially chosen, see~\cite{_DSN:GerKumLau15,_CCS:ClaPatShr19,PatRay22} for examples. 
However, PDS are not usually designed to work reliably in such settings, and their performance guarantees are typically proven for the case where the input distribution is benign. 
This disjunction can result in security vulnerabilities, for example, non-detection of network scanning attacks~\cite{PatRay22} or a Bloom filter reporting false positives on some targeted set of inputs~\cite{_DSN:GerKumLau15}. PDS may also inadvertently leak information about their inputs in settings where it is desirable to keep those inputs private, cf.\ \cite{_PETS:KKDM11,DBLP:conf/psd/BianchiBL12,DBLP:conf/acsac/GervaisCKG14,cardestprivacy}. Given the rising use of PDS in potentially malicious settings, there is an urgent need to better understand their performance in such settings, and to secure them against adversarial manipulation. 

\subsection{Our Contributions} 
In this paper, we ask (and answer) the question: \emph{How can we provably protect PDS against attacks at low cost?}
Given the vast range of PDS in the literature and in use, we focus our attention here on AMQ-PDS. These are amongst the most widely used PDS, with applications in dictionaries, databases and networking~\cite{DBLP:journals/im/BroderM03}. Bloom filters in particular are in common use, because of their implementation simplicity and easily understood performance guarantees; Cuckoo filters were more recently introduced~\cite{_CoNEXT:FAKM14} and offer superior performance to Bloom filters.

\paragraph{Syntax for AMQ-PDS} After establishing a syntax for AMQ-PDS, inspired by~\cite{_CCS:ClaPatShr19}, we surface \emph{consistency rules} satisfied by Bloom and insertion-only Cuckoo filters which we use to prove security theorems.

\paragraph{Simulation-based security.} We then develop simulation-based security definitions for analysing the security of AMQ-PDS. Such an approach was recently used to study cardinality estimation~\cite{PatRay22}, an important (but distinct) application domain for PDS. Inspired by~\cite{PatRay22}, our approach introduces two worlds, a real world in which the adversary interacts with an AMQ-PDS instantiation through an API (allowing it to e.g. insert items and make membership queries) to produce some output, and an ideal world in which a simulator $\simS$ produces the output. Our core security definitions say that the two outputs, one taken from each world, should be close.\footnote{Standard cryptographic notions here would mandate ``indistinguishable'' in place of ``close''. Our results do not reach this level for concrete AMQ-PDS like Bloom and insertion-only Cuckoo filters without blowing up the AMQ-PDS parameters, e.g.\ the required storage. Yet they are still useful when it comes to setting parameters in practice, for example, when limiting the false positive probability that an adversary can attain.} The intuition is that whatever the adversary can learn from interacting with the AMQ-PDS in the real world can be simulated without access to the AMQ-PDS instantiation. The only leakage, then, is whatever input $\simS$ receives. Security comes from carefully evaluating this input and/or restricting the simulator to behave in certain ways which imply it is effectively operating in a non-adversarial manner.

Our simulation-based approach appears conceptually complex, but it enables us to provide a fairly unified approach to both correctness (under adversarial input) and privacy of AMQ-PDS. It avoids the use of very specific winning conditions of the type used in~\cite{_CCS:ClaPatShr19}, since security is defined purely in terms of closeness between the real and ideal world. For an introduction to the simulation-based approach, we recommend~\cite{_EPRINT:Lindell16}.

\paragraph{Adversarial correctness.} We begin by establishing what degree of correctness it is reasonable to expect from an AMQ-PDS under non-adversarial inputs. Intuitively, we want to capture the behaviour of a given AMQ-PDS under ``honest'' or ``non-adversarial'' inputs. In principle, the distribution of queries that an honest user makes to a PDS will be application-dependent: a PDS deployed to store images will receive different inputs than one deployed to store IP addresses. However, we identify two properties of AMQ-PDS (satisfied by Bloom and insertion-only Cuckoo filters) that allow us to predict the state of an \amq under honest inputs, in an application-independent manner.
We use these to formalise the notion of a non-adversarially-influenced (NAI) state generator for an AMQ-PDS, and then define the NAI false positive probability. Essentially, this is the rate at which the AMQ-PDS incorrectly answers membership queries after having been populated with $n$ randomly chosen elements (where $n$ is a parameter of the definition), reflecting ``average case'' performance of the AMQ-PDS. This quantity can be computed or bounded from above for Bloom and Cuckoo filters, which will be sufficient for our purposes.

With this machinery in hand, we define adversarial correctness for AMQ-PDS in the following way: an AMQ-PDS is adversarially correct if there exists an ideal world simulator that provides a view of the AMQ-PDS to the adversary that is close to the view produced by an NAI generator. The definition implies that the adversary cannot tell the difference between the real world setting and the ideal world setting. Moreover, in the ideal world setting, the adversary effectively sees an AMQ-PDS operating under non-adversarial conditions. It follows from this chain of reasoning that, if our definition is satisfied, an adversary can only influence the false positive probability of the AMQ-PDS to a concretely bounded extent. Here, one should contrast the security definition with those more commonly seen in the simulation-based paradigm: we get security by restricting the \emph{behaviour} of $\simS$ rather than by limiting its \emph{input}. An analogous formulation was used in~\cite{PatRay22} when dealing with adversarial correctness of cardinality estimators. 

When it comes to proving application-independent results and using our adversarial correctness definition to study specific AMQ-PDS, we identify \emph{function-decomposability} as a key property.
Informally, this property says that the input to an AMQ-PDS is always first transformed using some function $F$ before any further processing is applied. Bloom filters operate in this way ($F$ represents the collection of hash functions used to map the input element to a vector of indices). Cuckoo filters do not, but we show how they can easily be modified to do so. Indeed, any existing AMQ-PDS can have its inputs ``wrapped'' in order to make it function-decomposable. Our main result shows that if an AMQ-PDS is function-decomposable, and $F$ is a truly random function, then our notion of adversarial correctness is achievable. In practice, we instantiate $F$ by a pseudorandom function (PRF), with the PRF key being held securely by the AMQ-PDS. This step introduces a PRF-to-random-function switching cost to our security bounds.\footnote{Note that we are not operating in the Random Oracle Model here, since the adversary does not have (direct) oracle access to $F$; a good analogy is constructing a symmetric encryption mode using a random permutation $\pi$ for the analysis, and then replacing $\pi$ with a block cipher.}

\paragraph{Privacy.}  We consider two distinct but related security models for privacy of AMQ-PDS.
The main question our models address is: what information can leak to an adversary about the elements already contained in an AMQ-PDS? 
Our first model recognises that such leakage may be inevitable if the adversary has access to a sufficiently rich API for a target AMQ-PDS (in particular, if it can make membership queries to the AMQ-PDS), but that it may be possible to prove that nothing more than the results of such queries may leak. This necessitates equipping the simulator with an oracle telling when a queried element is already contained in the AMQ-PDS. The second privacy model dispenses with this oracle, and only allows the simulator to learn \emph{how many} elements have been stored in the AMQ-PDS. It captures the intuitive idea that the only way an adversary should be able to break privacy is by guessing any of the previously stored elements. As one application, our security definition ensures that, provided the previously stored set has sufficiently high min-entropy, the adversary learns nothing about that set during its attack, even when equipped with a rich API for interacting with the AMQ-PDS, and even if it can compromise the AMQ-PDS' state. We show how these two models are related.

We also identify a property of AMQ-PDS that we call \emph{permutation invariance (PI)}. This property enables us to quickly establish privacy according to our definitions. Informally, PI says that an adversary cannot tell if a random permutation has been applied to the elements before they are added to the AMQ-PDS, or not. We are able to show that PI follows from function-decomposability, so the latter property useful to establishing both adversarial correctness and privacy of an AMQ-PDS.

\paragraph{Analysis of Bloom and Cuckoo filters.} Our final contribution is to use our security models to analyse two concrete AMQ-PDS, namely Bloom and insertion-only Cuckoo filters. Neither is secure according to our definitions in their original formulations, but we show how they can be made so, provably and at low cost, by replacing hash functions used in their constructions with keyed Pseudo-Random Functions (PRFs). Thus we invoke the use of cryptographic objects to achieve security. This is not completely cost-free of course, since now users must manage the needed cryptographic keys. However, PRFs are in general only a little more expensive to compute than hash functions (since a PRF on a given domain $\domain$ can always be constructed by first hashing on $\domain$ to obtain a fixed-length string and then applying a fixed-domain PRF, e.g. one based on a block cipher like AES, and fast, dedicated PRF constructions also exist~\cite{_INDOCRYPT:AumBer12}). Since our results do not require ``indistinguishability'' of the real and ideal worlds, one could potentially design faster and cheaper ``weaker'' PRFs apt for use in \amq. Moreover,~\cite{_C:NaoYog15} have shown that adversarially correct Bloom filters \emph{imply} one-way functions, in a weaker attack model than ours. Thus, using cryptographic approaches appears to be necessary for security of AMQ-PDS. We use our analysis to compute concrete bounds on how much security can be expected from PRF-based Bloom and insertion-only Cuckoo filters. These bounds can be used by practitioners who need to invoke AMQ-PDS in adversarial settings to choose appropriate parameters for their data structures. \fullversion{The code to replicate the analysis can be found at
\href{https://github.com/securing-pds/securing-pds}{\color{black!80}{https://github.com/securing-pds/securing-pds}}.}{}

\subsection{Related Work}
 
Clayton~\emph{et al.}~\cite{_CCS:ClaPatShr19} provided a provable-security treatment of PDS. They used a game-based formalism to analyse the correctness of Bloom filters, counting (Bloom) filters and count-min sketch data structures under adversarial inputs. Their approach required the introduction of a bespoke winning condition for the adversary and only addressed correctness. Our simulation-based approach dispenses with such a winning condition and allows us to consider both correctness and privacy in a single framework, at the cost of slightly worse bounds for their success condition of choice. We will consider their most powerful adversarial setting, where the adversary can access the internal state of the PDS and insert new items after an initial setup.
We provide a more detailed comparison of our correctness results in \cref{sec:correctness-discussion}. 

Naor and Yogev~\cite{_C:NaoYog15,DBLP:journals/talg/NaorY19} analyse adversarial correctness (but not privacy) of Bloom filters in a game-based setting in which the capabilities of the adversary are strictly weaker than ours (the adversary can initialise the filter with a set of items and then make membership queries, but cannot access the internal state of the filter or make further insertions). They explore the relations between secure Bloom filters and one-way functions. They show that pre-processing the inputs to a Bloom filter using a PRP achieves adversarial correctness. This is similar to our function-decomposability result (but limited to Bloom filters and in a weaker attack model). 

Earlier work on using Bloom filters in combination with cryptographic techniques to build searchable encryption (SE) schemes includes~\cite{_EPRINT:Goh03,_EPRINT:BelChe04,_C:BKOS07,DBLP:conf/ccs/RaykovaVBM09,DBLP:journals/tdp/NojimaK09}. 
In particular,~\cite{_EPRINT:Goh03} cleverly combines Bloom filters with PRFs to build secure indexes, 
while
~\cite{DBLP:journals/tdp/NojimaK09} introduces a simulation-based security model for analysing a privacy-enhanced Bloom filter using blind signatures and oblivious PRFs (however, the model in~\cite{DBLP:journals/tdp/NojimaK09} is weaker than ours, since it does not allow any new insertions after initialisation).
SE schemes target outsourced storage of data with added functionality (like searchability); this involves complex two (and multi-) party cryptographic protocols. On the other hand, our work targets making secure versions of PDS that are widely deployed in real-world computing systems. A key difference is that the target in SE schemes is achieving privacy for users against a malicious server, whereas we focus on privacy (and adversarial correctness) of PDS against malicious users. 

Paterson and Raynal~\cite{PatRay22} considered the HyperLogLog (HLL) cardinality estimator~\cite{flajolet2007hyperloglog}, introducing attacks and simulation-based definitions to study the correctness of HLL under adversarial inputs. Our approach is inspired by that of~\cite{PatRay22} but applies to a broad class of AMQ-PDS rather than to a specific cardinality estimator.  
In fact, the functionality of \amq hinders the analysis of their behaviour in adversarial settings; the information revealed by membership queries can be leveraged to make adaptive insertions, resulting in more complicated proofs than in the HLL case.

Prior attack work has focused on Bloom filters~\cite{_DSN:GerKumLau15}, flooding of hash tables~\cite{LiptonN93,_USENIX:CroWal03},  or key collision attacks on hash tables~\cite{VU_903934}. References~\cite{_DSN:GerKumLau15,_CCS:ClaPatShr19} provide partial summaries of prior work on the security of PDS, focused on Bloom filters. Work on privacy properties of Bloom filters can be found in~\cite{_PETS:KKDM11,DBLP:conf/psd/BianchiBL12,DBLP:conf/acsac/GervaisCKG14}. 

A recent line of work~\cite{Ben-EliezerJWY20,DBLP:journals/corr/abs-2011-07471,_C:KMNS21} has studied broad classes of streaming algorithms under adversarial input. The models used are game-based and limit the adversary's capabilities (no corruption of the data structure is permitted). However, they are able to avoid introducing cryptographic assumptions, at the cost of producing schemes that are inefficient in practice. Our approach involves replacing hash functions in existing constructions with keyed PRFs, and so does involve cryptographic building blocks, but remains fully practical. Meanwhile, \cite{HassidimKMMS20} considers adding robustness to streaming algorithms using differential privacy.

\subsection{Paper Organisation}
After preliminaries in \S~\ref{sec:preliminaries}, \S~\ref{sec:ins-only-amq} develops our syntax for AMQ-PDS and formalises the NAI concept that we use to characterise correctness in the non-adversarial setting. Adversarial correctness is addressed in \S~\ref{sec:correctness} and privacy in \S~\ref{sec:privacy}. In \S~\ref{sec:secure-instantiations} we concretely evaluate the security our theorems provide when instantiating specific AMQ-PDS using pseudorandom functions.

\section{Preliminaries}\label{sec:preliminaries}

\paragraph{Notation.}\label{sec:notation} Given an integer $m \in \ZZ_{\ge 1}$, we write $[m]$ to mean the set $\{1,2,\dots,m\}$. We consider all logarithms to be in base 2. If $S$ is a set, we denote by $\mathcal{P}(S)$ the power set of $S$, and $\plist(S)$ the set of all lists with non-repeated elements from $S$. Given a statement $S$, we denote by $[S]$ the function that returns $\top$ if $S$ is true and $\bot$ otherwise. 
Given two sets $\domain$ and $\range$, we define $\Funcs[\domain,\range]$ to be the set of functions from $\domain$ to $\range$. We write $F \getsr \Funcs[\domain,\range]$ to mean that $F$ is a random function $\domain\xrightarrow{F}\range$.
Given a set $\domain$ we define $\Perms[\domain]$ to be the set of permutations over $\domain$. We write $\pi \getsr \Perms[\domain]$ to mean that $\pi$ is a random permutation over $\domain$.
Given a set $S$, we denote the identity function over $S$ as $\identity_S \colon S \rightarrow S$.
If $D$ is a probability distribution, we write $x \getsr D$ to mean that $x$ is sampled according to $D$.
We denote the uniform distribution over a finite set $S$ as $U(S)$.
Given a set $S$ (resp.~a list $L$), we denote by $|S|$ (resp.~$|L|$) the number of elements in $S$ (resp.~$L$).
To indicate a fixed-length list of length $s$ initialised empty, we write $a \gets \bot^s$. We let $\Load{a}$ be the number of set entries of $a$. To insert an entry $x$ into the first unused slot in $a$ we write $a'\,{\gets}\,a \diamond x$ such that $a'\,{=}\,x\,\bot\,\dots\,\bot$ with $s{-}1$ trailing $\bot$s and $\Load{a'} = 1$. A further insertion $a''\,{\gets}\,a' \diamond y$ results in $a''\,{=}\,x\,y\,\bot\,\dots\,\bot$ with $\Load{a''} = 2$, and so on. We refer to the $i$-th entry in a list $a$ as $a[i]$.
In algorithms, we assume that all key-value stores are initialised with value $\bot$ at every index, using the convention that $\bot < n$, $\forall n \in \mathbb{R}$, and we denote it as $\{\}$. Given a key-value store $a$, we refer to the value of the entry with key $k$ as $a[k]$, and we refer to the set of all inserted values as $\Values{a}$. We write variable assignments using $\gets$, unless the value is output by a randomised algorithm, for which we use $\getsr$.

For any randomised algorithm $\textsf{alg}$, we may denote the coins that $\textsf{alg}$ can use as an extra argument $r \in \coinspace$ where $\coinspace$ is the set of possible coins, and write $\text{\textit{output}} \gets \somealgS(\text{input}_1, \text{input}_2, \dots, \text{input}_\ell; r).$
We may also suppress coins whenever it is notationally convenient to do so.
If an algorithm is deterministic, we allow setting $r$ to $\bot$.
We remark that the output of a randomised algorithm can be seen as a random variable over the output space of the algorithm. Unless otherwise specified, we will consider random coins to be sampled uniformly from $\coinspace$, independently from all other inputs and/or state. We may refer to such $r$ as ``freshly sampled''.
We write  $\somealgS^{f_1, \dots f_n}$ whenever $\somealgS$ is given oracle access to functions $f_1$, \dots, $f_n$.

In this work we will consider \amq that can store elements from finite domains $\domain$ by letting $\domain = \cup_{\ell=0}^L \{0,1\}^\ell$ for some large but finite value of $L$, say $L = 2^{64}$. This allows us to give a natural definition of the setting where an \amq is not influenced by an adversary in Def.~\ref{def:nai-state}. Our constructions make use of pseudorandom functions, which we model as truly random functions to which the \amq has oracle access. While a priori multiple random functions may be used by a particular \amq, our results are stated in terms of a single function $F \colon \domain \rightarrow \range$, with $\domain$ being the finite set described above, and $\range$ depending on the \amq algorithms and public parameters. We will see that in practice this will not be a limitation, even in the case of insertion-only Cuckoo filters, originally described as using two hash functions.

\fullversion{
\begin{definition}
	Consider the PRF experiment in Fig.~\ref{fig:PRF}.
	We say a pseudorandom function family $R \colon \Kspace \times \domain \rightarrow \range$ is $(q, t, \varepsilon)$-secure if for all adversaries $\advB$ running in time at most $t$ and making at most $q$ queries to its $\RRO$ oracle in $Exp_R^{PRF}$, we have: 
	\begin{align*}
		\text{Adv}_R^{PRF} (\advB) &\coloneqq \abs{ \pr{b' = 1 | b=0} - \pr{b' = 1 | b=1} } \leq \varepsilon.
	\end{align*}
	We say $\advB$ is a $(q,t)$-PRF adversary.
\end{definition}
\begin{figure}[t]
	\begin{pchstack}[boxed,center,space=1em]
		\procedure[linenumbering,headlinecmd={\vspace{.1em}\hrule\vspace{.3em}}]{$Exp_R^{PRF}(\advB)$}{%
			K \getsr \mathcal{K}; \ F \getsr \textrm{Func}[\domain,\range] \\
			b \getsr \bin;\ b' \getsr \advB^{\textbf{RoR}} \\
			\pcreturn b'
		}
		\procedure[linenumbering,headlinecmd={\vspace{.1em}\hrule\vspace{.3em}}]{Oracle $\RRO(x)$}{%
			\pcif b = 0:\ y \gets R_K(x) \\
			\pcelse:\ y \gets F(x) \\
			\pcreturn y 
		}
	\end{pchstack}
	\caption{The PRF experiment.}
	\label{fig:PRF}
\end{figure}
}{
\begin{definition}
	We say a pseudorandom function family $R \colon \Kspace \times \domain \rightarrow \range$ is $(q, t, \varepsilon)$-secure if any adversary $\advB$ running in time at most $t$ and making at most $q$ queries to either $R_K$ where $K \getsr \mathcal{K}$ or~a random function $F \getsr \textrm{Func}[\domain,\range]$, has distinguishing advantage:
	\begin{align*}
		\text{Adv}_R^{PRF}(\advB) \coloneqq |\Pr[\advB^{R_K(\cdot)} = 1] - \Pr[\advB^{F(\cdot)} = 1]| \leq \varepsilon.
	\end{align*}
	We say $\advB$ is a $(q,t)$-PRF adversary.
\end{definition}
}

\paragraph{Probabilistic Data Structures (PDS)}\label{sec:pds-prelimns} 
Different PDS have been proposed to efficiently provide approximate responses to various kinds of queries. For example, Bloom filters~\cite{bloom1970space} provide answers to approximate membership queries (AMQs); that is, they answer queries of the form $x \overset{\mathrm{?}}{\in} \setS$ for some finite set $\setS$. Count-min sketches~\cite{_JoA:CorMut05} provide approximate frequency estimates for events in a data stream, and HyperLogLog~\cite{flajolet2007hyperloglog} computes the approximate number of distinct elements in a data stream.
For each problem, multiple PDS designs have been proposed, with a tradeoff between their performance and the features they offer. For example, while Bloom filters provide AMQ answers for a set $\setS \subset \domain$ under the assumption that items can only be inserted into $\setS$ but not removed, Cuckoo filters~\cite{_CoNEXT:FAKM14} also allow deletion of items, but at the cost of introducing false negative errors. 

The specific design of a PDS implies certain \emph{consistency rules} that the data structure will satisfy. For example, Bloom filters provide the guarantee of having no false negatives, which means they will never answer ``false'' to a membership query on $x$ when $x \in \setS$. Consistency rules will be important when we prove security properties of PDS; if they are not respected, an adversary in a security game may be able to observe inconsistent behaviour and tell they are interacting with a simulation of the PDS rather than the real PDS.

In this paper, we will study PDS that rely on hash functions to provide good expected behaviour on non-adversarial inputs. In our proofs, we will replace such hash functions with random functions, and eventually replace these with keyed pseudorandom functions. We refer to the latter as \emph{keyed} PDS. We focus on insertion-only \amq, i.e.~those that support insertions and membership queries, but not deletions. 

\section{Insertion-only \amq}\label{sec:ins-only-amq}
We proceed to define the syntax of an insertion-only \amq. We remark that we only consider \amq with deterministic membership checks that do not modify the state of the \amq, since these include the two most popular \amq, namely Bloom and (insertion-only) Cuckoo filters.

We denote public parameters for an \amq by $pp$. 
We denote the state of a PDS $\Pi$ as $\sigma \in \statespace$, where  $\statespace$ denotes the space of possible \emph{states} of $\Pi$.
The set of elements that can be inserted into the \amq is denoted by $\domain$, unless stated otherwise.
We consider a syntax consisting of three algorithms:
\begin{itemize}[leftmargin=1.5em]
	\item The setup algorithm $\sigma \gets \setupS(pp;\, r)$ sets up the initial state of an empty PDS with public parameters $pp$; it will always be called first to initialise the \amq.
	\item The insertion algorithm $(b,\, \sigma') \gets \upS(x,\, \sigma;\, r)$, given an element $x \in \domain$, attempts to insert it into the \amq, and returns a bit $b \in \{\bot,\top\}$ representing whether the insertion was successful ($b = \top)$ or not ($b = \bot$), and the state $\sigma'$ of the \amq after the insertion.
	\item The membership querying algorithm $b \gets \qryS(x,\, \sigma)$, given an element $x \in \domain$, returns a bit $b \in \{\bot,\top\}$ (approximately) answering whether $x$ was previously inserted ($b = \top$) or not ($b = \bot$) into the \amq.
\end{itemize}
We note that $\qryS$ does not change the state of the \amq, and hence does not output a new $\sigma'$ value. 

In the case of Bloom and Cuckoo filters, $\qryS$ calls may return a false positive result with a certain probability. That is, we may have $\top\,{\gets}\,\qryS(x,\, \sigma)$ even though no call $\upS(x,\, \sigma';\, r)$ was made post setup and prior to the membership query.
We refer to the probability $\Pr[\top\,{\gets}\,\qryS(x, \sigma) \mid x\ \text{was not inserted into}\ \Pi]$ as the \emph{false positive probability} 
of an \amq $\Pi$.
This probability depends on $\sigma$, which could be generated as the result of a sequence of adversarially chosen insertions, or which could just arise through insertions made by honest users.

\paragraph{AMQ-PDS in the honest setting.}
To argue about how the state of an \amq under adversarial queries may deviate from expected in a ``honest'' or ``non-adversarial'' setting, we first define what we expect to see in an honest setting.
In general, non-adversarial input distributions are application-dependent. For example, honest users may sample inputs from distributions with different entropy in different applications. Similarly, different applications may imply that domain elements may be likely to be inserted multiple times, or just once. 
This suggests that proving results for arbitrary AMQ-PDS may require defining the distributions produced by non-adversarial actors, in order to quantify the expected performance of the data structure under honest inputs.

We overcome this issue by noticing two properties that an AMQ-PDS can satisfy and that suffice to argue about their performance under non-adversarial inputs in an application-independent manner: function-decomposability and reinsertion invariance.

\begin{definition}[Function-decomposability]\label{def:f-decomposability}
	Let $\Pi$ be an insertion-only $\amq$ and let $F \getsr \Funcs[\domain,\range]$ with $\range \subset \domain$ be a random function to which $\Pi$ has oracle access. Let $\identity_\range$ be the identity function over $\range$. We say that $\Pi$ is $F$-decomposable if we can write 
		\begin{align*}
			\upS^F(x,\, \sigma;\, r) &= \upS^{\identity_\range}(F(x),\, \sigma;\, r)\quad \forall x \in \domain,\, \sigma \in \Sigma,\, r \in \coinspace, \\
			\qryS^F(x,\, \sigma) &= \qryS^{\identity_\range}(F(x),\, \sigma)\quad \forall x \in \domain,\, \sigma \in \Sigma,
		\end{align*}
		where $\upS^{\identity_\range}$ and $\qryS^{\identity_\range}$ cannot internally evaluate $F$ due to not having oracle access to it and $F$ being truly random.
	
	Function-decomposability also applies to \amq with oracle access to multiple functions. For example, if $\Pi$ had oracle access to $t$ functions $F_1, \dots, F_t$ and was $F_1$-decomposable, we would write 
		\begin{align*}
			\upS^{F_1, \dots, F_t}(x,\, \sigma;\, r) &= \upS^{\identity_\range, F_2, \dots, F_t}(F_1(x),\, \sigma;\, r)\ \forall x, \sigma, r, \\
			\qryS^{F_1, \dots, F_t}(x,\, \sigma) &= \qryS^{\identity_\range,F_2, \dots, F_t}(F_1(x),\, \sigma)\ \forall x, \sigma.
	\end{align*}
\end{definition}

Function-decomposability has the effect of ``erasing'' any structure on the input domain $\domain$. Essentially, function-decomposable \amq replace any input element $x \in \domain$ with a fixed element $y \in \range$ sampled uniformly at random, and proceed to do any further processing on $y$. This allows us to disregard the input distribution on $\domain$ and instead think of input elements sampled uniformly at random from $\range$.

In \S~\ref{sec:correctness} and \S~\ref{sec:privacy}, we will prove theoretical guarantees for \amq instantiated using pseudorandom functions (PRFs). After a switch from PRFs to truly random functions, we will be able to assume function-decomposability.

\begin{definition}[Reinsertion invariance]\label{def:reinsertion-invariance}
	Consider an \amq $\Pi$. We say $\Pi$ is \emph{reinsertion invariant} if for all $x \in \domain,\ \sigma \in \statespace$ such that $\top \gets \qryS(x,\,\sigma)$, we have $(b, \sigma') \gets \upS(x, \sigma; r) \implies \sigma = \sigma'\ \forall r \in \coinspace.$ Informally, if $x$ appears to have been inserted, then further insertions of $x$ will not cause the state $\sigma$ of $\Pi$ to change.
\end{definition}

This property is shared by Bloom and insertion-only Cuckoo filters, and it appears natural since insertion-only \amq aim to represent sets and not multisets; hence repeated insertions need not change the state.

The state of a function-decomposable, reinsertion invariant AMQ-PDS 
is not affected by any structure on the input elements sampled from $\domain$, nor by elements being reinserted more than once. Rather, it only depends on the number of distinct elements in the data structure. This allows us to define the notion of an $(n, \varepsilon)$-non-adversarially-influenced state as follows.

\begin{definition}[$(n, \varepsilon)$-NAI]\label{def:nai-state}
	Let $\varepsilon > 0$, and let $n$ be a non-negative integer.
	Let $\Pi$ be an \amq with public parameters $pp$ and state space $\Sigma$, such that its $\upS$ algorithm makes use of oracle access to functions $F_1,\dots,F_t$.
	Let $\somealgS$ be a randomised algorithm outputting values in $\statespace$.
	Let $\sigma$ and $\sigma^{(n)}$ be random variables representing respectively the outputs of $\somealgS$ and of the randomised algorithm $n$-$\naigen$ described in Fig.~\ref{fig:NAI-gen}.
	We say that $\somealgS$ outputs an $(n, \varepsilon)$-non-adversarially-influenced state (denoted by $(n, \varepsilon)$-NAI)
	if $\sigma$ is $\varepsilon$-statistically close to $\sigma^{(n)}$.
	\begin{figure}[t]
		\centering
		\begin{pcvstack}[boxed]
			\procedure[linenumbering,headlinecmd={\vspace{.1em}\hrule\vspace{.3em}}]{$n$-$\naigen^{F_1,\dots,F_t}(pp)$}{
				\sigma^{(0)} \getsr \setupS(pp)\\
				{[}x_1, \dots, x_n{]} \getsr \{S \in \plist(\domain) \mid |S| = n\}\label{line:sample-in-NAI-gen}\\
				\pcfor j=1, \dots, n: \t (b, \sigma^{(j)}) \getsr \upS^{F_1,\dots,F_t}(x_j, \sigma^{(j-1)})\label{line:up-in-NAI-gen}\\
				\pcreturn \sigma^{(n)}
			}
		\end{pcvstack}
		\caption{Algorithm returning non-adversarially-influenced (NAI) state.}\label{fig:NAI-gen}
	\end{figure}
\end{definition}

In Def.~\ref{def:nai-state}, the $n$-$\naigen$ algorithm imitates the behaviour of an honest user inserting distinct elements into the PDS.\footnote{Note that in the case of Cuckoo filters, not all insertions may succeed; see discussion later in this section.} The distribution output by $n$-$\naigen$ then becomes the benchmark for how close to honestly-generated (or non-adversarially-influenced, NAI) the state of an \amq is. 
Now we are ready to introduce the NAI false positive probability for an \amq.

\begin{definition}[NAI false positive probability]\label{def:nai-fpp}
	Let $\Pi$ be a function-decomposable reinsertion invariant \amq with public parameters $pp$, using functions $F_1$,\dots,$F_t$ sampled respectively from distributions $D_{F_1}$,\dots,$D_{F_t}$ to instantiate its functionality. 
	Let $n$ be a non-negative integer.
	Define the \emph{NAI false positive probability after $n$ distinct insertions} as
	\begin{align*}
		\pfp{n}{\coloneqq}\Pr\left[\begin{matrix}F_1{\getsr}D_{F_1}, {\dots}, F_t{\getsr}D_{F_t},\\ \sigma \getsr \text{$n$-$\naigen(pp)$},\\x \getsr \domain\setminus V\\ \end{matrix}\hspace{-.3em}: \top{\gets}\qry^{F_1, \dots, F_t}(x,\sigma)\right],
	\end{align*}
	where $V$ is the list $[x_1,\dots,x_n]$ sampled on line~\ref{line:sample-in-NAI-gen} of $n$-$\naigen(pp)$.
\end{definition}
\noindent
Definition \ref{def:nai-fpp} captures the probability that a non-adversarial user experiences a false positive membership query result after inserting $n$ distinct elements into $\Pi$.

\paragraph{Bloom filters.}
The most popular \amq is the Bloom filter~\cite{bloom1970space}. Informally, a Bloom filter consists of a bitstring $\sigma$ of length $m$ initially set to $0^m$, and a family of $k$ independent hash functions $H_i: \bits^* \rightarrow [m]$, for $i\in[k]$. An element $x$ is inserted into the filter by setting bit $H_i(x)$ of $\sigma$ to $1$, for every $i \in [k]$. A membership query on $x$ is carried out by checking if the $k$ bits $H_i(x), i \in [k]$ are set to $1$ in $\sigma$. Bloom filters have no false negatives, i.e.\ if $x$ has been inserted then a membership query on $x$ always returns $\top$. However, Bloom filters can have false positives, i.e.\ a membership query on $x$ can return  $\top$ even when $x$ has not been inserted, 
due to the potential for collisions in the hash functions $H_i$. 

In this work, we bundle the $k$ hash functions $H_i$ into a single function $F \colon \domain \rightarrow \range = [m]^k$. We will later instantiate $F$ with a pseudorandom function, rather than a fixed hash function, to achieve our security notions. This bundling is convenient for making our formal description of a Bloom filter (that follows) fit with our general \amq syntax.
We now formally define Bloom filters.

\begin{definition}\label{def:bitmap}
	Let $B_{m,k} : [m]^k \rightarrow \bits^{m}$ be the map that on input $\vec{x} \in [m]^k$ returns an $m$-bit ``bitmap'' where all bits are zero except those at the indices $x_i$, $i \in [k]$. 
\end{definition}

\begin{definition}\label{def:bloom-filter} 
Let $m,\,k$ be positive integers. We define an $(m,k)$-Bloom filter to be the \amq with algorithms defined  in Fig.~\ref{fig:bloom-filter-api}, with $pp = (m,k)$, and $F \colon \domain \rightarrow \range \equiv [m]^k$. 
\begin{figure}[t]
	\Wider[4em]{
	\centering
	\begin{pcvstack}[boxed,space=0.5em]
		\begin{pchstack}
		\procedure[linenumbering,headlinecmd={\vspace{.1em}\hrule\vspace{.3em}}]{$\setupS(pp)$}{%
			m,\,k \gets pp;\ \sigma \gets 0^m\\
			\pcreturn \sigma
		}
		\procedure[linenumbering,headlinecmd={\vspace{.1em}\hrule\vspace{.3em}}]{$\upS^F(x, \sigma)$}{%
			\sigma' \gets \sigma \lor B_{m,k}(F(x))\\
			\pcreturn \top, \sigma'
		}
		\procedure[linenumbering,headlinecmd={\vspace{.1em}\hrule\vspace{.3em}}]{${\qryS^F}(x, \sigma)$}{%
			b \gets B_{m,k}(F(x)) \\
			\pcreturn [b = \sigma \wedge b]
		}
		\end{pchstack}	
	\end{pcvstack}
	}
	\caption{\amq syntax instantiation for the Bloom filter.}
	\label{fig:bloom-filter-api}
\end{figure}
\end{definition}

We now recall from the literature an estimate for and an upper bound on the NAI false positive probability for Bloom filters.
\begin{lemma}(\cite{bloom1970space},\cite[Theorem 4.3]{_SIGMETRICS:GoeGup10})
\label{lem:bf-false-positive-probability} 
Let $\Pi$ be an $(m,k)$-Bloom filter using a random function $F \colon \domain \rightarrow [m]^k$. Define $\opfp{n} \coloneqq \left(1 - e^{-\frac{(n+0.5)k}{m-1}}\right)^k$.
Then for any $n$, 1) \(\pfp{n} \xrightarrow{\: m \to \infty \: } (1 - e^{-nk/m})^k\), and 2) \(\opfp{n} \ge \pfp{n}.\)
\end{lemma}

\paragraph{Insertion-only Cuckoo filters.}
Proposed as an improvement over Bloom filters, Cuckoo filters~\cite{_CoNEXT:FAKM14} allow deletion of elements, at the cost of potentially introducing false negatives if a user attempts to delete an element that was not previously inserted.
In ~\cite[\S3.1]{_CoNEXT:FAKM14}, the authors also consider the case of \emph{insertion-only Cuckoo filters}, a variant that does not support deletions.
In their original description, insertion-only Cuckoo filters use two hash functions $H_I\colon \domain \rightarrow \{0,1\}^{\lambda_I}$ and $H_T\colon \domain \rightarrow \{0,1\}^{\lambda_T}$. Let $\sigma$ consist of a collection $(\sigma_i)_i$ of $m = 2^{\lambda_I}$ fixed-length lists, or ``buckets'', indexed by $i \in [m]$ and each containing $s$ slots, together with a stash $\sigma_{evic}$ containing one more slot.\footnote{We note that while no restrictions on $m$ are mentioned explicitly in~\cite{_CoNEXT:FAKM14}, it was recently pointed out~\cite{_CoNEXT:BosTraSta20} that $m$ should be a power of two in order to avoid the potential introduction of false negative results.} 
A detailed description of the internals of Cuckoo filters can be found in\fullversion{~\cref{app:cf-algorithms}}{ the full version}.
In contrast to Bloom filters, two hash functions $H_T$ and $H_I$ are used at different points in the $\upS$ and $\qryS$ algorithms of Cuckoo filters. Hence, instead of bundling them into a single function $F$, we will give $\upS$ and $\qryS$ access to both functions.
In \S~\ref{sec:cf-correctness-variants}, we will discuss different approaches to replacing these with PRFs in order to achieve our security notions.

\begin{definition}\label{def:cuckoo-filter} Let $pp = (s,\,\lambda_I,\,\lambda_T,\,num)$ be a tuple of positive integers. We define an $(s,\lambda_I,\lambda_T,num)$-Cuckoo filter to be the \amq with algorithms defined in \fullversion{Fig.~\ref{fig:cuckoo-filter-api} of~\cref{app:cf-algorithms}}{the full version}, making use of hash functions $H_T {\colon}\,\domain\,{\rightarrow}\,\{0,1\}^{\lambda_T}$ and $H_I {\colon}\,\domain\,{\rightarrow}\,\{0,1\}^{\lambda_I}$.
\end{definition}

In the original definition of~\cite{_CoNEXT:FAKM14}, the false positive probability of a Cuckoo filter with all its buckets full is computed assuming $H_T$ is a random function, and is given by $1 - (1 - {2^{-\lambda_T}})^{2s+1}$. In~\cite{_CoNEXT:FAKM14}, this value is shown to be an upper bound on the false positive probability of a Cuckoo filter containing $n \le m \cdot s$ elements.

\begin{lemma}(\cite{_CoNEXT:FAKM14})\label{lem:cf-false-positive-probability-upperbound}
	Let $\Pi$ be a $(s,\lambda_I,\lambda_T,num)$-Cuckoo filter and let $H_T\colon \domain \rightarrow \{0,1\}^{\lambda_T}$ be a random function. For any non-negative integer $n$,
	${\opfp{n}} \coloneqq 1 - \left(1 - {2^{-\lambda_T}}\right)^{2s+1} \ge {\pfpsym}(FP \mid n).$
\end{lemma}

While Bloom and Cuckoo filters instantiate the insertion-only \amq syntax described above, they also satisfy some additional properties that we refer to as ``consistency rules'', captured below.

\begin{definition}[Insertion-only \amq consistency rules]\label{def:bf-consistency-rules}
	Consider an \amq $\Pi$. We say $\Pi$ has:
	\begin{itemize}[leftmargin=1.0em]
		\item \emph{Element permanence} if for all $x \in \domain,\ \sigma \in \statespace$ such that $\top \gets \qryS(x,\,\sigma),$ and for any sequence of insertions resulting in a later state $\sigma'$, \quad $b \gets \qryS(x, \sigma') \implies b = \top.$
		\item \emph{Permanent disabling} if given $\sigma \in \statespace$ such that there exists $x \in \domain,\, r \in \coinspace$ where $(b,\,\overline{\sigma}) \gets \upS(x, \sigma; r)$ and $b = \bot$, then $\overline{\sigma} = \sigma$ and for any $x' \in \domain$, $r' \in \coinspace$, \quad $(b', \sigma') \gets \upS(x', \sigma; r') \Rightarrow b' = \bot \text{ and } \sigma' = \sigma.$
		\item \emph{Non-decreasing membership probability} if for all $\sigma \in \statespace$, $x,\, y \in \domain$, $r \in \coinspace$, \quad $(b, \sigma') \gets \upS(x, \sigma; r) \Rightarrow \Pr[\top \gets \qryS(y, \sigma)] \le \Pr[\top \gets \qryS(y, \sigma')].$
	\end{itemize}
\end{definition}

\section{Adversarial Correctness}\label{sec:correctness}

In this section, we develop simulation-based security definitions with which we then analyse the adversarial correctness of \amq. We derive bounds on the correctness of insertion-only \amq that are function-decomposable, reinsertion invariant, and obey the consistency rules in Def.~\ref{def:bf-consistency-rules}. Finally, we apply our results to provide correctness guarantees for PRF-instantiated Bloom filters and a straightforward variant of insertion-only Cuckoo filters. Amongst other things, these guarantees limit an adversary's ability to carry out pollution attacks~\cite{_DSN:GerKumLau15} and target-set coverage attacks~\cite{_CCS:ClaPatShr19} on Bloom filters. Both attacks involve an adversary manipulating the false positive probability of a Bloom filter, and our analysis shows that this is not possible (up to some security bounds that we make concrete in \S~\ref{sec:secure-instantiations}).

\paragraph{Settings.}\label{sec:deployment-settings}

Our model  (see Fig.~\ref{game:r-o-i}) considers an adversary $\advA$ interacting with an \amq $\Pi$ in two stages. In the first stage, the data structure is initialised empty, and $\advA$ provides a finite set of elements to insert into it through the $\repO$ oracle, which can be called only once. 
In the second stage, the adversary is given access to three other oracles: $\qryO$, responding to queries of the form ``has element $x$ been inserted into $\Pi$?'', $\upO$, inserting an element provided by $\advA$ into $\Pi$, and $\revO$, returning $\Pi$'s current state.

While both the $\repO$ and $\upO$ oracles allow insertion of elements, defining these as separate oracles allows us to treat initialisation and subsequent insertions/queries as distinct stages. Then, by disabling access to the $\upO$ oracle, we can define an \emph{immutable} setting in which no insertions are allowed after initialisation, as in \cite{_CCS:ClaPatShr19}. Further, this separation of $\repO$ and $\upO$ oracles will come in useful in our treatment of privacy in \S~\ref{sec:privacy}.

\subsection{Notions of Correctness}\label{sec:correctness-notions}

Our analysis of the adversarial correctness of \amq uses a simulation-based approach.
We start with a high-level explanation of our approach in order to provide some intuition. 

In our security framework, the adversary $\advA$ plays in either the real or ideal world. In the real world, it interacts with a keyed \amq $\Pi$, where it has access to oracles that allow it to insert elements into the data structure, as well as to make membership queries for elements of its choice.
In the ideal world, it interacts with a simulator $\simS$, constructed so as to provide an NAI view of $\Pi$ to $\advA$.
At the end of its execution, $\advA$ produces some output, which is given to a distinguisher $\dD$. $\advA$'s output is arbitrary -- for example, it could be the state of the AMQ-PDS obtained by making an appropriate oracle query. $\dD$'s task is to compute which world $\advA$ was operating in, based on its output. 
By bounding $\dD$'s ability to distinguish between the ideal and real worlds, we can quantify how much more harm $\advA$ can do in the real world (where it can make adaptive insertions and membership queries) compared to the ideal world (where everything is handled by $\simS$ in an NAI manner). 

We begin by defining the $\gameRoI$ game in Fig.~\ref{game:r-o-i}. We will use $\Real$ and $\Ideal$ to denote the real $(d = 0)$ and ideal $(d = 1)$ versions of $\gameRoI$, respectively. The game's output is the bit $d'$ generated by a distinguisher $\dD$ operating on $\advA$'s output (which is an arbitrary string whose length is incorporated into the running time of $\advA$).
Throughout the paper, if an oracle $\textbf{O}$ is not directly specified, we assume it is defined as in Fig.~\ref{game:r-o-i}.

\begin{definition}
Let $\Pi$ be an insertion-only \amq, with public parameters $pp$, and let $R_K$ be a keyed function family. 
We say $\Pi$ is $(q_u, q_t, q_v,t_a, t_s, t_d, \varepsilon)$-adversarially correct if, 
for all adversaries $\advA$ running in time at most $t_a$ and making at most a single query to $\repO$ and $q_u, q_t, q_v$ queries to oracles $\upO,\qryO,\revO$ respectively in the $\gameRoI$ game (Fig.~\ref{game:r-o-i}) with a simulator $\simS$ \emph{that provides an NAI view of $\Pi$ to $\advA$} and runs in time at most $t_s$,
and for all distinguishers $\dD$ running in time at most $t_d$, we have: 
\begin{align*}
	\text{Adv}_{\Pi, \mathcal{A, S}}^{RoI} (\dD){\coloneqq} \big| \, \pr{\Real(\advA, \dD){=}1}{-}\pr{\Ideal(\advA, \dD,\simS){=}1}\big| {\le} \varepsilon.
\end{align*}
\end{definition}
\begin{figure}[t]
	\Wider[4em]{
		\centering
		\begin{pchstack}[boxed,space=.0em]
			\begin{pcvstack}
				\procedure[linenumbering,space=1em,headlinecmd={\vspace{.1em}\hrule\vspace{.3em}}]{\gameRoI$(\advA, \simS, \dD, pp)$}{%
					d \getsr \{0, 1\} \\ 
					\pcif d = 0 \t \pcmycomment{\Real} \\
					\ \ K \getsr \mathcal{K}; F \gets R_K \\
					\ \ \tinit \gets \bot \\
					\ \ \sigma \gets \setupS(pp)\\
					\ \ out \getsr \advA^{\repO, \upO, \qryO, \revO} \\ 
					\pcelse \t \pcmycomment{\Ideal} \\
					\ \ out \getsr \simS(\advA, pp) \\
					\pcreturn d' \getsr \dD(out)
				}
			\end{pcvstack}
			\begin{pcvstack}[space=0.05em]
				\procedure[linenumbering,headlinecmd={\vspace{.1em}\hrule\vspace{.3em}}]{Oracle $\repO(\setV)$}{%
					\pcif \tinit = \top: \pcreturn \bot \\
					\tinit \gets \top \\
					\pcfor x \in V \\
					\ \ (b, \sigma) \getsr \upS^F(x, \sigma)\\
					\pcreturn \top
				}
				\procedure[linenumbering,headlinecmd={\vspace{.1em}\hrule\vspace{.3em}}]{Oracle $\upO(x)$}{%
					\pcif \tinit = \bot: \pcreturn \bot \\
					(b, \sigma) \getsr \upS^F(x, \sigma)\\
					\pcreturn b
				}
			\end{pcvstack}
			\begin{pcvstack}
				\procedure[linenumbering,headlinecmd={\vspace{.1em}\hrule\vspace{.3em}}]{Oracle $\qryO(x)$}{%
					\pcif \tinit = \bot\\
					\ \ \pcreturn \bot \\
					\pcreturn \qryS^F(x, \sigma)
				}
				\procedure[linenumbering,space=1em,headlinecmd={\vspace{.1em}\hrule\vspace{.3em}}]{Oracle $\revO()$}{%
					\pcreturn \sigma
				}
			\end{pcvstack}
		\end{pchstack}
	}
	\caption{Correctness game for \amq $\Pi$.}
	\label{game:r-o-i}
\end{figure}

\begin{remark}
	While we explicitly only cover the case where the adversary calls the $\repO$ oracle once, a hybrid argument could be used to derive a bound for the case with $q_r > 1$ $\repO$ queries, as long as the function $F$ is resampled between $\repO$ calls. This would come at the cost of introducing a factor $q_r$ to our bounds.
\end{remark}

\begin{remark}\label{rem:whyAdvCorrect}
We explain why our definition captures adversarial correctness. Consider an arbitrary adversary $\advA$ that in the course of its execution makes $\upO(\cdot)$ queries on adversarially selected inputs $x_1,\ldots,x_n$. These are (potentially) interspersed with other types of query permitted to $\advA$. 
Consider an extension of $\advA$, named $\advA^*$, that behaves exactly as $\advA$ does, but which makes a final membership query $\qryO(x)$ with $x \getsr \domain\setminus \{x_1,\ldots,x_n\}$. 
Suppose the output of $\advA^*$ is the result of that final query (a binary value), and $\dD$'s output is identical to that of $\advA^*$. 
Then it is easy to see that $\pr{\Real(\advA^*, \dD)}$ is exactly the \emph{adversarial} false positive probability of $\Pi$ produced by $\advA$, while $\pr{\Ideal(\advA^*, \dD,\simS)}$ is the NAI false positive probability. The definition, if satisfied, then says that these two probabilities must be within $\varepsilon$ of each other. Even if $\varepsilon$ cannot be shown to be very small for some specific AMQ-PDS, we may still obtain a useful result about adversarial false positive probability in practice.
The above argument involves an arbitrary $\advA$ and a specific $\dD$. The reader may imagine that other choices of $(\advA,\dD)$ may capture additional correctness properties. See \S~\ref{sec:correctness-discussion} for further discussion.
 \end{remark}
 
The details of how the simulator is constructed (and how to bound the distinguishing advantage) depend on the data structure under consideration. Recall that we only consider \amq that support insertions and membership queries, but not deletions.
In Fig.~\ref{fig:correctness-simulator}, we give a simulator $\simS$ that replicates the behaviour of \amq that satisfy the consistency rules from Def.~\ref{def:bf-consistency-rules}, are function-decomposable (see Def.~\ref{def:f-decomposability}) and reinsertion invariant (see Def.~\ref{def:reinsertion-invariance}).
By inspection, the runtime of $\simS$ is not significantly higher than that of the underlying \amq.

\begin{figure}[t]
	\Wider[30em]{
	\centering
	\begin{pcvstack}[boxed,space=.2em]
		\begin{pchstack}[space=.0em]
				\begin{pcvstack}[space=.2em]
					\procedure[linenumbering,space=auto,headlinecmd={\vspace{.1em}\hrule\vspace{.3em}}]{Simulator $\simS(\advA, pp)$}{%
						F \getsr \Funcs[\domain,\range] \label{line:F-rnd-ideal} \\
						\tinit \gets \bot;\ \tupenabled \gets \top \\
						\sigma \gets \setupS(pp) \\
						\tinserted,\, \tFPlist,\, \tCALQ \gets \{\},\, \{\},\, \{\} \\ 
						i \gets 0\ \pcmycomment{Qry counter} \\ 
						ctr \gets 0\ \pcmycomment{Distinct insertions} \\ 
						out \getsr \advA^{\textbf{RepSim, UpSim, QrySim, RevealSim}} \\
						\pcreturn out 
					}
					\procedure[linenumbering,headlinecmd={\vspace{.1em}\hrule\vspace{.3em}}]{Oracle $\repsimO(\setV)$}{%
						\pcif \tinit = \top: \pcreturn \bot \\
						\tinit \gets \top \\
						\pcfor x \in V: b \gets \upsimO(x) \\
						\pcreturn \top
					}
				\end{pcvstack}
				\begin{pcvstack}
					\procedure[linenumbering,space=1em,headlinecmd={\vspace{.1em}\hrule\vspace{.3em}}]{Oracle $\upsimO(x)$}{%
						\pcif \tinit = \bot: \pcreturn \bot \\
						\pcif \tinserted[x] = \bot\\
						\ \ (b, \sigma) \getsr \upS^F(x, \sigma) \\
						\ \ \tupenabled \gets b \\
						\ \ \pcif b = \top \\
						\ \ \t \tinserted[x] \gets \top \\
						\ \ \t ctr \gets ctr + 1 \\ 
						\ \ \pcreturn b \\
						\pcelse \\
						\ \ \pcreturn \tupenabled \label{line:reinsertion-invariant-return}
					}
						\procedure[linenumbering,space=1em,headlinecmd={\vspace{.1em}\hrule\vspace{.3em}}]{Oracle $\textbf{RevealSim}()$}{%
						\pcreturn \sigma
					}
				\end{pcvstack}
			\end{pchstack}
		\procedure[linenumbering,headlinecmd={\vspace{.1em}\hrule\vspace{.3em}}]{Oracle $\qrysimO(x)$}{%
			\pcif \tinit = \bot: \pcreturn \bot \\
			i \gets i + 1 \\ 
			\pcmycomment{Element was inserted or determined a false positive} \\
			\pcif \tinserted[x] = \top \text{ or } \tFPlist[x] = \top \label{line:elem-permanence-1}\\
			\t \pcreturn \top \label{line:elem-permanence-2}\\
			\pcmycomment{Element was not inserted and not false positive} \\
			\pcif \tCALQ[x] = ctr\ \pcmycomment{If no changes since last query of x} \label{line:qry-determinism-1}\\
			\t \pcreturn \bot \label{line:qry-determinism-2}\\
			\pcmycomment{Response needs to be (re)computed} \\
			\tCALQ[x] \gets ctr \\ 
			a_i^{\Ideal} \getsr \qry^{\identity_\range}(Y \getsr \range, \sigma) \label{line:bernoulli} \\ 
			a_i^{G^*} \gets \qryS^F(x, \sigma)\label{line:aG} \label{line:a-G*} \\
			\label{line:corr-E*} a \gets a_i^{\Ideal}\ \pcbox{a \gets a_i^{G^*}}\ \pcmycomment{\Ideal}\ \textcolor{Gray}{\pcbox{G^*}} \\
			\pcif a = \top: 
			\tFPlist[x] \gets \top \\
			\pcreturn a
		}
	\end{pcvstack}
	}
	\caption{Simulator $\simS$ used in Theorem~\ref{thm:general-correctness}. Lines \ref{line:a-G*}-\ref{line:corr-E*} corresponding to intermediate game $G^*$ are used in our proof.}
	\label{fig:correctness-simulator}
\end{figure}

We now proceed to state and prove our correctness theorem.
\begin{restatable}{theorem}{correctnessthm}\label{thm:general-correctness}
Let $q_u,\, q_t, \, q_v$ be non-negative integers, and let $t_a,\, t_d > 0$. 
Let $F\colon\domain\rightarrow\range$.
Let $\Pi$ be an insertion-only \amq with public parameters $pp$ and oracle access to $F$, such that $\Pi$ satisfies the consistency rules from Def.~\ref{def:bf-consistency-rules}, $F$-decomposability (Def.~\ref{def:f-decomposability}) and reinsertion invariance (Def.~\ref{def:reinsertion-invariance}). 
Let $n$ be the number of elements provided by $\advA$ for initial insertion into $\Pi$ by a query call to $\repO$.\footnote{Note there is no guarantee that all $n$ elements are successfully inserted.} 
Let $\alpha$ (resp. $\beta$) be the number of calls to $F$ required to insert (resp. query) an element in $\Pi$ using its $\upS$ (resp.~$\qryS$) algorithm.

If $R_K \colon \domain \rightarrow \range$ is an $(\alpha (n + q_u) + \beta q_t, t_a + t_d, \varepsilon)$-secure pseudorandom function with key $K \getsr \mathcal{K}$, then $\Pi$ is $(q_u, q_t, q_v, t_a, t_s, t_d, \varepsilon')$-adversarially correct with respect to the simulator in Fig.~\ref{fig:correctness-simulator}, where $\varepsilon' = \varepsilon + 2 q_t \cdot \pfp{n+q_u}$ and $t_s \approx t_a$.
\end{restatable}
\begin{proof}[sketch]
We start by defining an intermediate game $G$ that replaces the PRF in $\Real$ with a random function. We then bound the closeness of $\Real$ and $G$ in terms of the PRF advantage $\varepsilon$. To bound the distance between $G$ and $\Ideal$, we construct a game $G^*$ (Fig.~\ref{fig:correctness-simulator}) that looks identical to $G$, and show that $G^*$ and $\Ideal$ are equal up until the ``bad" event that $a_i^{\Ideal} \neq a_i^{G^*}$ for some $i \in [q_t]$. Then, we show that our simulator constructs an NAI view of $\Pi$ in $\Ideal$. Finally, we upper bound the probability of the bad event to obtain our result. The full proof is given in \fullversion{Appendix \ref{app:correctness-proofs}}{the full version}. 
\end{proof}

While Theorem \ref{thm:general-correctness} only refers to a single oracle function $F$ for notational simplicity, the same result holds also for \amqs using $t$ oracle functions $F_1, \dots F_t$ and being $F_1$-decomposable.
This requires adding sampling of the functions $F_1, \dots F_t$ from distributions $D_{F_1}, \dots D_{F_t}$, given by the specification of the \amq,
at the beginning of $\gameRoI$ (Fig.~\ref{game:r-o-i}), and either allowing oracle access to $F_2,\dots,F_t$ to $\simS$ or sampling them also at the beginning of the simulator  (Fig.~\ref{fig:correctness-simulator}).
Then one would replace all calls to $\up^{F}, \qry^{F}, \qryS^{\identity_\range}$ with $\up^{F_1, \dots F_t}$,$ \qry^{F_1, \dots F_t}$, $\qryS^{\identity_\range,F_2, \dots F_t}$.
We stress that the proof would still only incur into the PRF-switching cost of $F_1$, since only $F_1$-decomposability is used. 
However, every $D_{F_1}, \dots D_{F_t}$ would appear in the definition of $\pfp{n + q_u}$ and as such may directly influence the NAI false positive probability.

\subsection{Guarantees for Bloom and Cuckoo filters}\label{sec:correctness-examples}

Our goal is to give bounds on the adversarial correctness of Bloom and Cuckoo filters. 
We will prove that Bloom and a straightforward variant of insertion-only Cuckoo filters satisfy the consistency rules from Def.~\ref{def:bf-consistency-rules}, function-decomposability and reinsertion invariance.
This in turn allows us to use Theorem \ref{thm:general-correctness}  to provide concrete correctness guarantees.

\paragraph{Bloom filters.}\label{seubsec:example:BF}

We start by proving the function-decomposability of Bloom filters.

\begin{restatable}{lemma}{bfareconsistent}
	Bloom filters with oracle access to a random function $F$ are $F$-decomposable, reinsertion invariant, and satisfy the insertion-only \amq consistency rules from Def.~\ref{def:bf-consistency-rules}.
	\label{l:f-dec:BF}
\end{restatable}
\begin{proof} 
	Observe that $F$ is only used on the inputs to the $\upS$ and $\qryS$ algorithms in Fig.~\ref{fig:bloom-filter-api}.
	By identifying $\range = [m]^k$ with a subset of $\domain$, so that formally $\range \subset \domain$, and 
	since $\identity_\range(F(x)) = F(x)$ for any $x \in \domain$, the result follows.
	\ACM{}{\qed}
\end{proof}

\noindent We then apply Theorem~\ref{thm:general-correctness} to Bloom filters instantiated using PRFs.

\begin{corollary}
Let $n, q_u, q_t, q_v$ be non-negative integers, and let $t_a, t_d > 0$. Let $F: \domain \to \range$.
Let $\Pi$  be a Bloom filter with public parameters $pp$ and oracle access to $F$. 
If $R_K$ for $K \getsr \spaceK$ is an $(n+q_u+q_t, t_a + t_d, \varepsilon)$-secure pseudorandom function and $F = R_K$, then $\Pi$ is $(q_u, q_t, q_v, t_a, t_s, t_d, \varepsilon')$-adversarially correct, where $\varepsilon' = \varepsilon + 2 q_t \cdot \pfp{n+q_u}$ and $t_s \approx t_a$.
\end{corollary}
\begin{proof}
From the instantiation of Bloom filters given in Fig.~\ref{fig:bloom-filter-api}, we observe that each $\upS$ and $\qryS$ call contains one call to the function $F$. Then, using Lemma~\ref{l:f-dec:BF}, Theorem~\ref{thm:general-correctness} holds with $\alpha = \beta = 1$. \ACM{}{\qed}
\end{proof}

\paragraph{Insertion-only Cuckoo filters.}\label{sec:cf-correctness-variants}
\begin{figure}[t]
	\centering
	\begin{subfigure}{12em}
		\centering
		\begin{pchstack}[boxed]
			\procedure[linenumbering,headlinecmd={\vspace{.1em}\hrule\vspace{.3em}}]{$\upS^{H_T,H_I}(x, \sigma)$}{%
				tag \gets H_T(x) \\
				i_1 \gets H_I(x) \\
				i_2 \gets i_1 \xor H_I(tag) \\
				\dots
			}
		\end{pchstack}
		\caption{Original variant.~\cite{_CoNEXT:FAKM14}}\label{fig:up-cuckoo-not-f-decomposable}
	\end{subfigure}
	\begin{subfigure}{12em}
		\centering
		\begin{pchstack}[boxed]
			\procedure[linenumbering,headlinecmd={\vspace{.1em}\hrule\vspace{.3em}}]{$\upS^{F,H_T,H_I}(x, \sigma)$}{%
				y \gets F(x) \\
				tag \gets H_T(y) \\
				i_1 \gets H_I(y) \\
				i_2 \gets i_1 \xor H_I(tag) \\
				\dots
			}
		\end{pchstack}
		\caption{\emph{PRF-wrapped} variant.}\label{fig:prf-wrapped-cuckoo-filters}
	\end{subfigure}
	\caption{Beginning of the $\upS$ algorithm for insertion-only Cuckoo filter variants.}
\end{figure}

Unfortunately, 
Cuckoo filters are not function-decomposable.
From the first few instructions of $\up^{H_I, H_T}$ (see~Fig.~\ref{fig:up-cuckoo-not-f-decomposable}), we see that both hash functions $H_I$ and $H_T$ (which could be replaced with PRFs as for Bloom filters) need to be evaluated on $x$. If one were to attempt $H_T$-decomposition of $\upS$, that is, to instantiate $\upS^{\identity_\range, H_I}(H_T(x), \sigma)$, it would not be possible to evaluate $i_1 \gets H_I(x)$ from $H_T(x)$ alone. An attempt to $H_I$-decompose $\upS$ would pose the reverse problem, having to evaluate $tag \gets H_T(x)$ from only $H_I(x)$. It would also introduce a new problem, having to evaluate $i_2 \gets i_1 \xor H_I(tag)$ without access to $H_I$.\footnote{A proof of NAI  would also struggle with $H_I$ being evaluated both on $x$ and on $tag$ for every call to $\upS$: it would mean that on a single call to $\upS(x)$, $\advA$ would also learn $H_I(tag)$, where $tag \ne x$ with high probability.} To overcome this barrier, we propose the following minor variant of insertion-only Cuckoo filters that achieves function-decomposability and satisfies all our consistency rules, allowing it to satisfy the requirements of Theorem~\ref{thm:general-correctness}.
\paragraph{PRF-wrapped insertion-only Cuckoo filters.}
To address function-decomposability issues in insertion-only Cuckoo filters, we propose a generic technique: preprocessing the inputs $x \in \domain$ to $\upS, \qryS$ with a random function $F \colon \domain \rightarrow \range$, for some $\range \subset \domain$ (including potentially $\range = \domain$). This results in the $\upS^{F, H_T, H_I}(x, \sigma)$ algorithm shown in Fig.~\ref{fig:prf-wrapped-cuckoo-filters}, and in a similarly ``PRF-wrapped'' $\qryS$ algorithm. 
The resulting \amq is easy to implement since it only requires adding a PRF call on inputs, before passing them to the existing $\upS^{H_T, H_I}$ and $\qryS^{H_T,H_I}$ implementations.

\begin{restatable}{lemma}{cfareconsistent}\label{lem:cf-prf-factorable}
	PRF-wrapped insertion-only Cuckoo filters with oracle access to a random function $F$ are $F$-decomposable, reinsertion invariant, and satisfy the insertion-only \amq consistency rules from Def.~\ref{def:bf-consistency-rules}.
\end{restatable}
\begin{proof} 
	This follows by inspection of the proposed modifications of the $\upS$ and $\qryS$ algorithms in Fig.~\ref{fig:prf-wrapped-cuckoo-filters}. \ACM{}{\qed}
\end{proof}

As a consequence of adding a PRF computation on inputs, the NAI false positive probability is slightly increased by the probability of finding a collision of the PRF. 

\begin{restatable}{lemma}{prfcfpfp}\label{clm:cf-variants-fp-bound}
	Let $n$ be a non-negative integer, let $\Pi$ be a PRF-wrapped insertion-only Cuckoo filter with public parameters $pp = (s$, $\lambda_I$, $\lambda_T$, $num)$, wrapped using a PRF $R_K\colon \domain \rightarrow \range$, and let $\Pi'$ be an original insertion-only Cuckoo filter with the same public parameters $pp$ as $\Pi$'s. Let $\opfp{n} \coloneqq \overline{P_{\Pi',pp}}(FP \mid {n}) + \frac{(2s+2)^2}{2\,|\range|}$.
	Then $\pfp{n} \le \opfp{n}$, where $\frac{(2s+2)^2}{2\,|\range|}$ can be made cryptographically small.
\end{restatable}
\begin{proof}[sketch]
	The result follows by repeating the analysis for original Cuckoo filters~\cite{_CoNEXT:FAKM14} (Lemma~\ref{lem:cf-false-positive-probability-upperbound}), but accounting for the chance of a collision between $2s + 2$ uniformly random elements in $\range$. See \fullversion{\cref{app:proof-cf-variants-fp-bound}}{the full version} for the proof details.
\end{proof}

Finally, we can apply Theorem~\ref{thm:general-correctness}.
\begin{corollary}
Let $n, q_u, q_t, q_v$ be non-negative integers, and let $t_a, t_d > 0$. Let $F: \domain \to \range$.
Let $\Pi$ be a PRF-wrapped insertion-only Cuckoo filter with public parameters $pp$ and oracle access to $F$. 
If $R_K$ for $K \getsr \spaceK$ is an $(n+q_u+q_t, t_a + t_d, \varepsilon)$-secure pseudorandom function and $F = R_K$, then $\Pi$ is $(q_u, q_t, q_v, t_a, t_s, t_d, \varepsilon')$-adversarially correct, where $\varepsilon' = \varepsilon + 2 q_t \cdot \pfp{n+q_u}$ and $t_s \approx t_a$.
\end{corollary}
\begin{proof}
From the instantiation of PRF-wrapped Cuckoo filters (Fig.~\ref{fig:prf-wrapped-cuckoo-filters}), observe that each $\upS$ and $\qryS$ call contains one call to the function $F$. Applying Lemma~\ref{lem:cf-prf-factorable}, Theorem~\ref{thm:general-correctness} holds with $\alpha = \beta = 1$. \ACM{}{\qed}
\end{proof}

\subsection{Discussion on correctness}\label{sec:discussion}\label{sec:correctness-discussion}

We start by comparing our results with those of Clayton et al. \cite{_CCS:ClaPatShr19}. They analysed the adversarial correctness of PDS under four deployment settings, characterised by having \emph{private} or \emph{public} \emph{representations}, and by being either \emph{mutable} or \emph{immutable}. Our correctness result (Theorem~\ref{thm:general-correctness}) holds in all four settings by allowing the adversary up to $q_v$ $\revO$ queries and $q_u$ $\upO$ queries. The immutable setting then corresponds to $q_u = 0$ and the private setting to $q_v = 0$.

In \cite{_CCS:ClaPatShr19}, a game-based approach is used to derive bounds on the correctness of Bloom filters in the above adversarial settings. There, the adversary's win condition is to cause the event $S_r \coloneqq [\advA$ makes $\qryO(\cdot)$ queries resulting in at least $r$ false positives$]$. Using our simulation-based approach, we can derive new bounds for the setting of~\cite{_CCS:ClaPatShr19}. To see how, notice that extending Remark~\ref{rem:whyAdvCorrect}, Theorem~\ref{thm:general-correctness} also gives a bound on the difference in probabilities of event $S_r$ in both the real and ideal worlds. Our bound then involves a term of the form $2 q_t \cdot \pfp{n+q_u}$, while the bound in \cite{_CCS:ClaPatShr19} depends strongly on $r$. This should not be surprising given that our approach is general while that of \cite{_CCS:ClaPatShr19} involves a specific winning condition posed in terms of $r$. While this implies we are less flexible with the value of $r$ we implicitly tolerate, our approach covers any adversary (with no assumptions on its behaviour), illustrating the power of simulation-based notions. This includes adversaries who specify their objective not explicitly in terms of false positives caused, but perhaps in terms of a target Hamming weight of the Bloom filter's state, or in terms of subsets of the filter's index to be set to $1$ (as in target-set coverage attacks~\cite{_CCS:ClaPatShr19}). Further, by not requiring a choice of $r$, our results do not require satisfying particular constraints on $r$, as in~\cite[Theorem~3]{_CCS:ClaPatShr19}. 
In \fullversion{Appendix~\ref{sec:comparison}}{the full version}, we provide an in-depth comparison of our results with those of~\cite{_CCS:ClaPatShr19}, and outline how our approach covers other adversarial objectives such as target-set coverage attacks. 

We note that in the immutable setting, a slight modification of our Theorem~\ref{thm:general-correctness} proof gives a tighter bound with $\varepsilon' = \varepsilon$ (\fullversion{Appendix~\ref{sec:immutable}}{see the full version}).

Although our proof methodology is similar in spirit to that of HLL~\cite{PatRay22}, the analysis of adversarial correctness for \amq is much more involved. The ability to make membership queries on an element gives the adversary valuable information on how useful it would be to insert that element; cardinality estimates, on the other hand, do not reveal such information. 

Finally, we comment on the implications of our bounds. While the distinguishing advantage in $\gameRoI$ may not be negligible as one might expect in cryptographic proofs, by relating the adversarial setting to the well-studied honest setting from the PDS literature, we can place concrete bounds on the success of any \amq adversary. We return to this point in \S~\ref{sec:secure-instantiations}.

\section{Privacy}\label{sec:privacy}

We now shift focus to privacy guarantees for \amq, addressing the following question: to what extent does the functionality of an \amq compromise the privacy of the elements that it stores? We explore simulation-based privacy notions for \amq. We propose two such notions for quantifying privacy, each associated with a different leakage profile in the ideal world, and investigate the relationship between the two. We identify a specific property, \emph{permutation invariance (PI)}, as being of central importance in establishing privacy, and show that it is implied by function-decomposability of an \amq. Finally, we apply our results to Bloom and PRF-wrapped insertion-only Cuckoo filters.

\paragraph{Settings.}\label{sec:privacy-settings} 

We model interactions by an adversary $\advA = (\advA_1, \advA_2)$ with an $\amq$ $\Pi$ in two stages.
In the first stage, a randomised algorithm $\advA_1$ populates $\Pi$ with a set of elements $V \subset \domain$ via the $\repO$ oracle.
In the second stage, an adversary $\advA_2$ attempts to learn something about $V$.
We stress that no state is shared between $\advA_1$ and $\advA_2$ (in contrast, in \S~\ref{sec:deployment-settings}, the implicit two-part adversary is allowed to share state).
 
We consider two adversarial settings. 
In the \emph{snapshot} setting, at the end of the first stage, the adversary $\advA_2$ is given $\Pi$'s state $\sigma$ via access to the $\revO$ oracle, but has no other oracle access.
In the \emph{adaptive} setting, $\advA_2$ is instead given query access to $\Pi$ via the $\qryO$, $\upO$ and $\revO$ oracles. 
These settings capture various real world scenarios. For example, a system intrusion might lead to the leakage of $\sigma$ but go unnoticed for some time, and hence public access to $\Pi$ is not disabled. Analysing privacy in the adaptive setting allows us to quantify what harm an adversary could do in such a scenario. 

\paragraph{Leakage profiles for \amq.}\label{sec:privacy-leakage-profile}

To define a notion of privacy for an \amq $\Pi$, we first need to characterise its leakage profile. This describes the information leaked as a result of $\Pi$'s functionality, which also depends on whether we are in the snapshot or adaptive setting. We model the leakage as a set of functions that a simulator is allowed to use as oracles, and justify their inclusion below.

In Fig.~\ref{fig:leaks:PDS:ElemRepPriv} we define two leakage functions for an \amq representing a set $\setV$ of elements.
The first is $\repleakO$, which leaks $|V|$.  
This captures the fact that for commonly used $\amq$, one can estimate $|V|$ from observing $\sigma$. 
For example, if $\Pi$ is a Bloom filter, this could be indicated by the number of bits set to 1 in $\sigma$, or if $\Pi$ is a Cuckoo filter, it could be estimated from the number of tags stored in $\sigma$ (up to the probability of collisions in $H_T$).

In settings that allow membership queries, we require a second leakage function $\elemleakO$, capturing the fact that $\advA_2$ can always issue $\qryO(x)$ queries and learn their output. While this may seem to result in a weak privacy notion, such leakage is unavoidable in the real world if access to the \amq' API is provided. 

In this section, we will show that this leakage profile serves as an upper bound of the real leakage of function-decomposable \amqs, by constructing simulators that use it to provide consistent views of \amq to adversaries.

\begin{figure}[t]
	\centering
	\begin{pchstack}[boxed,space=1em]
		\procedure[linenumbering,space=1em,headlinecmd={\vspace{.1em}\hrule\vspace{.3em}}]{oracle $\repleakO()$}{%
			\pcreturn \abs{\setV}
		}
		\procedure[linenumbering,space=1em,headlinecmd={\vspace{.1em}\hrule\vspace{.3em}}]{oracle $\elemleakO(x)$}{%
			\pcreturn [x \in \setV]
		}
	\end{pchstack}
	\caption{Leakage profile for \amq.}
	\label{fig:leaks:PDS:ElemRepPriv}
\end{figure}

\subsection{Notions of Privacy: Elem-Rep privacy}\label{sec:privacy-definition}

We start with a high-level explanation of our first privacy definition, \textit{Elem-Rep privacy}. We again employ a simulation-based approach. Let the \amq be populated with elements from a set $\setV$ by adversary $\advA_1$. Then, adversary $\advA_2$ interacts with the \amq through the setting-specific oracles. 
Adversary $\advA_2$ plays in either a real or ideal world. In the real world, it interacts with a keyed \amq initialised with the elements in $\setV$. In the ideal world, it interacts with a simulator $\simS$. The simulator does not know $V$, but has access to $\repleakO$ in the snapshot setting, and additionally $\elemleakO$ in the adaptive setting. The output of $\advA_2$ is then given to a distinguisher $\dD$ along with $\setV$. By showing that the distinguisher's outputs in the real and ideal world are close, we prove that the adversary cannot learn much more about the elements in $\setV$ through interacting with the \amq in the real world than in the ideal world (where it can only learn the specified leakage from $\simS$). 

We formalise the above in Fig.~\ref{game:priv:ElemRep:r-o-i}, in the $\gameRoIERP$ game. We will use $\Real$ and $\Ideal$ to denote the real ($d {=} 0$) and ideal ($d{=}1$) versions of the $\gameRoIERP$ game. 
As in \S~\ref{sec:correctness}, if an oracle $\textbf{O}$ is not directly specified, we will assume it is defined as in Fig.~\ref{game:r-o-i}.

\begin{definition}[Elem-Rep privacy]\label{def:elem-rep-privacy} 
Let $\Pi$ be an insertion-only \amq with public parameters $pp$, and let $R_K$ be a keyed function family.
Let $\advA=(\advA_1,\advA_2)$ be a tuple of algorithms. 
We say $\Pi$ is $(q_u,q_t, q_v, t_a, t_s, t_d, \varepsilon)$-Elem-Rep private if 
there exists a simulator $\simS$ (that runs $\advA_2$, calls oracle $\repleakO$ at most once, calls oracle $\elemleakO$ only when $\advA_2$ calls its $\qryO$, $\upO$ oracles and on the same argument,
and runs in time at most $t_s$) 
such that, 
for all $\advA_1$,
for all $\advA_2$ running in time at most $t_a$ and making $q_u$, $q_t$, $q_v$ queries to oracles $\upO$, $\qryO$, $\revO$ respectively, 
and for all distinguishers $\dD$ running in time at most $t_d$, 
we have:
\begin{align*}\label{eq::priv:ElemRep:r-o-i}
 \text{Adv}&_{\Pi, \advA, \simS}^{RoIElemRepP} (\dD) \\
 & {\coloneqq}\,\big| \Pr[\Real(\advA, \dD)\,{=}\,1] {-} \Pr[\Ideal(\advA, \dD, \simS)\,{=}\,1] \big|  \,{\leq}\, \varepsilon,
\end{align*}
in $\gameRoIERP$ (Fig.~\ref{game:priv:ElemRep:r-o-i}).
\end{definition}

\begin{figure}[t]
	\Wider[10em]{
		\centering
		\begin{pchstack}[boxed,space=.0em]
			\procedure[linenumbering,headlinecmd={\vspace{.1em}\hrule\vspace{.3em}}]{$\gameRoIERP(\advA, \simS, \dD, pp)$}{%
						d \getsr \{0, 1\} \\ 
						\tinit \gets \bot\\
						V \getsr \advA_1 \label{line:elemrepgame-advA1} \\
						\pcif d = 0 \t \pcmycomment{\Real} \\
						\t K \getsr \mathcal{K}; F \gets R_K \\
						\t \sigma \gets \setupS(pp) \\
						\t \repO(V) \\
						\t out \getsr \advA_2^{\scriptsize \pcbox{\upO, \qryO,}\ \textbf{Reveal}}\\
						\pcelse \t \pcmycomment{\Ideal} \\
						\t out \getsr \simS^{\scriptsize \pcbox{\elemleakO,}\  \repleakO}(\advA_2, pp) \\
						d' \getsr \dD(out, \setV) \\ 
						\pcreturn d' 
			}
		\procedure[linenumbering,headlinecmd={\vspace{.1em}\hrule\vspace{.3em}}]{$\gameRoIRP(\advA, \simS, \dD, pp)$}{%
			d \getsr \{0, 1\} \\ 
			\tinit \gets \bot\\
			V \getsr \advA_1\\
			\pcif d = 0 \t \pcmycomment{$\overline{\Real}$} \\
			\t K \getsr \mathcal{K};\ F \gets R_K \\
			\t \sigma \gets \setupS(pp) \\
			\t \pcfor x \in V\\
			\t \t (b,\sigma) \getsr \upS^F(x,\sigma)\\
			\t \tinit \gets \top\\
			\t out \getsr \advA_2^{\scriptsize \pcbox{\upO, \qryO,}\ \revO}\\
			\pcelse \t \pcmycomment{$\overline{\Ideal}$} \\
			\t out \getsr \simS^{\repleakO}(\advA_2, pp) \\
			\pcreturn d' \getsr \dD(out, \setV)
		}
		\end{pchstack}
	}
	\caption{Elem-Rep (resp.~Rep) privacy game for \amq $\Pi$, with respect to the leakage profile obtained by combining $\repleakO$ and $\elemleakO$ (resp.~the leakage profile consisting only of $\repleakO$), in the snapshot and \pcbox{\text{adaptive}}settings.}
	\label{game:priv:ElemRep:r-o-i}\label{game:priv:Rep:r-o-i}
\end{figure}

	Informally, Def.~\ref{def:elem-rep-privacy} implies that the API of an Elem-Rep private \amq $\Pi$ does not leak more than the number of elements in $\Pi$ and the true query responses for elements queried via $\upO$ and $\qryO$.

	In the following sections, we show how to compute bounds on the Elem-Rep privacy of an \amq. We start by introducing a property that we call  \emph{permutation invariance} (PI), show that it is implied by function-decomposability, and that in turn it implies a bound on Elem-Rep privacy of an \amq.
	
\subsection{Permutation invariance (PI) 
}\label{sec:pi-f-decomp}
Consider the $\gamePI$ game in Fig.~\ref{g:priv:PI:nn}, where an adversary $\advA$, who has access to the $\repO, \upO, \qryO, \revO$ oracles, must distinguish between an \amq where the inputs to all queries are either randomly permuted or not.

\begin{definition}\label{def:PI:nn} 
	Let $\Pi$ be an insertion-only \amq, with public parameters $pp$.
	We say $\Pi$ is $(q_u,q_t, q_v, t_a, \varepsilon)$-\textup{permutation invariant} ($\varepsilon$-PI for short, $0$-PI when $\varepsilon = 0$)
	if, for all adversaries $\advB$ running in time at most $t_a$ and making first a single query to $\repO$ and then $q_u$, $q_t$, $q_v$ queries to oracles $\upO$, $\qryO$, $\revO$ respectively, we have:		
	\begin{align*}
		\text{Adv}_{\Pi}^{PI} (\advB)
		{:=} \abs{\Pr[\gamePI(\advB)\,{=}\,1\,|\,c\,{=}\,0] {-} \Pr[\gamePI(\advB)\,{=}\,1\,|\,c\,{=}\,1]}{\leq}\varepsilon,
	\end{align*}
	in $\gamePI$ (Fig.~\ref{g:priv:PI:nn}). We say $\advB$ is a $(q_u,q_t,q_v,t_a)$-PI adversary.
\end{definition}	
	
\begin{figure}[t]
	\Wider[10em]{
			\centering
		\begin{pchstack}[boxed,space=0pt]
			\begin{pcvstack}
			\procedure[linenumbering,space=1em,headlinecmd={\vspace{.1em}\hrule\vspace{.3em}}]{$\gamePI(\advB)$}{%
				F \getsr \Funcs[\domain, \range] \\ 
				\tinit \gets \bot \\
				\sigma \gets \setupS(pp) \\
				c \getsr \{0, 1\} \\ 
				\pcif c = 1\\
				\t \pi \getsr \Perms[\domain]\\
				\pcelse:\,\pi \gets \identity_\domain \\
				c' {\getsr} \advB^{\tiny \repO, \pcbox{\upO, \qryO,} \revO} \\
				\pcreturn c'
			}
			\end{pcvstack}
			\begin{pcvstack}
				\procedure[linenumbering,headlinecmd={\vspace{.1em}\hrule\vspace{.3em}}]{Oracle $\repO(\setV)$}{%
					\pcif \tinit = \top: \pcreturn \bot \\
					\tinit \gets \top \\
					\pcfor x \in V\\
					\ \, (b,\sigma){\getsr}\upS^{F}(\pi(x), \sigma)\\
					\pcreturn \top
				}
			\procedure[linenumbering,headlinecmd={\vspace{.1em}\hrule\vspace{.3em}}]{Oracle $\upO(x)$}{%
				\pcif \tinit = \bot: \pcreturn \bot \\
				(b,\sigma){\getsr}\upS^{F}(\pi(x), \sigma)\\
				\pcreturn b
			}
		\end{pcvstack}
		\begin{pcvstack}
			\procedure[linenumbering,headlinecmd={\vspace{.1em}\hrule\vspace{.3em}}]{Oracle $\qryO(x)$}{%
				\pcif \tinit = \bot \\
				\t \pcreturn \bot \\
				a{\gets}\qryS^{F}(\pi(x), \sigma) \\
				\pcreturn a
			}
			\procedure[linenumbering,headlinecmd={\vspace{.1em}\hrule\vspace{.3em}}]{Oracle $\revO$()}{%
				\pcreturn \sigma
			}
		\end{pcvstack}
		\end{pchstack}
	}
		\caption{PI game for AMQ-PDS $\Pi$ in the snapshot and \pcbox{\text{adaptive}} settings.}
		\label{g:priv:PI:nn}
	\end{figure}

Our next result relates function-decomposability of an \amq (as per Def.~\ref{def:f-decomposability}) to its permutation invariance.
	\begin{restatable}{lemma}{fdecomparezeropi}
		Let $F \getsr \Funcs[\domain,\range]$ be a random function, and let $\Pi$ be an $F$-decomposable \amq with public parameters $pp$ and oracle access to $F$. Then $\Pi$ is $0$-PI.
		\label{l:f-dec:PI}
	\end{restatable}
\begin{proof}[sketch]
By $F$-decomposability, we rewrite $\upS^F(\pi(x),\, \sigma)$ calls in $\gamePI$ as $\upS^{\identity_\range}(F(\pi(x)),\, \sigma)$, and similarly for $\qryS$. Since $F(x)$ and $F(\pi(x))$ are random functions, the $c\,{=}\,0,\,1$ versions of $\gamePI$ are indistinguishable. A full proof is given in \fullversion{Appendix~\ref{app:proof-f-dec:PI}}{the full version}.
\end{proof}

Similarly to correctness (\S~\ref{sec:correctness-notions}), Lemma \ref{l:f-dec:PI} only refers to a single oracle function $F$, but it also holds for \amqs using $t$ oracle functions $F_1, \dots F_t$ and being $F_1$-decomposable.

	\subsection{Elem-Rep privacy from PI}\label{sec:privacy-proofs}
	
	We now show how to compute a bound on the Elem-Rep privacy of an \amq using permutation invariance.

	\begin{restatable}{theorem}{newelemrepprivacy}\label{th:general-privacy}
		Let $q_u,\, q_t, \, q_v$ be non-negative integers, and let $t_a,\, t_d > 0$.
		Let $F: \domain \to \range$.
		Let $\Pi$ be an insertion-only \amq with public parameters $pp$ and oracle access to $F$.
		Let $n$ be the maximum number of elements returned by $\advA_1$ on line~\ref{line:elemrepgame-advA1} of $\gameRoIERP$ in Fig.~\ref{game:priv:ElemRep:r-o-i}.
		Let $\alpha$ (resp. $\beta$) be the number of calls to $F$ required to insert (resp. query) an element in $\Pi$ using its $\upS$ (resp.~$\qryS$) algorithm.
		
		If $F \equiv R_K \colon \domain \rightarrow \range$ is an $(\alpha (n + q_u) + \beta q_t, t_a + t_d, \varepsilon_{\text{PRF}})$-secure pseudorandom function with key $K \getsr \mathcal{K}$,
		and $\Pi$ is $(q_u, q_t, q_v, t_a, \varepsilon_{\text{PI}})$-permutation invariant,
		then $\Pi$ is $(q_u,$ $q_t$, $q_v$, $t_a$, $t_s$, $t_d$, $\varepsilon_{\text{PRF}} + \varepsilon_\text{PI})$-Elem-Rep private, where $t_s \approx t_a$.
	\end{restatable}
	\begin{proof}[sketch]
	We construct a simulator $\simS_\text{Elem-Rep}$ for $\Ideal$ in Fig.~\ref{f:SElemRep}, and an adversary $\advB$ in $\gamePI$ that runs the $\gameRoIERP$ adversary $\advA$ internally. Then, we rewrite $\advA$'s advantage in terms of the permutation invariance of $\Pi$, the distance between $\Real$ and the $c = 0$ version of $\gamePI$, and that of $\Ideal$ and the $c = 1$ version of $\gamePI$. The first (resp.~second) term is bounded by $\varepsilon_\text{PI}$ (resp.~$\varepsilon_{\text{PRF}}$). Observing that $\Ideal$ is indistinguishable from the $c\,{=}\,1$ version of $\gamePI$ gives the result.
	The full proof is given in \fullversion{Appendix~\ref{app:proof-general-privacy}}{the full version}.
\begin{figure}
	\Wider[2em]{
		\small
		\centering
		\begin{pcvstack}[boxed]
			\begin{pchstack}[space=.1em]
				\procedure[linenumbering,headlinecmd={\vspace{.1em}\hrule\vspace{.3em}}]{Simulator $\simS_\text{Elem-Rep}(\advA_2, pp)$}{%
					F \getsr \Funcs[\domain,\range] \\
					\sigma \gets \setupS(pp) \\
					P \gets \{\}\ \pcmycomment{key-value store} \\
					Y \gets \{\}\ \pcmycomment{set} \\ 
					\repsimO() \\
					out \getsr \advA_2^{\upsimO,\qrysimO,\revealsimO} \\
					\pcreturn out
				}
				\procedure[linenumbering,headlinecmd={\vspace{.1em}\hrule\vspace{.3em}}]{Procedure $\repsimO()$}{%
					n' \gets \repleakO() \\ 
					\pcfor i \in [n'] \\
					\ \ y \getsr \domain \setminus Y \\
					\ \ Y \gets Y \cup \{ y \} \\ 
					\ \ (b, \sigma) \getsr \upS^F(y, \sigma)\\
					\pcreturn \top 
				}
			\end{pchstack}
			\begin{pchstack}[space=.1em]
				\procedure[linenumbering,headlinecmd={\vspace{.1em}\hrule\vspace{.3em}}]{Procedure $\PermuteO(x)$}{%
					\pcif P[x] = \bot \\
					\ \ \pcif \elemleakO(x) = \top \\
					\ \ \t P[x] \getsr Y \setminus \Values{P} \\
					\ \ \pcelse \\
					\ \ \t P[x] \getsr \domain \setminus (\Values{P}\,{\cup}\,Y) \\
					\pcreturn P[x] 
				}
				\begin{pcvstack}[space=.1em]
					\procedure[linenumbering,space=1em,headlinecmd={\vspace{.1em}\hrule\vspace{.3em}}]{Oracle $\upsimO(x)$}{%
						(b, \sigma) \getsr \upS^F(\PermuteO(x), \sigma) \\
						\pcreturn b
					}
					\procedure[linenumbering,headlinecmd={\vspace{.1em}\hrule\vspace{.3em}}]{Oracle $\qrysimO(x)$}{%
						\pcreturn \qryS^F(\PermuteO(x), \sigma)
					}
					\procedure[linenumbering,headlinecmd={\vspace{.1em}\hrule\vspace{.3em}}]{Oracle $\revealsimO()$}{%
						\pcreturn \sigma
					}
				\end{pcvstack}
			\end{pchstack}
		\end{pcvstack}
	}
	\caption{Simulator used to prove Theorem~\ref{th:general-privacy}.}
	\label{f:SElemRep}
\end{figure}
	\end{proof}
	
	This result essentially tells us that, up to the $\varepsilon_\text{PRF} + \varepsilon_\text{PI}$ bound, the state of a function-decomposable \amq representing some set $V$ does not leak more than the number of elements in the data structure, and that querying the \amq does not reveal more than the true answers to set membership queries of queried elements. However, such a guarantee may not be useful for computing concrete bounds on privacy in practice, as it does not explicitly quantify the impact of what $\elemleakO$ reveals or how this is related to the distribution of elements in $V$. This motivates our alternative privacy notion, introduced in the next section.
	
	\subsection{Notions of Privacy: Rep privacy} 
	
	In this section, we define a second privacy notion, \textit{Rep privacy}. Here, the simulator no longer has access to the $\elemleakO$ oracle, but instead only to the $\repleakO$ oracle, allowing it to learn $|V|$ where $V \getsr \advA_1$. We formalise this privacy definition by the $\gameRoIRP$ game in Fig.~\ref{game:priv:Rep:r-o-i}. We will use $\overline{Real}, \overline{Ideal}$ to denote the real ($d = 0$), ideal ($d = 1$) versions of the game, respectively.
	
	\begin{definition}[Rep privacy]\label{def:rep-privacy}
		Let $\Pi$ be an insertion-only \amq with public parameters $pp$, and let $R_K$ be a keyed function family.
		Let $\advA=(\advA_1,\advA_2)$ be a tuple of algorithms. 
		We say $\Pi$ is $(q_u,q_t, q_v, t_a, t_s, t_d, \varepsilon)$-Rep private if
		there exists a simulator $\simS$ (that runs $\advA_2$, calls $\repleakO$ at most once, and runs in time at most $t_s$) 
		such that, 
	    for all $\advA_1$,
		for all $\advA_2$ running in time at most $t_a$ and making $q_u,q_t, q_v$ queries to oracles $\upO,\qryO,\revO$ respectively, and
		for all distinguishers $\dD$ running in time at most $t_d$,
		we have:
		\begin{align*}
		\text{Adv}_{\Pi, \advA, \simS}^{RoIRepP} (\dD) {:=} 
		\big| \Pr[\overline{Real}(\advA, \dD) {=} 1] {-} \Pr[\overline{Ideal}(\advA, \dD, \simS) {=} 1] \big| 
	{\leq} \varepsilon,
		\end{align*}
		in $\gameRoIRP$ (Fig.~\ref{game:priv:Rep:r-o-i}).
	\end{definition}

		While the notion of Rep privacy is exactly the same as Elem-Rep privacy in the snapshot setting, their relationship is more subtle in the adaptive setting. By removing the simulator's access to $\elemleakO$, the bound obtained through Rep privacy is directly related to the probability of guessing elements in $V$. In fact, we show this formally in the following theorem.

	\begin{restatable}{theorem}{oldrepprivacythm}\label{th:er-to-r}
		Let $q_u, q_t, q_v$ be non-negative integers, and let $t_a, t_d > 0$.
		Let $\Pi$ be an insertion-only \amq, with public parameters $pp$.	
		Suppose $\Pi$ is $(q_u$, $q_t$, $q_v$, $t_a$, $t_s$, $t_d$, $\varepsilon)$-Elem-Rep private with simulator $\simS$. 
		Then there exists a simulator $\simS'$ (that is constructed from $\simS$) such that 
		$\Pi$ is $(q_u, q_t, q_v, t_a, t_{s'}, t_d, \varepsilon')$-Rep private with simulator $\simS'$ and $t_s \approx t_{s'}$.
		Here $\varepsilon' = \varepsilon + \Pr[{W \cap \setV \ne \emptyset}]$, with 
		$W$ denoting the set of elements queried by $\advA_2$ to its $\upO, \qryO$ oracles in $\overline{\Ideal}$ within $\simS'$. 
	\end{restatable}
\begin{proof}[sketch]
Our goal is to relate $\overline{\Real}, \overline{\Ideal}$ from $\gameRoIRP$ to $\Real, \Ideal$ from $\gameRoIERP$. We first construct a simulator $\simS'$ that is the same as $\simS$, but where every $\elemleakO$ call is replaced with $\bot$. Then, $\Ideal$ and $\overline{\Ideal}$ (the worlds simulated by $\simS$ and $\simS'$, respectively) are identical up until $\advA_2$ makes an $\upO$ or $\qryO$ query on something in $V$, which occurs with probability $\Pr[{W \cap \setV \not= \emptyset}]$. Since $\Real$ and $\overline{\Real}$ are identical, we use the $\varepsilon$-Elem-Rep privacy of $\Pi$ to obtain the result. See\fullversion{~Appendix~\ref{app:proof-er-to-r}}{ the full version} for the full proof.
\end{proof}

While Theorems \ref{th:general-privacy} and \ref{th:er-to-r} only refer to a single oracle function $F$, the same result holds also for \amqs using $t$ oracle functions $F_1, \dots F_t$ and being $F_1$-decomposable.
	
\subsection{Guarantees for Bloom and Cuckoo filters}

Using our analysis, we derive results on the privacy of Bloom filters.

\begin{restatable}{corollary}{bfelemrepprf}\label{c:BF:F:ERP}
		Let $n, q_u, q_t, q_v$ be non-negative integers, and let $t_a, t_d > 0$. Let $F: \domain \to \range$.
		Let $\Pi$  be a Bloom filter with public parameters $pp$ and oracle access to $F$. Let $\delta$ denote $ \pr{W \cap V \neq \emptyset}$.
		If $R_K$ for $K \getsr \spaceK$ is an $(n+q_u+q_t, t_a + t_d, \varepsilon)$-secure PRF and $F = R_K$, then $\Pi$ is $(q_u, q_t, q_v, t_a, t_s, t_d, \varepsilon + \delta)$-Rep private, where $t_s \approx t_a$.
\end{restatable}
\begin{proof}
Since $\Pi$ is $F$-decomposable (Lemma \ref{l:f-dec:BF}), it satisfies $0$-PI (Lemma \ref{l:f-dec:PI}). We then set $\varepsilon_\text{PI} = 0$, $\varepsilon_\text{PRF} = \varepsilon$ in Theorem \ref{th:general-privacy} with $\alpha = \beta = 1$ to obtain an Elem-Rep privacy bound, and finally apply Theorem \ref{th:er-to-r} to convert this to a Rep privacy bound. \ACM{}{\qed} 
\end{proof}

We can prove similar results for the Cuckoo filter.
\begin{restatable}{corollary}{cfelemrepprf}\label{c:CF:F:ERP}
		Let $n, q_u, q_t, q_v$ be non-negative integers, and let $t_a, t_d > 0$. Let $F: \domain \to \range$.
		Let $\Pi$  be a PRF-wrapped insertion-only Cuckoo filter with public parameters $pp$ and oracle access to $F$. Let $\delta$ denote $\pr{W \cap V \neq \emptyset}$.
		If $R_K$ for $K \getsr \spaceK$ is an $(n+q_u+q_t, t_a + t_d, \varepsilon)$-secure pseudorandom function and $F = R_K$, then $\Pi$ is $(q_u, q_t, q_v, t_a, t_s, t_d, \varepsilon + \delta)$-Rep private, where $t_s \approx t_a$.
\end{restatable}
\begin{proof}[sketch]
The proof proceeds similarly to Corollary \ref{c:BF:F:ERP}, using Lemmas \ref{lem:cf-prf-factorable} and \ref{l:f-dec:PI} along with Theorems \ref{th:general-privacy} and \ref{th:er-to-r}.
\end{proof}

We note that our PRF-wrapped variant of the insertion-only Cuckoo filter simplifies not only our correctness analysis but also privacy, by attaining function-decomposability.
Furthermore, the original insertion-only Cuckoo filter (Fig.~\ref{fig:up-cuckoo-not-f-decomposable}) does not satisfy $0$-PI due to a trivial distinguishing attack (see \fullversion{Appendix~\ref{app:cf-not-pi}}{full version}).

\subsection{Discussion on privacy}\label{sec:privacy-other-pds}

We explored two ways of defining a simulation-based privacy notion for \amq, each with respect to a specific leakage profile. Our first privacy definition, Elem-Rep privacy, has a leakage profile capturing the information intrinsically leaked by the \amq about the elements it stores. However, the bound obtained by quantifying the Elem-Rep privacy is not trivial to interpret; one must also carefully analyse the leakage to determine the amount of privacy obtained in practice, as is common in simulation-based notions. Our alternative definition, Rep privacy, has a smaller and simpler leakage profile, but the bound explicitly depends on how easy it is to guess elements stored in the \amq. This approach may be more useful in practice; for example, it allows to directly relate privacy to the min-entropy of the distribution of stored elements.

Our results confirm the intuition that one cannot hope to achieve privacy if the elements stored in the \amq are easy to predict, a ``low min-entropy'' scenario. In this setting, $\elemleakO$ would reveal a substantial amount of information as the adversary is likely to query elements in $V$, leading to the Elem-Rep privacy bound being a weak guarantee in practice. Similarly, we would not obtain a good bound via Rep privacy either, as the term $\pr{W \cap V \neq \emptyset}$ would be high. On the other hand, a ``high min-entropy'' scenario results in strong guarantees from both privacy notions. In fact, the number of $\elemleakO$ queries that give the adversary useful information is directly linked to its likelihood of 
guessing elements in $V$. 

We note that our privacy theorems can be used to analyse various real-world scenarios. For example, by setting $q_u = 0$ in Theorems~\ref{th:general-privacy} and \ref{th:er-to-r}, we cover the "static data" scenario, where an application first adds a set of elements to a PDS, and the adversary's goal is to learn these elements through only set membership queries.

Throughout this section, we have assumed that leaking $|\setV|$ is acceptable, as this is unavoidable for Bloom and Cuckoo filters in the \emph{public} setting (i.e.~when the $\revO$ oracle is available). However, there may be settings where $|\setV|$ is sensitive, in which case one may want to investigate alternative \amq (see\fullversion{~\cref{app:leaking-nothing}}{ full version}).

\section{Secure instances}\label{sec:secure-instantiations}

We sketch how to use our results to instantiate \amq instances achieving provable guarantees. An expanded discussion is provided in \fullversion{Appendix~\ref{app:secure-instantiations}}{the full version}, covering both correctness and privacy. Here we focus on the former aspect.
We aim to use our bounds to set \amq parameters. Recall the guarantee from Theorem~\ref{thm:general-correctness}:
\begin{align*}
	&\text{Adv}_{\Pi, \mathcal{A, S}}^{{\text{RoI}}} (\dD) {=}  \abs{\Pr[\Real(\advA, \dD) {=} 1] {-} \Pr[\Ideal(\advA, \dD,\simS) {=} 1]} \\
	& {=} \abs{\Pr[\dD(\advA) {=} 1] {-} \Pr[\dD(\simS(\advA, pp)) {=} 1]} {\leq} \varepsilon +  2 q_t {\cdot} \pfp{n + q_u},
\end{align*}
where $\varepsilon$ is a PRF distinguishing advantage, $n$ is the number of elements initially inserted into the \amq, and $q \coloneqq q_u + q_t$ is the total number of queries made by $\advA$ that influence its success probability. Crucially, usually $\pfp{n + q_u}$ can be estimated using well-established upper bounds, cf.\ Lemmas~\ref{lem:bf-false-positive-probability} and~\ref{clm:cf-variants-fp-bound}.

This result allows us to establish an upper bound on the probability $\Pr[\dD(\advA) {=} 1]$ of an adversary $\advA$ with the given query budget $q$ finding a sequence of queries to an \amq $\Pi$ that allows them to satisfy some desired predicate $P$ in the $\Real$ world, by relating this probability to that of $\advA$ satisfying $P$ in the $\Ideal$ world. 
As a practical example, we investigate the choice of adversarially correct parameters for Bloom and PRF-wrapped Cuckoo filters for one possible query budget; two more budgets are explored in \fullversion{Appendix~\ref{app:secure-instantiations}}{the full version}. Concretely, we bound the probability that after $n+q_u$ adversarial insertions, querying a random non-inserted element returns a false positive result, cf.~Remark~\ref{rem:whyAdvCorrect}.  In Fig.~\ref{fig:main-body-correct-instances} we plot an upper bound of the false positive probability against the size of the data structure in both adversarial and non-adversarial settings for various public parameters. Our results show that achieving protection against adversarial inputs requires roughly doubling (Bloom) or trebling (insertion-only Cuckoo) the storage used, as compared to the honest setting.

\begin{remark}
	We stress that to obtain correctness and privacy, the key to any PRFs used needs to be stored securely, say in hardware, so as to resist being exposed by a state reveal. \fullversion{Furthermore, our theorems technically provide guarantees only as long as the number of maximum PRF queries assumed by the theorem is respected (e.g., $\alpha (n + q_u) + \beta q_t$ in Theorem~\ref{thm:general-correctness}). Since PRF re-keying would require additional bookkeeping to answer further queries correctly, the \amq shelf-life should be carefully considered when using our results.}{}
\end{remark}

\begin{figure}[t]
	\centering
	\hspace{-1em}
	\begin{subfigure}{.45\textwidth}
		\centering
			\adjustbox{trim={0\width} {8pt} {0} {0}, clip}{
				\includegraphics[width=.9\textwidth]{\detokenize{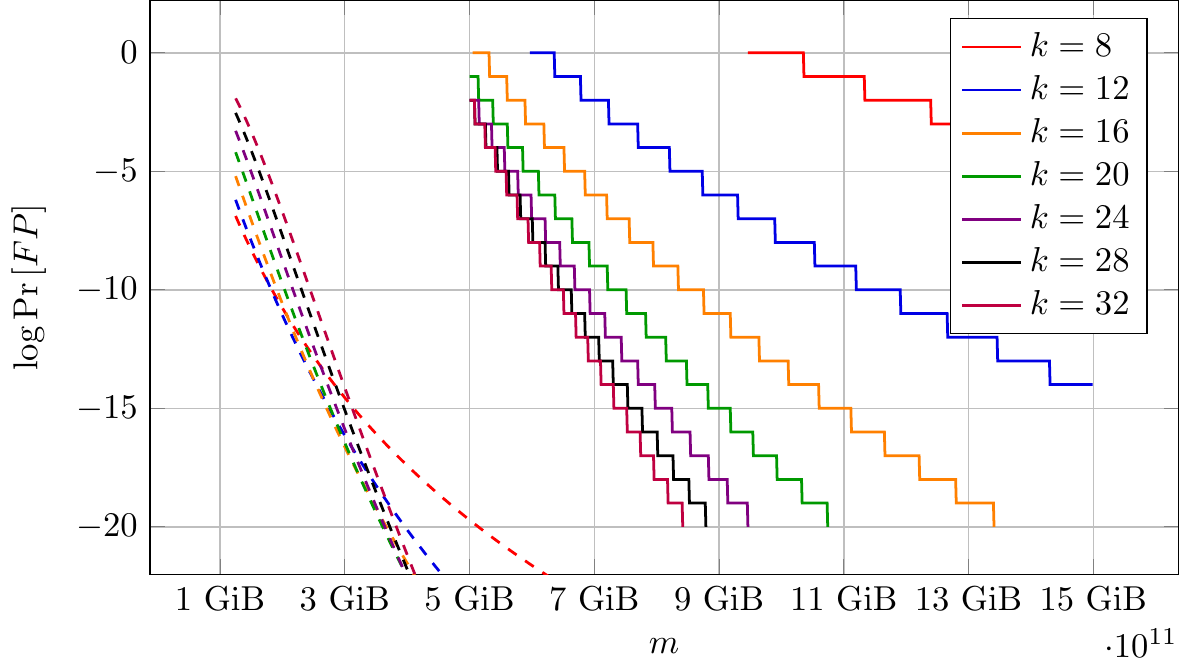}}}
		\caption{Bloom, $\log{n} = 7$, $\log{(q_u + q_t)} = 30$}\label{fig:mb-bf-correct-instances-small-n-big-q}
	\end{subfigure}
	\\\vspace{.5em}
	\hspace{-1em}
	\begin{subfigure}{.47\textwidth}
	\centering
		\adjustbox{trim={0\width} {9pt} {0} {0}, clip}{
			\includegraphics[width=.9\textwidth]{\detokenize{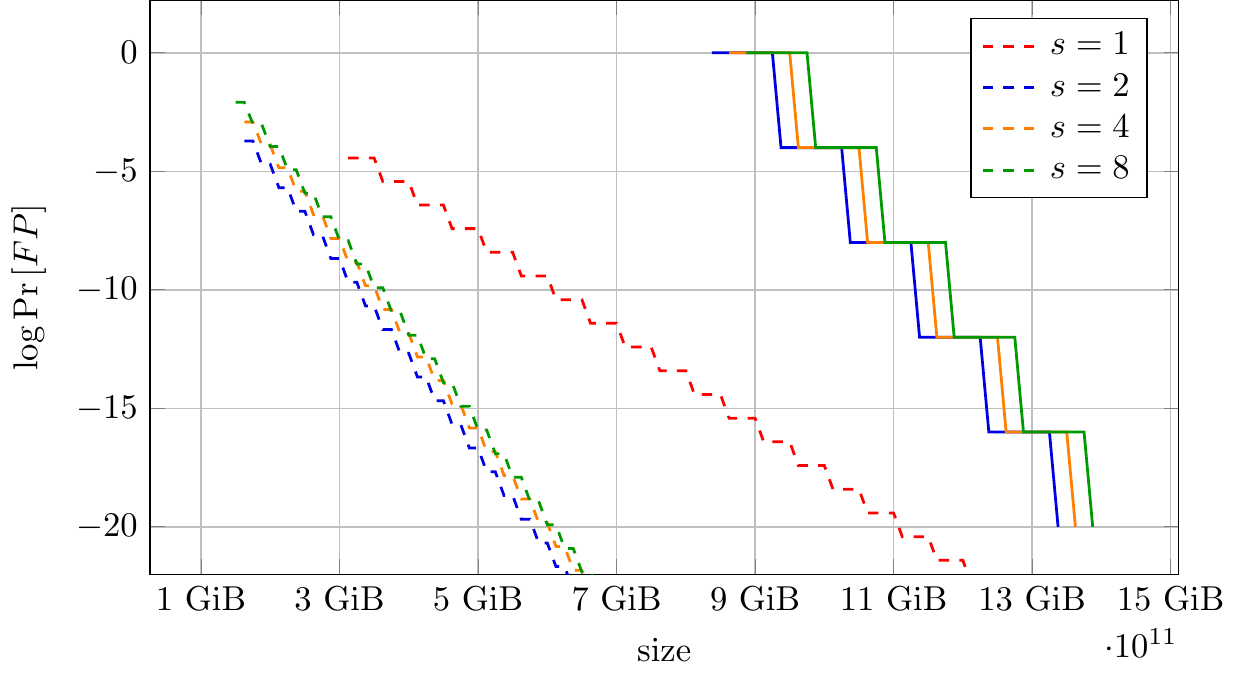}}}
	\caption{Cuckoo, $\log{n} = 7$, $\log{(q_u + q_t)} = 30$}\label{fig:mb-cf-correct-instances-small-n-big-q}
	\end{subfigure}
\caption{Correctness guarantees vs.~storage trade-offs for Bloom and PRF-wrapped insertion-only Cuckoo filters. Solid lines represent adversarial guarantees ($\log\Pr[FP] \ge \log\Pr[\dD(\advA) = 1]$). Dashed lines represent the values obtained assuming NAI ($\log\Pr[FP] = \log\opfp{n+q_u}$).}\label{fig:main-body-correct-instances}
\end{figure}

\section{Conclusions}

We have introduced a framework for analysing the correctness (under adversarial input) and privacy of AMQ-PDS. We employed a simulation-based approach, with correctness and privacy emerging through imposing different constraints on the simulators. We have applied our approach to study Bloom and insertion-only Cuckoo filters, showing how they may be securely instantiated and highlighting the cost of adding security over unprotected instances. 

Our work lays a foundation for further study and analysis. Some important topics for future investigation include:
\begin{itemize*}
\item How tight are the bounds we provide? Are there matching attacks or can tighter bounds be proven? This is important since using our bounds to set concrete AMQ-PDS parameters incurs overhead in storage.
\item Can our Cuckoo filter analysis be extended to allow deletion of elements? This would require extension of our syntax and models, as well as careful modification of the consistency rules. 
\item 
Can our simulation-based approach be extended to other classes of PDS? Does ``wrapping'' the inputs of a PDS using a PRF always work as a protection? Investigating this would require developing a general syntax for PDS beyond that in our work and in~\cite{_CCS:ClaPatShr19}.
\item We focused on a strong adversarial model, where the adversary can access the \amq's state and has full adaptivity in its queries, necessitating the use of PRFs to achieve security. Do weaker primitives, e.g.\ UOWHFs, suffice in weaker settings, such as if the state is not accessible to the adversary? For example, the results of~\cite{_CCS:ClaPatShr19} suggest that salted hashes may indeed suffice.
\item 
While we considered an adversary attacking an honest service provider via an API, one may also ask: how can a user of the API obtain guarantees about the accuracy of the service being provided? 
\item How are our simulation-based correctness notions and the game-based ones in~\cite{_CCS:ClaPatShr19} related? Can we show a form of equivalence akin to that between semantic and IND-CPA security?
\end{itemize*}

We conclude by remarking that cryptographic tools and thinking seem to be broadly applicable to the problem of understanding the behaviour of PDS in adversarial settings. This topic is currently relatively under-researched, but is of growing importance in view of how rapidly PDS are being adopted in practice.

\begin{acks}
	The work of Fili\'{c} was supported by the Microsoft Research Swiss Joint Research Center.
	The work of Paterson was supported in part by a gift from VMware. 
	The work of Virdia was supported by the Zurich Information Security and Privacy Center. 
\end{acks}

\bibliographystyle{ACM-Reference-Format}
\ACM{
\balance
}{}
\bibliography{local}

\fullversion{
\ACM{
	\appendix
}
{
\IEEEtrans{
\appendices
}{
\clearpage
\thispagestyle{empty}
~
\vfill
\begin{center}
	\LARGE{\bfseries Supplementary Material}
\end{center}
\vfill
~
\clearpage
\normalsize
\appendix
}
}

\section{Statistical distance}\label{app:stat-distance}

\begin{definition}[Statistical distance]\label{def:statistical-distance}\label{def:statistically-close}
	Let $X$ and $Y$ be random variables with finite support $D = \supp{X} = \supp{Y}$. We define the statistical distance $\statdist$ between $X$ and $Y$ as $\statdist(X,Y) \coloneqq \frac{1}{2}\sum_{z \in D}\abs{\Pr[X=z]-\Pr[Y=z]}.$
	If $SD(X,Y) \le \varepsilon$, we say $X$ and $Y$ are $\varepsilon$-statistically close.
\end{definition}

\begin{lemma}[Data Processing Inequality]\label{lem:data-processing}
	Let $X$ and $Y$ be random variables with finite support $D = \supp{X} = \supp{Y}$.
	\begin{enumerate}[leftmargin=1.5em]
		\item Let $f \colon D \rightarrow R$ be a function such that $f(X)$ and $f(Y)$ are random variables with finite support $R$. Then $\statdist(f(X), f(Y)) \le \statdist(X,Y)$.
		\item Let $Z$ be a random variable with finite support $S$ pairwise independent from $X$ and $Y$, and let $g \colon D \times S \rightarrow R$ be a function such that $g(X,Z)$ and $g(Y,Z)$ are random variables with finite support $R$. Then $\statdist(g(X,Z), g(Y,Z)) \le \statdist(X,Y)$.
	\end{enumerate}
\end{lemma}
\begin{proof}
	\begin{enumerate}[leftmargin=1.5em]
		\item Let $f^{-1}(w) \coloneqq \{z \in D \mid f(z) = w)\}$.
		By direct computation,
		\begin{align*}
			\statdist&(f(X), f(Y)) = \frac{1}{2}\sum_{w \in R} \Big|\Pr[f(X) = w] - \Pr[f(Y)=w]\Big|\\
			&= \frac{1}{2}\sum_{w \in R} \Big|\Pr[X \in f^{-1}(w)] - \Pr[Y \in f^{-1}(w)]\Big|\\
			&= \frac{1}{2}\sum_{w \in R} \left|\sum_{z \in f^{-1}(w)}\Pr[X = z] - \Pr[Y = z]\right|\\
			&\le \frac{1}{2}\sum_{w \in R} \sum_{z \in f^{-1}(w)} \Big|\Pr[X = z] - \Pr[Y = z]\Big|\\
			&= \frac{1}{2}\sum_{z \in D} \Big|\Pr[X = z] - \Pr[Y = z]\Big| = \statdist(X,Y). 
		\end{align*}
		\item We start by using the proof above on the random variables $(X,Z)$ and $(Y,Z)$ defined over $D \times S$, such that $$\statdist(g(X,Z), g(Y,Z)) \le \statdist((X,Z), (Y,Z)).$$
		By assumption, $Z$ is pairwise independent from $X$ and $Y$, such that 
		\begin{align*}
			\Pr[(X,Z) = (x,z)] &= \Pr[X=x \wedge Z=z] = \Pr[X=x] \cdot \Pr[Z=z],
		\end{align*}
		and similarly $\Pr[(Y,Z) = (y,z)] = \Pr[Y=y] \cdot \Pr[Z=z].$
		Then, by direct computation,
		\begin{align*}
			\statdist&((X,Z), (Y,Z)) = \frac{1}{2}\sum_{s \in S} \sum_{d \in D} \left|
			\begin{matrix*}
				\Pr[(X,Z) = (d,s)] \\
				\hspace{1em} - \Pr[(Y,Z) = (d,s)]
			\end{matrix*}
			\right|\\
			&= \frac{1}{2}\sum_{s \in S} \sum_{d \in D} \Big|\Pr[X=d]\cdot\Pr[Z=s] - \Pr[Y=d]\cdot\Pr[Z=s]\Big|\\
			&= \frac{1}{2}\sum_{s \in S} \Pr[Z=s] \sum_{d \in D} \Big|\Pr[X=d] - \Pr[Y=d]\Big|\\
			&= \frac{1}{2} \sum_{d \in D} \Big|\Pr[X=d] - \Pr[Y=d]\Big| = \statdist(X,Y). \ACM{}{\queedee}
		\end{align*}
	\end{enumerate}
\end{proof}

\begin{lemma}\label{lem:data-processing-close}
	Let $X$ and $Y$ be $\varepsilon$-statistically close random variables with finite support $D = \supp{X} = \supp{Y}$, let $Z$ be a random variable with finite support $S$ pairwise independent from $X$ and $Y$. Let $f \colon D \rightarrow R$ and let $g \colon D \times S \rightarrow R$ be functions. Then $f(X)$ and $f(Y)$ are $\varepsilon$-statistically close random variables, and so are $g(X,Z)$ and $g(Y,Z)$.
\end{lemma}
\begin{proof}
	This follows directly from Def.~\ref{def:statistically-close} and Lemma~\ref{lem:data-processing}. \ACM{}{\qed}
\end{proof}

\section{Cuckoo filters}\label{app:cuckoo-filters}
\subsection{Algorithms}\label{app:cf-algorithms}
We first describe in more detail the functioning of insertion-only Cuckoo filters, as proposed in~\cite{_CoNEXT:FAKM14}. We refer the reader to Fig.~\ref{fig:cuckoo-filter-api} for the listings of the $\setupS$, $\upS$ and $\qryS$ algorithms.

The $\upS$ algorithm evaluated on an element $x$ uses $H_T$ to compute its \emph{tag} or \emph{fingerprint} $H_T(x)$, and $H_I$ and its tag to compute a pair of bucket indices $i_1=H_I(x),\ i_2 = i_1 \oplus H_I(H_T(x))$. If one of the two buckets $\sigma_{i_1}$, $\sigma_{i_2}$ is not full, the tag is added to it and insertion is complete. If both buckets are full, one of the two indices is picked at random, $i \getsr \{i_1, i_2\}$, together with a random slot $z \in [s]$. The tag $\tau$ currently stored in the slot $\sigma_i[z]$ is \emph{evicted}, and $H_T(x)$ is inserted in its place. Tag $\tau$ is then inserted into $\sigma_{i \oplus H_I(\tau)}$, which happens to be 
its other valid bucket (due to the relation between $i_1$ and $i_2$). This may require a further eviction and insertion of the evicted element, if $\sigma_{i \oplus H_I(\tau)}$ is also full. Up to $num$ such evictions are performed. If after $num$ evictions there is still another tag $\tau'$ to be relocated, 
the exact behaviour is not specified in~\cite{_CoNEXT:FAKM14}. The pseudocode for the insertion procedure presented in~\cite{_CoNEXT:FAKM14} suggests that $\tau'$ would be dropped from the filter, potentially turning it into a false negative element. However, the reference implementation published by the authors of~\cite{_CoNEXT:FAKM14}\footnote{\url{https://github.com/efficient/cuckoofilter/blob/917583d6abef692dfa8e14453bd77d6e0b61eef3/src/cuckoofilter.h\#L139}} stores $\tau'$ in a special stash $\sigma_{evic}$, and disables further insertions to the filter, 
preventing a false negative. 
We adopt the latter option to simplify our analysis. 
We note that while such an event would cause further insertions to fail (while leaving the state unchanged), 
Cuckoo filter parameters should be chosen as to make this be a low probability event,\footnote{In~\cite{_CoNEXT:FAKM14} the authors lower-bound this probability and then investigate it experimentally, however their argument does not apply to insertion-only Cuckoo filters where no duplicate tags are inserted in a bucket. The analysis for computing this probability should relate to the analysis for estimating the probability of bad insertions in Cuckoo Hashing, which has received rigorous analysis~\cite{DBLP:journals/jal/PaghR04,DBLP:journals/ipl/DevroyeM03,DBLP:journals/talg/DrmotaK12}.} meaning that in practice many calls to $\upS$ can be performed before insertions fail. Such a limit in the number of usable $\upS$ calls is also implied in practice for Bloom filters, since too many insertions raise the false positive probability beyond acceptable levels.
As suggested in~\cite{_CoNEXT:FAKM14} when describing insertion-only Cuckoo filters, we slightly modify the insertion algorithm so as to not insert duplicates of the same tag into the buckets $i_1,\ i_2$.

The $\qryS$ algorithm computes $i_1,\ i_2$ and checks if the element's tag is stored in either bucket, or if the element has been previously stashed.

\begin{figure*}[t]
	\Wider[10em]{
		\centering
		\begin{pchstack}[boxed,space=0.5em]
			\begin{pcvstack}[space=.2em]
					\procedure[linenumbering,headlinecmd={\vspace{.1em}\hrule\vspace{.3em}}]{$\setupS(pp)$}{%
							s, \lambda_I, \lambda_T, num \gets pp\\
							\pcmycomment{Initialise $2^{\lambda_I}$ buckets, $s$ $\lambda_T$-bit slots}\\
							\pcfor i \in 2^{\lambda_I}: \ \sigma_i \gets \bot^s\\
							\sigma_{evic} \gets \bot \\
							\pcreturn \sigma \gets (\sigma_i)_i, \sigma_{evic}
						}
					\procedure[linenumbering,headlinecmd={\vspace{.1em}\hrule\vspace{.3em}}]{$\qryS^{H_T, H_I}(x, \sigma)$}{%
							tag \gets H_T(x) \\
							i_1 \gets H_I(x) \\
							i_2 \gets i_1 \xor H_I(tag) \\
							a \gets [tag \in \sigma_{i_1}\ \text{or}\ tag \in \sigma_{i_2}\ \text{or}\ tag = \sigma_{evic}]\\
							\pcreturn a 
						}
				\end{pcvstack}
				\procedure[headlinecmd={\vspace{.1em}\hrule\vspace{.3em}}]{$\upS^{H_T,H_I}(x, \sigma)$}{
						\begin{subprocedure}
								\procedure[linenumbering]{}{
										tag \gets H_T(x) \\
										i_1 \gets H_I(x) \\
										i_2 \gets i_1 \xor H_I(tag) \\
										\pcmycomment{check if $\upS$ was disabled, first} \\
										\pcif \sigma_{evic} \ne \bot : \pcreturn \bot,\, \sigma\\
										\pcmycomment{if tag is already in either bucket} \\ 
										\pcif tag \in \sigma_{i_1}\ \text{or}\ tag \in \sigma_{i_2} : \pcreturn \top,\, \sigma \label{cf:alg:checkifduplicate}\\
										\pcmycomment{check if any bucket has empty slots} \\
										\pcfor i \in \{i_1, i_2\}\ \pcmycomment{in that order}\\
										\t \pcif \Load{\sigma_{i}} < s\\
										\t \t \sigma_{i} \gets \sigma_{i}\ \diamond\ tag\\
										\t \t \pcreturn \top,\, \sigma\\
										\pcmycomment{if no empty slots, displace something} 
									}
							\end{subprocedure}
						\>
						\begin{subprocedure}
								\procedure[lnstart=13,linenumbering]{}{
										i \getsr \{i_1, i_2\} \\
										\pcfor g \in [num] \\ 
										\t slot \getsr [s] \\ 
										\t elem \gets \sigma_{i,slot}\ \pcmycomment{element to be evicted} \\
										\t \pcmycomment{swap elem and tag} \\ 
										\t \sigma_{i}[slot] \gets tag;\ tag \gets elem \\ 
										\t i \gets i \xor H_I(tag) \label{line:later-call-to-H_T-in-cf}\\
										\t \pcif \Load{\sigma_i} < s \\
										\t \t \sigma_i \gets \sigma_i\ \diamond\ tag\\
										\t \t \pcreturn \top,\, \sigma \\
										\pcmycomment{could not store $x$ without an eviction}\\
										\sigma_{evic} \gets tag\ \pcmycomment{last value of $tag$ after loop}\\
										\pcreturn \top,\, \sigma
									}
							\end{subprocedure}
					}
			\end{pchstack}
		}
	\caption{\amq syntax instantiation for the Cuckoo filter.}
	\label{fig:cuckoo-filter-api}
\end{figure*}

\subsection{Proof of Lemma \ref{clm:cf-variants-fp-bound}}\label{app:proof-cf-variants-fp-bound}
\prfcfpfp*
\begin{proof}
By Def.~\ref{def:nai-fpp}, $\pfp{n}$ is computed assuming that the elements $x_i \in \domain$ that are inserted via $\upS$ and the elements queried via $\qryS$ are distinct.
PRF-wrapped Cuckoo filters can be seen as Cuckoo filters where the elements being inserted are the output of a PRF evaluated on such distinct elements.
In particular, in the $\Ideal$ world where $\pfp{n}$ is to be computed, we replace our PRF with a random function $F \getsr \Funcs[\domain,\range]$ (i.e.~we set $D_F = U\left(\Funcs[\domain,\range]\right)$ following the notation in Def.~\ref{def:nai-fpp}).

We proceed to compute the desired bound. Let $S = \{\tau_1, \dots, \tau_{2s+1}\}$ be a set of $2s+1$ strings in $\{0,1\}^{\lambda_T}$ such that $\tau_i = H_T(F(x_i))$ for $i \in [2s+1]$, and let $y$ be a string in $\{0,1\}^{\lambda_T}$ such that $y = H_T(F(x))$. Following the analysis of~\cite{_CoNEXT:FAKM14}, we consider $H_T$ to be a random function, such that the $\tau_i$ and $y$ are uniformly distributed in $\{0,1\}^{\lambda_T}$.
We then have
\begin{align*}
	&\pfp{n} \le \opfp{n} = \Pr[y \in S] \enspace \text{by the analysis done in~\cite{_CoNEXT:FAKM14}} \\
	&\ = P[y \in S \wedge (F(x),\,F(x_1),\dots,F(x_{2s+1})\ \text{are distinct})]\\
	&\ \quad + P[y \in S \wedge (F(x),\,F(x_1),\dots,F(x_{2s+1})\ \text{are not distinct})]\\
	&\ \le \overline{P_{\Pi',pp}}(FP \mid {n}) + P[F(x),\,F(x_1),\dots,F(x_{2s+1})\ \text{are not distinct}]
\end{align*}
Since $F$ is a random function and by Def.~\ref{def:nai-fpp} the elements $x$, $x_1$, \dots, $x_{2s+1}$ are distinct, $P[F(x),\,F(x_1),\dots,F(x_{2s+1})\ \text{are not distinct}]$ is the probability that there is a collision in a collection of $2s+2$ uniformly random strings sampled from $\range$.
By the birthday bound, this probability is at most $\frac{(2s+2)^2}{2\,|\range|}$, hence giving the result. \ACM{}{\qed}
\end{proof}

\subsection{Permutation invariance of original Cuckoo filters}
\label{app:cf-not-pi}

We show that original Cuckoo filters (Fig.~\ref{fig:up-cuckoo-not-f-decomposable}) do not satisfy 0-PI, by describing a distinguishing attack.

Consider $\gamePI$ in Fig.~\ref{g:priv:PI:nn}, where $\Pi$ is an insertion-only original Cuckoo filter. In order to insert or query an element, its corresponding tag and buckets are computed as follows when $c = 0$:
\begin{align*}
tag  = H_T (x); \quad
i_1  = H_I (x); \quad
i_2  = i_1 \xor H_I (tag) = i_1 \xor H_I (H_T (x) ).
\end{align*}
On the other hand, when $c = 1$, we have:
\begin{align*}
tag  = H_T (\pi (x)); \enspace
i_1  = H_I ( \pi (x)); \enspace
i_2  = i_1 \xor H_I (tag) = i_1 \xor H_I (H_T (\pi (x) )).
\end{align*}
It is then possible to distinguish between the $c = 0, 1$ versions of $\gamePI$ in the following way. 

Suppose $c = 0$. By inserting $x$ and then making a $\revO$ query, we learn its tag $H_T(x)$. We then make a sequence of pairs of $\upO$ and $\revO$ queries on random inputs. Eventually, $x$ will be moved between buckets as a result of eviction (and we can tell when this happens from the output of the $\revO$ queries). By noting which buckets $i_1, i_2$, it moves between, we can compute $i_1 \xor i_2 = H_I (H_T(x))$. This tells us what the first bucket would be if we inserted the element $H_T(x)$, i.e. if we made an $\upO(H_T(x))$ query. 

Now suppose $c = 1$. By inserting $x$ we learn its tag $H_T(\pi(x))$, and by noting which buckets $i_1, i_2$ it moves between (in the same way as before), we can compute $i_1 \xor i_2 = H_I (H_T ( \pi(x)))$. However, with overwhelming probability this no longer corresponds to the first bucket if we inserted the element $H_T(\pi(x))$. 

Thus, by inserting an element and then later inserting its tag, one can distinguish between the real and ideal worlds.

\section{Correctness}\label{app:correctness} 

\subsection{Proof of Theorem \ref{thm:general-correctness}}
\label{app:correctness-proofs}

We first prove a utility lemma, which we will use to prove that the view of $\Pi$ simulated by $\simS$ in $\Ideal$ is NAI for insertion-only \amq that satisfy function-decomposability, reinsertion invariance and permanent disabling. 

\begin{lemma}
\label{lem:bf-utility-lemma}\label{lem:utility-lemma}
	Let $F \getsr \Funcs[\domain,\range]$ be a random function, and let $\Pi$ be an $F$-decomposable \amq with public parameters $pp$ and an oracle access to function $F$.
	Let $\varepsilon \ge 0$, let $i$ be a positive integer and let $\sigma^{(i-1)}$ be a random variable representing the output of an algorithm returning $(i{-}1, \varepsilon)$-NAI state for $\Pi$.
	Define random variables $r$, $Y$ and $\sigma$ such that $r \sim U(\coinspace)$,  $Y \sim U(\range)$, $r$, $Y$ and $\sigma^{(i-1)}$ are pairwise independent, and $(b, \sigma) \coloneqq \upS^{\identity_\range}(Y, \sigma^{(i-1)}; r)$.
	Then $\sigma$ is $(i, \varepsilon)$-NAI.
\end{lemma}
\begin{proof}
	Let $\overline{\sigma}^{(i-1)}$ be a random variable representing the output of ($i{-}1$)-$\naigen^F(pp)$, executed independently from $Y$, $r$ and $\sigma^{(i-1)}$.
	By assumption $\sigma^{(i-1)}$ is $(i{-}1, \varepsilon)$-NAI, hence $$\statdist(\sigma^{(i-1)}, \overline{\sigma}^{(i-1)}) \le \varepsilon.$$
	
	Define a random variable $\overline{\sigma}$ such that $(\overline{b}, \overline{\sigma}) \coloneqq \upS^{\identity_\range}(Y,\overline{\sigma}^{(i-1)}; r)$.
	Since $Y$ and $r$ are pairwise independent from $\sigma^{(i-1)}$ and $\overline{\sigma}^{(i-1)}$, we can invoke Lemma~\ref{lem:data-processing-close} with $g(x, (Y, r)) = \upS^{\identity_\range}(Y,\, x; r)$, evaluated on $x = {\sigma^{(i-1)}}$ and $\overline{\sigma}^{(i-1)}$. This implies that $\sigma$ and $\overline{\sigma}$ are $\varepsilon$-statistically close.
	
	To prove that $\sigma$ is $(i, \varepsilon)$-NAI, all that is now required is arguing that $\overline{\sigma}$ has the same distribution as the output of $i$-$\naigen^F(pp)$.
	Following the notation in Fig.~\ref{fig:NAI-gen}, we start by noticing that given a specific output of $(i{-}1)$-$\naigen^F(pp)$ such as $\overline{\sigma}^{(i-1)}$ above, one can obtain a state $\overline{\sigma}^{(i)}$ following the distribution of the output of $i$-$\naigen^F(pp)$ by sampling a distinct $x_i \in \domain$ not in the set $\{x_1,\dots,x_{i-1}\}$ of elements sampled inside $(i{-}1)$-$\naigen^F(pp)$ (line~\ref{line:sample-in-NAI-gen} of 
	Fig.~\ref{fig:NAI-gen}), and evaluating $(b', \overline{\sigma}^{(i)}) \getsr \upS^{F}(x_i, \overline{\sigma}^{(i-1)})$ as in line~\ref{line:up-in-NAI-gen} of 
	Fig.~\ref{fig:NAI-gen} (for $j \gets i$ and where $\upS$ uses only one oracle $F$). In particular, since by assumption $\Pi$ is 
	$F$-decomposable, we can rewrite $(b', \overline{\sigma}^{(i)}) \gets \upS^{\identity_\range}(F(x_i), \overline{\sigma}^{(i-1)}; r')$ for  freshly sampled coins $r' \sim U(\coinspace)$.
	
	\begin{figure}[h]
		\centering
		\begin{tikzpicture}[node distance=1.35cm, auto]
			\node (P) {$\sigma^{(i-1)}$};
			\node (B) [right of=P, xshift=2cm] {$\overline{\sigma}^{(i-1)}$};
			\node (A) [below of=P] {$\sigma$};
			\node (C) [below of=B] {$\overline{\sigma}$};
			\node (BB) [right of=B, xshift=.7cm] {$\sim (i{-}1){\text{-}}\naigen(pp)$};
			\node (CC) [right of=C] {$\sim i{\text{-}}\naigen(pp)$};
			\draw[<->] (P) to node {$\statdist \le \varepsilon$} (B);
			\draw[->] (P) to node [swap] {$\upS^{\identity_\range}(Y,\ \cdot\ ; r)$} (A);
			\draw[<->] (A) to node {$\statdist \le \varepsilon$} (C);
			\draw[->] (B) to node {$\upS^{\identity_\range}(Y,\ \cdot\ ; r)$} (C);
		\end{tikzpicture}
		\caption{Diagram of relations between random variables in Lemma~\ref{lem:utility-lemma}.}
	\end{figure}
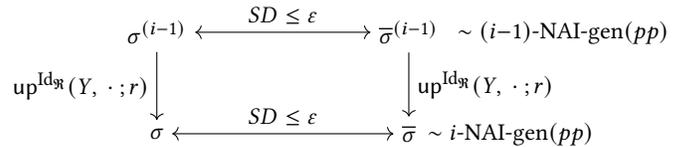
	
	Since the $\{x_1,\dots,x_i\}$ above are distinct and $F$ is a random function, $F(x_j)$ is uniformly distributed over $\range$ and pairwise independent of $\sigma^{(i-1)}$, $\overline{\sigma}^{(i-1)}$ and $r'$. Since $r' \sim r$ and $F(x_j) \sim Y$, we see that $\overline{\sigma}^{(i)}$ and $\overline{\sigma}$ have the same distribution. By $\varepsilon$-statistical closeness of $\sigma$ and $\overline{\sigma}$, $\sigma$ is then $(i, \varepsilon)$-NAI.
	\ACM{}{\qed}
\end{proof}

We now proceed to prove Theorem \ref{thm:general-correctness}.

\correctnessthm*

\begin{proof}
	We start by defining an intermediate game $G$ in Fig.~\ref{game:r-o-g}. Let $\Real$ denote the $d = 0$ version of Real-or-$G$, let $G$ denote the $d = 1$ version of Real-or-$G$ (or equivalently the $d = 0$ version of $G$-or-Ideal), and let $\Ideal$ denote the $d = 1$ version of $G$-or-Ideal.
		
	Our proof then proceeds in the following way. We first bound the closeness of $\Real, G$ in Lemma \ref{l:r-o-g} in terms of the PRF advantage, and that of $G, \Ideal$ in Lemma \ref{l:g-o-i} in terms of the probability of some ``bad" event $\E$. Then, in Lemma \ref{lem:ideal-is-NAI} we show that our simulator constructs an NAI view of $\Pi$ in $\Ideal$. Finally, we bound the probability of the event $\E$ in Lemma \ref{lem:E-upper-bound} to obtain our result.
	
		\begin{figure}[h]
		\Wider[30em]{
			\centering
			\begin{pchstack}[boxed,space=.2em]
				\procedure[linenumbering,headlinecmd={\vspace{.1em}\hrule\vspace{.3em}}]{Real-or-$G(\advA, \dD, pp)$}{%
							d \getsr \{0, 1\} \\ 
							\pcif d = 0 \t \pcmycomment{\Real} \\
							\t K \getsr \mathcal{K}; F \gets R_K \\ 
							\pcelse \t \pcmycomment{$G$} \\
							\t F \getsr \Funcs[\domain,\range] \\
							\tinit \gets \bot \\
							\sigma \gets \setupS(pp) \\
							out \getsr \advA^{\repO, \upO, \qryO,\revO} \\ 
							d' \getsr \dD(out) \\ 
							\pcreturn d'
				}
				\procedure[linenumbering,headlinecmd={\vspace{.1em}\hrule\vspace{.3em}}]{$G$-or-Ideal$(\advA, \simS, \dD, pp)$}{%
							d \getsr \{0, 1\} \\ 
							\pcif d = 0 \t \pcmycomment{$G$} \\ 
							\t F \getsr \Funcs[\domain,\range] \\
							\t \tinit \gets \bot \\
							\t \sigma \gets \setupS(pp) \\
							\t out \getsr \advA^{\repO, \upO, \qryO,\revO} \\ 
							\pcelse \t \pcmycomment{\Ideal} \\
							\t out \getsr \simS(\advA, pp) \\ 
							d' \getsr \dD(out) \\ 
							\pcreturn d'
				}
			\end{pchstack}
		}
		\caption{Intermediate game $G$ for the proof of Theorem~\ref{thm:general-correctness}.
		}
		\label{game:r-o-g}
		\label{game:g-o-i}
	\end{figure}
	
	\begin{lemma}\label{l:r-o-g}
		The difference in probability of an arbitrary $t_d${-} distinguisher $\dD$ outputting $1$ in experiments of game Real-or-$G$ in Fig.~\ref{game:r-o-g} with a $(q_u, q_t, q_v, t_a)$-\amq adversary $\advA$ is bounded by the maximal PRF advantage $\varepsilon$ of an $(\alpha(n + q_u) + \beta q_t, t_a + t_d, \varepsilon)$-PRF adversary attacking $R_K$: 
		\begin{align*}
			\text{Adv}^{\text{Real-or-}G}_{\Pi,\advA,\simS}(\dD) \coloneqq \left|\pr{\Real(\advA, \dD){=}1}{-}\pr{G(\advA, \dD){=}1}\right| \leq \varepsilon.
		\end{align*}
	\end{lemma}
	\begin{proof}
		Consider the PRF adversary $\advB$ in Fig.~\ref{fig:prf-adv}, who instantiates the \amq that $\advA$ queries using its $\RRO$ oracle, in relation to the Real-or-$G$ game from Fig.~\ref{game:r-o-g}. 
		\begin{figure}
			\centering
			\begin{pcvstack}[boxed,space=1em]
				\procedure[linenumbering,space=1em,headlinecmd={\vspace{.1em}\hrule\vspace{.3em}}]{PRF adversary $\advB^\RRO$}{%
					F \gets \RRO; \, \tinit \gets \bot\\
					\sigma \gets \setupS(pp) \\
					\pcreturn d' \getsr \dD(\advA^{\repO, \upO, \qryO,\revO})
				}
			\end{pcvstack}
			\caption{PRF adversary $\advB$ for Lemma~\ref{l:r-o-g}.}
			\label{fig:prf-adv}
		\end{figure}
		
		When $b = 0$, $\advB$ is running $\Real$ for $\advA$, where the PRF $R_K$ is used to handle $\advA$'s oracle queries to $\Pi$. When $b = 1$, $\advB$ is instead running $G$ for $\advA$, where the truly random function $F$ is used to handle $\advA$'s oracle queries to $\Pi$. By inspection, the advantage of $\advB$ is
		\begin{align*}
			\text{Adv}_R^{PRF}(\advB) = \text{Adv}^{\text{Real-or-}G}_{\Pi,\advA,\simS}(\dD).
		\end{align*}
		By assumption in Theorem~\ref{thm:general-correctness}, $R_K$ is an ($\alpha(n + q_u) + \beta q_t, t_a + t_d, \varepsilon$)-secure PRF, hence no adversary $\advB$ making at most $\alpha(n+q_u) + \beta q_t$ queries and running in time at most $t_a + t_d$ can have advantage greater than $\varepsilon$. Therefore, 
		\begin{align*}
			\text{Adv}^{\text{Real-or-}G}_{\Pi,\advA,\simS}(\dD) 
			\leq \varepsilon. \tag*{} 
		\end{align*}
	\end{proof}
	
	\begin{lemma}\label{l:g-o-i}
	For $i \in [q_t]$, let $a_i^{\Ideal}$ be the response to $\advA$'s $i^{th}$ $\qryO$ query in the Ideal game, let $a_i^{G^*}$ be the response in the $G^*$ game (see~Line~\ref{line:corr-E*} of Fig.~\ref{fig:correctness-simulator}), and let $\E$ be the event that these differ for some~$i$:
	\begin{align*}
	\E := \big[a_i^{\Ideal} \neq a_i^{G^*}\ \text{for some}\ i \in [q_t]\big].
	\end{align*}
		The difference in probability of an arbitrary distinguisher $\dD$ outputting $1$ in experiments of game $G$-or-Ideal with a $(q_u ,q_t, q_v, t_a)$-\amq adversary $\advA$ is bounded by $\Pr[\E]$. In other words,
		\begin{align*}
			\text{Adv}^{G\text{-or-\Ideal}}_{\Pi,\advA,\simS}(\dD){\coloneqq}\big|\Pr[G(\advA, \dD){=}1]{-}\Pr[\Ideal(\advA, \dD,\simS){=}1]\big|{\leq}\Pr[\E].
		\end{align*}
	\end{lemma}
	\begin{proof}	

		First, let us denote by $G^*$ a modified version of $\Ideal$ that runs $\simS$, but with the following modification: in the oracle \qrysimO, it always replies $a \gets \qryS^F(x_i, \sigma)$, instead of sampling a random $a \leftarrow \qry^{\identity_\range} (Y \getsr \range, \sigma)$ for every query on a non-inserted element.
		The main difference between $G$ and $G^*$ is that in the latter, $\advA$ interacts with a simulator that intercepts some of the queries to the simulated \amq, and does some input-output bookkeeping absent in $G$.
		However, by inspection of Fig.~\ref{fig:correctness-simulator}, the extra operations run by $\simS$ in $\repsimO$ and $\upsimO$ in $G^*$, compared to $\repO$ and $\upO$ in $G$, do not affect the return values (in particular, reinsertion invariance and permanent disabling of $\Pi$ means that line~\ref{line:reinsertion-invariant-return} of $\upsimO$ agrees with $\upO$).
		Hence, the only possible difference from the point of view of $(\advA, \dD)$ could come from $\qrysimO$. 
		Inspecting $\qrysimO$, we can see that lines~\ref{line:elem-permanence-1}-\ref{line:elem-permanence-2} will not cause a discrepancy due to the assumption that $\Pi$ satisfies Element permanence (see Def.~\ref{def:bf-consistency-rules}). Similarly, lines~\ref{line:qry-determinism-1}-\ref{line:qry-determinism-2} will not cause discrepancies since they will only be executed if no insertions were made since the last call to $\qrysimO$ on the current element, and the previous call return $\bot$, in which case $\qrysimO$ will return $\bot$ again. Since $\qryS^F$ is a deterministic algorithm, also $\qryO$ in $G$ would return $\bot$ in this scenario. Since no other operations in $\qrysimO$ affect its return value, the game $G^*$ is identical to $G$ (Fig.~\ref{game:g-o-i}) from the point of view of $(\advA, \dD)$.
		
		We now look at the differences between $G^*$ and $\Ideal$. The answers to $\repsimO$ and $\upsimO$ queries on the same inputs are generated in the same manner in both $\Ideal$ and $G^*$, and therefore do not lead to inconsistencies between the two games. The same reasoning applies for $\qrysimO$ queries on elements that are positive. However, for $\qrysimO$ queries on elements that are not yet positive, the answers will be the same across the games if and only if for each element $x_i$ queried by $\advA$'s $i$-th $\qrysimO$ query, we have 
		\begin{align}\label{eq:e*}
			a_i^{\Ideal} = a_i^{G^*} = \qryS^F(x_i, \sigma).
		\end{align}
		Therefore, the games $\Ideal$ and $G^*$ (and hence $G$) are equal from the perspective of $(\advA,\dD)$ at least up until the event that \eqref{eq:e*} does not hold for some $i \in [q_t]$. Recalling that this event is denoted by $\E$, we have
		\begin{align*}
			|\Pr[\Ideal(\advA, \dD,\simS) {=} 1] {-} \Pr[G(\advA, \dD) {=} 1]| \leq \Pr[\E]. 
		\end{align*}
	\end{proof}
\begin{lemma}
	\label{lem:bf-in-ideal-are-nai}\label{lem:amq-in-ideal-are-nai}\label{lem:ideal-is-NAI}
	Consider the $\Ideal$ game with simulator $\simS$ from Fig.~\ref{fig:correctness-simulator}. Suppose $i$ distinct elements $I = \{x_1, \dots, x_i\}$ are queried to $\upsimO$ (including those queried as part of the $\repsimO$ query) with $x_j$ being the $j$-th distinct element being queried.
	After the first call to $\upsimO(x_i)$, call $\revealsimO()$ to obtain $\sigma$, the state of $\Pi$. Then $\sigma$ is $(i, \varepsilon)$-NAI, where $\varepsilon = 0$.
\end{lemma}
\begin{proof}
We now prove that the view constructed by $\simS$ in $\Ideal$ is NAI for $\Pi$.
	
	The proof works by induction.
	Suppose that $j{-}1$ distinct elements have been queried to $\upsimO$ by $\advA$ at some point of $\simS$'s simulation (including elements queried as part of $\repsimO$). Let $\Sigma$ be the set of possible states of $\Pi$.
	Suppose $\advA$ makes a new query $\upsimO(x)$. If $x \in \{x_1,\dots,x_{j-1}\}$ and previously $\upsimO(x)$ returned $\top$, then by construction of $\upsimO$ (which simulates reinsertion invariance), $\Pi$'s state will remain unchanged after the current $\upsimO(x)$ query. Otherwise, if $x \in \{x_1,\dots,x_{j-1}\}$ and previously $\upsimO(x)$ returned $\bot$, by permanent-disabling of $\Pi$, a further $\upsimO(x)$ query will leave $\Pi$'s state similarly unchanged. Hence, only $\upsimO(x)$ queries where $x \notin \{x_1,\dots,x_{j-1}\}$ should be considered during the induction.
	
	Let $Q_j$ be the set $\{y_1, \dots, y_{q_j}\} \subset \domain$ of distinct elements queried to $\qrysimO$ by $\advA$ during $\Ideal$ up until the first $\upsimO(x_j)$ query for a given $x_j$ with $j \le i$ is made.
	Let $\sigma^{(j-1)}$ (resp.~$\sigma^{(j)}$) be a random variable describing the output of $\revealsimO()$ just before (resp.~after) the first $\upsimO(x_j)$ query.

	Before any $\upsimO$ or $\repsimO$ queries are made in $\Ideal$, the output of $\revealsimO()$ is $(0, 0)$-NAI, and hence $\sigma^{(j)}$ is $(j, \varepsilon)$-NAI for $j = 0$.
	
	Suppose $\sigma^{(j-1)}$ is $(j{-}1, \varepsilon)$-NAI, and let $x_j$ be the $j^{\text{th}}$ distinct element being queried to $\upsimO$.
	Let $(b^{(j)}, \sigma^{(j)}) \gets \upS^{F}(x_j, \sigma^{(j-1)}; r')$ be the $\upS$ call made inside $\upsimO(x_j)$, for some freshly sampled $r' \sim U(\coinspace)$.
	Since by construction the $\qrysimO$ oracle inside $\simS$ never uses or outputs the value of $F(x_j)$ in the $\Ideal$ game, and since $x_j$ has not been yet queried to $\upsimO$, we know that $\sigma^{(k)}$ for every $k < j$ and the output of $\qrysimO(y)$ for $y \in Q_k$ for every $k \le j$ are independent of the value of $F(x_j)$. Hence the value of $F(x_j)$ is uniformly distributed over $\range$ and independent from the state of $\Pi$, and therefore we can invoke Lemma~\ref{lem:utility-lemma} with $Y = F(x_j)$, $r = r'$ and $\sigma = \sigma^{(j)}$.
	By Lemma~\ref{lem:bf-utility-lemma}, then $\sigma^{(j)}$ is $(j, \varepsilon)$-NAI. Doing $i$ steps of the induction gives the result. \ACM{}{\qed}
\end{proof}

\begin{lemma}\label{lem:E-upper-bound}
Let the event $\E$ be defined as in Lemma \ref{l:g-o-i}.
Then, 
	\begin{align*}
	 \Pr[\E] 
	  \leq 2 q_t \cdot \pfp{n + q_u}.
	\end{align*}
\end{lemma}
\begin{proof}
	Our final step is to compute $\Pr[\E]$.
While $\E$ can only occur for elements $x_i$ that have not been inserted into $\Pi$ at the time they are queried to $\qrysimO$ (and hence could return a \emph{false} positive result), we do not know how many of $\advA$'s $q_t$ queries will be used on such elements. Thus, we will consider all $i \in [q_t]$ in order to bound $\Pr[\E]$.
We will calculate $\Pr[\E]$ in the game $\Ideal$, with $a_i^\Ideal$ and $a_i^{G^*}$ defined as in lines \ref{line:bernoulli} and \ref{line:aG} of $\qrysimO$ in Fig.~\ref{fig:correctness-simulator}.
Let $\sigma^{(i)}$ (resp.~$ctr^{(i)}$) denote the state of $\Pi$ instantiated by $\simS$ (resp.~number of inserted elements into $\Pi$) in $\Ideal$ at the time of $\advA$'s $i$-th $\qryO$ query.
We have
\begin{align*}
	&a_i^{G^*} \gets \qryS^F(x_i,\sigma^{(i)}), \enspace  
	\E_i \coloneqq \left[ a_i^{\Ideal} \neq a_i^{G^*} \right],\\
	&\enspace\enspace \E = \bigvee_{i = 1}^{q_t} \E_i, \enspace \text{and}\enspace \Pr[\E] \leq \sum_{i=1}^{q_t} \Pr[\E_i].
\end{align*} 
We proceed to bound $\Pr[\E_i]$. 
We have
\begin{align*}
	\Pr[\E_i] &= \Pr[a_i^{\Ideal} \neq a_i^{G^*}]\\
	&= \Pr[(a_i^{\Ideal} = \top \wedge a_i^{G^*} = \bot) \vee (a_i^{\Ideal} = \bot \wedge a_i^{G^*} = \top)]\\
	& \le \Pr[(a_i^{\Ideal} = \top) \vee (a_i^{G^*} = \top)] \\
	& \le \Pr[a_i^{\Ideal} = \top]   + \Pr[a_i^{G^*} = \top].
\end{align*}

Furthermore, we note that 
\begin{align*}
	&\underset{\text{$\advA$'s coins}}{\underset{F \getsr \Funcs[\domain,\range]}{\Pr}}[a_i^{G^*} = \top] \\
	&\quad = \underset{F \getsr \Funcs[\domain,\range]}{\Pr}[\top \gets \qryS^F(x_i,\sigma^{(i)})] \quad \begin{matrix*}[l] \text{\small by $\advA$'s choices being} \\ \text{\small nullified by $F$'s sampling} \end{matrix*} \\
	&\quad = \underset{F \getsr \Funcs[\domain,\range]}{\Pr}[\top \gets \qryS^{\identity_\range}(F(x_i), \sigma^{(i)})] \quad \begin{matrix*}[l] \text{\small by $F$-decomposability} \\ \text{\small of $\Pi$} \end{matrix*} \\
	&\quad= \underset{r \getsr \coinspace}{\underset{F \getsr \Funcs[\domain,\range]}{\Pr}}[\top \gets \qryS^{\identity_\range}(Y \overset{r}{\getsr} \range, \sigma^{(i)})] \quad \begin{matrix*}[l] \text{\small by $\advA$ having} \\ \text{\small no information on} \\ \text{\small the value of $F(x_i)$} \end{matrix*}\\
	&\quad=\underset{\text{$\advA$'s coins}}{\underset{r \getsr \coinspace}{\underset{F \getsr \Funcs[\domain,\range]}{\Pr}}}[a^{Ideal}_i = \top],
\end{align*}
meaning that  $\Pr[\E_i]\,{\le}\,\Pr[a_i^{\Ideal} {=} \top] + \Pr[a_i^{G^*} {=} \top] = 2 \Pr[a_i^{\Ideal} {=} \top].$

We now want to argue that 
\begin{align*}
	\Pr[a_i^\Ideal = \top] &= \sum_{c = 0}^{n + q_u} \Pr[a_i^\Ideal = \top \mid ctr^{(i)}=c] P[ctr^{(i)} = c]\\
	&= \sum_{c = 0}^{n + q_u} \pfp{c} P[ctr^{(i)} = c]
\end{align*}
Indeed, let $D_F = U(\Funcs[\domain,\range])$, then from Def.~\ref{def:nai-fpp}, we have
\begin{align*}
	&\pfpsym(FP \mid c) \coloneqq \Pr\left[\begin{matrix}\hat{F} \getsr D_{F}\\ \hat{\sigma} \getsr \text{$c$-$\naigen(pp)$}\\ \hat{x} \getsr U(\domain\setminus V)\end{matrix}:\ \top \gets \qry^{\hat{F}}(\hat{x},\, \hat{\sigma})\right]\\
	&\ = \Pr\left[\begin{matrix}\hat{F} \getsr D_{F}\\ \hat{\sigma} \getsr \text{$c$-$\naigen(pp)$}\\ \hat{x} \getsr U(\domain\setminus V)\end{matrix}:\ \top \gets \qry^{\identity_\range}(\hat{F}(\hat{x}),\, \hat{\sigma})\right] \quad \begin{matrix*}[l] \text{\small by function} \\ \text{\small decomposability}\end{matrix*}\\
	&\ = \Pr\left[\begin{matrix}\hat{F} \getsr D_{F}\\ \hat{\sigma} \getsr \text{$c$-$\naigen(pp)$} \end{matrix}:\ \top \gets \qry^{\identity_\range}(\hat{Y} \getsr \range,\, \hat{\sigma})\right] \quad \begin{matrix*}[l] \text{\small by $\hat{x} \notin V$ and} \\ \text{\small $\hat{F}$ being truly} \\ \text{\small random}\end{matrix*}\\
	&\ = \underset{r \getsr \coinspace}{\underset{F \getsr \Funcs[\domain,\range]}{\Pr}}\hspace{-1em}[\top \gets \qryS^{\identity_\range}(Y \overset{r}{\getsr} \range, \sigma^{(i)})\,|\,ctr^{(i)} {=} c] \enspace \begin{matrix*}[l] \text{\small in $\Ideal$, by $\sigma^{(i)}$} \\ \text{\small being $(c, 0)$-NAI} \\ \text{\small (see Lemma~\ref{lem:ideal-is-NAI})} \end{matrix*} \\
	&\ = \Pr[a^{Ideal}_i = \top \mid ctr^{(i)} = c],
\end{align*}
so that
\begin{align*}
	\Pr[\E_i] &\le 2 \Pr[a_i^{\Ideal} = \top] = 2 \sum_{c = 0}^{n + q_u} \pfp{c} P[ctr^{(i)} = c]\\
	&\le 2 \max_{c \in [n+q_u]}\pfp{c} \sum_{c = 0}^{n + q_u}  P[ctr^{(i)} = c]\\
	&\le 2 \max_{c \in [n+q_u]} \pfp{c}.
\end{align*}
Finally, we compute
\begin{align*}
	\Pr&[\E] \leq \sum_{i=1}^{q_t} \pr{ \E_i } \leq 2 q_t \cdot \max_{c \in [n+q_u]} \pfp{c} \\
	&\leq 2 q_t \cdot \pfp{n + q_u}. \ACM{}{\queedee} \quad \begin{matrix*}[l] \text{\small by non-decreasing membership} \\ \text{\small probability (see Def.~\ref{def:bf-consistency-rules})} \end{matrix*}
\end{align*}
\end{proof}
We are now ready to prove Theorem \ref{thm:general-correctness}. Combining Lemmas \ref{l:r-o-g}, \ref{l:g-o-i}, \ref{lem:ideal-is-NAI} and \ref{lem:E-upper-bound}, 
\begin{align*}
	\text{Adv}_{\Pi, \mathcal{A, S}}^{{\text{RoI}}} (\dD) & =  | \pr{\Real(\advA, \dD) {=} 1} {-} \pr{\Ideal(\advA, \dD,\simS) {=} 1} |\\
	& \leq  \abs{\Pr[\Real(\advA, \dD) {=} 1] {-} \Pr[G(\advA, \dD) {=} 1]} \\
	&\quad + \abs{\Pr[G(\advA, \dD) {=} 1] {-} \Pr[\Ideal(\advA, \dD,\simS) {=} 1]}\\
	& = \text{Adv}_{\Pi, \mathcal{A, S}}^{{\text{Real-or-}G}}(\dD)+\text{Adv}_{\Pi, \mathcal{A, S}}^{{G\text{-or-Ideal}}}(\dD) \\
	& \leq \varepsilon + \Pr[\E]\\
	& \leq \varepsilon +  2 q_t \cdot \pfp{n + q_u}. \ACM{}{\queedee}
\end{align*}
\end{proof}

\begin{remark}\label{rmk:utility-cf-caveats}
	While the proof and lemmas in this section only
	explicitly refer to \amqs with a single oracle function $F$ for notational simplicity, the same results hold in the case where more oracles $F_1, \dots, F_t$ are used. In that case, assuming $F_1$-decomposability, 
	one simply needs to replace sampling of function $F$ with sampling of functions $F_1, \dots, F_t$ at the beginning of games $\gameRoI$ (Fig.~\ref{game:r-o-i}), Real-or-$G$ and $G$-or-Ideal (Fig.~\ref{game:r-o-g}), in the PRF adversary used in the proof of Lemma \ref{l:r-o-g} (Fig.~\ref{fig:prf-adv}), and either give the simulator $\simS$ (Fig.~\ref{fig:correctness-simulator}) oracle access to $F_2,\dots,F_t$, or sampling them again also in $\simS$. Moreover, every $\upS^{F}, \upS^{\identity_\range}, \qryS^{F}, \qryS^{\identity_\range}$ should be replaced with $\upS^{F_1, \dots, F_t}$, $\upS^{\identity_\range, F_2, \dots, F_t}$, $\qryS^{F_1, \dots, F_t}$, $\qryS^{\identity_\range, F_2, \dots, F_t}$ respectively, and $\naigen^F$ with $\naigen^{F_1, \dots, F_t}$.
\end{remark}

\subsection{Comparison with bounds of \cite{_CCS:ClaPatShr19}}
\label{sec:comparison}
We compare our bounds to those derived in~\cite{_CCS:ClaPatShr19} using a game-based notion of correctness, in the case of Bloom filters. 
The game defined in~\cite{_CCS:ClaPatShr19}, which we will denote $G'$, exposes similar oracles to those in Fig.~\ref{game:r-o-i}, the main difference being that $\revO$ is embedded into $\repO$ and $\upO$. However, any adversary $\advA'$ in $G'$ in the public scenario can be translated to an adversary $\advA$ for $\gameRoI$, by following every $\repO,\upO$ query with $\revO$, and having $\advA$ output the transcript. 
We refer to $\advA$ as the simulation-based equivalent of $\advA'$.

In $G'$, the adversary's win condition is causing the event $S_r \coloneqq [\advA$ makes $\qryO(\cdot)$ queries resulting in at least $r$ false positives$]$.
Letting $\Pr_{\advB,G}[X]$ be the probability of event $X$ happening in game $G$ with adversary $\advB$, Clayton~\emph{et al.}~\cite{_CCS:ClaPatShr19} derive upper bounds for  $\Pr_{\advA', G'}[S_r]$.
On the other hand, if $\dD_{X}$ is a distinguisher in $\gameRoI$ for an arbitrary property $X$, our simulation-based notion allows us to upper bound
\begin{equation}\label{corr:bound:SIM}
	\left|\Pr_{\advA,\, \Real}[X]-\Pr_{\advA,\, \Ideal}[X]\right| \leq \varepsilon + \Pr[\E],
\end{equation} with $\E$ and $\varepsilon$ defined as in Lemma~\ref{l:g-o-i} and Theorem~\ref{thm:general-correctness}. 
With $\advA'$ and $\advA$ as above, both making $q_u, q_t$ queries to $\upO, \qryO$ respectively, we have $\Pr_{\advA',\, G'}[S_r] = \Pr_{\advA,\, \Real}[S_r]$.

Let $(S_r, q)$ be the event ``$S_r$ happened with the \amq containing at most $q$ elements'', and let $\Ideal'$ be a game equivalent to the $\Ideal$ experiment in $\gameRoI$, but defining oracles as in $G'$. Then the bounds from~\cite{_CCS:ClaPatShr19} can be written as $\Pr_{\advA', G'}[S_r] \leq \varepsilon + \Pr_{\advA',\Ideal'}[(S_r,n+q_u+r)]$, 
where $\varepsilon$ is the PRF distinguishing advantage. 
As before, we can then write
\begin{align}\label{corr:bound:GAME}
	\Pr_{\advA', G'}[S_r] \leq
	\varepsilon +\hspace{-0.5em}\Pr_{\advA,\, \Ideal}[(S_r,n{+}q_u{+}r)].
\end{align}

Let $\dD_{S_r}$ be a distinguisher that outputs $1$ if the event $S_r$ happened given $\advA$'s transcript (which includes $V$), as in \gameRoI. Using the fact that $\Pr_{\advA', G'}[S_r] = \Pr_{\advA,\Real}[S_r] = \Pr[\Real(\advA,\dD_{S_r})]$, we can use our simulation-based proof to derive new bounds in the setting of~\cite{_CCS:ClaPatShr19}. Indeed,

\begin{align*}
&\abs{\Pr_{\advA',\, G'}[S_r]-\Pr_{\advA',\, \Ideal'}[S_r,n+q_u]} \\
&\quad = \abs{\Pr[\Real(\advA,\, \dD_{S_r})=1]- \Pr[\Ideal(\advA,\,\dD_{S_r},\,\simS)=1]}.
\end{align*}
Using \eqref{corr:bound:SIM}, we obtain
\begin{align}\label{corr:bound:GAME:SIM}
	\Pr_{\advA',\, G'}[S_r] \leq \varepsilon + \Pr_\Ideal[\E] + \Pr_{\advA,\, \Ideal}[(S_r, n+q_u)].
\end{align}
The resulting difference between \eqref{corr:bound:GAME:SIM} and \eqref{corr:bound:GAME} is the trade-off between $\Pr_{\advA,\, \Ideal}[(S_r,\, n+q_u+r)]$ and $\Pr_\Ideal[\E] + \Pr_{\advA,\, \Ideal}[(S_r, \,n+q_u)]$. This trade-off is natural if one considers that by not specifying $r$, we introduce $\Pr_\Ideal[\E]$, where $\neg\E$ implies $\advA$ does not find more false positives than expected, which is what $\advA'$ is trying to achieve in $G'$.

\paragraph{Analysis of target-set coverage attacks.}
To conclude this discussion, we note that Theorem \ref{thm:general-correctness} can be applied to a broader range of predicates than $S_r$ considered above. For example, we can define a predicate which can be used to study target-set coverage attacks: $S_{\text{TSC}, L} =$ [``$\advA$ makes target set $L$ a set of false positives''] where $(\advA,\dD)$ are defined as follows.

Let adversary $\advA$ be such that it hardcodes the set $L$, and let $\advA$ be such that it makes all its queries making sure not to use $\upO$ on elements from $L$.
At the end of the game $\advA$ checks if it succeeded in making the set $L$ a set of false positives by querying the whole set $L$ (if it did not already). $\advA$ outputs $\top$ in this case; otherwise $\advA$ outputs $\bot$. Now $\dD$ just outputs whatever it gets from $\advA$.

A reduction between the problems of finding false-positive elements in a keyed \amq and of covering a target set in a keyed \amq is not known, making it unclear whether~\cite{_CCS:ClaPatShr19}'s results say anything about the latter problem. However, our Theorem~\ref{thm:general-correctness} allows us to bound the probability of the latter directly. Indeed, for such an $(\advA,\dD)$ pair, Theorem \ref{thm:general-correctness} implies
\begin{align}\label{corr:bound:GAME:TSC}
	\Pr_{\advA,\, \Real}[ S_{\text{TSC}, L}] \leq \varepsilon + \Pr_\Ideal[\E] + \Pr_{\advA,\, \Ideal}[(S_{\text{TSC}, L}, n+q_u)],
\end{align}
where $\Pr_{\advA,\, \Real}[S_{\text{TSC}, L}]$  is the probability of an adversary succeeding in the target-set coverage attack on $L$ in the real world. 
This implies that, with the above selection of $(\advA,\dD)$ pair, Theorem \ref{thm:general-correctness} also cleanly provides an upper bound on the success probability of target-set coverage attacks.

\subsection{Immutable setting}
\label{sec:immutable}
	In the ``immutable'' setting from~\cite{_CCS:ClaPatShr19} where one forbids calls to the $\upO$ oracle by setting $q_u = 0$, one can use a different proof approach to bound the distance between the $G^*$ and $\Ideal$ games in Fig.~\ref{fig:correctness-simulator}.
	
	In particular, after the $\repsimO$ call (identical in 
	both games), one would look at the statistical distance of the output of $\qrysimO$ calls. Since no 
	changes are made to the state of the simulator by $\advA$ past the $\repsimO$ call, every first 
	$\qryO$ query on a non-inserted element will return a false positive result  with probability $\Pr[\top \gets \qry^{\identity_\range}(Y \getsr \range, \sigma)]$,
	the false positive probability for the \amq given the state,
	in both games. This means that $G^*$ and $\Ideal$'s outputs are identically distributed, and  $\text{Adv}_{\Pi, \mathcal{A, S}}^{RoI} (\dD) 
	\leq \varepsilon$. 
	
	This analysis highlights that even if the event $a_i^\Ideal \neq a_i^{G^*}$ occurs for some $i \in [q_t]$, it does not lead to distinguishing the two worlds in the immutable scenario.

\section{Privacy}\label{app:privacy-proofs}

\subsection{Proof of Lemma \ref{l:f-dec:PI}}\label{app:proof-f-dec:PI}
\fdecomparezeropi*

\begin{proof}
	Since $\Pi$ satisfies $F$-decomposability, we can rewrite
	every $\upS^F(\pi(x),\, \sigma)$ call (resp.~$\qryS^F(\pi(x),\, \sigma)$ call) as $\upS^{\identity_\range}(F(\pi(x)),\, \sigma)$ (resp.~$\qryS^{\identity_\range}(F(\pi(x)),\, \sigma)$).
	Let $\advB$ be a $(q_u,q_t,q_v,t_a)$-PI adversary.
	Since $F$ is a random function, $F(\pi(x))$ is a random function in both $c = 0,\,1$ cases of $\gamePI$.
	Therefore, the games played by $\advB$ in $\gamePI$ are identical up to using two different random functions to preprocess data before inserting or querying it. 
	$\advB$'s views are then statistically identical, and 
	\begin{align*}
		\text{Adv}_{\Pi}^{PI}(\advB) = \abs{\Pr[\gamePI(\advB) = 1|c\,{=}\,0] {-} \Pr[\gamePI(\advB) = 1|c\,{=}\,1]}\,{=}\,0. \ACM{}{\queedee}
	\end{align*}  
\end{proof}

\subsection{Proof of Theorem \ref{th:general-privacy}}\label{app:proof-general-privacy}
\newelemrepprivacy*
\begin{proof}

Let $\advA = (\advA_1, \advA_2)$ be an adversary playing the game $\gameRoIERP$ in Fig.~\ref{game:priv:ElemRep:r-o-i} with any $t_d$-distinguisher $\dD$. Recall that $\Real$ (resp. $\Ideal$) denotes the $d = 0$ (resp. $d = 1$) version of $\gameRoIERP$.
The proof works by constructing a simulator $\simS$ for $\Ideal$ and an adversary $\advB$ for the $\gamePI$ game in Fig.~\ref{g:priv:PI:nn} that internally runs $\advA$, and relating their success probabilities. 

By definition of Elem-Rep privacy and by $\Pi$ being $\varepsilon$-PI, $\advA$'s advantage in $\gameRoIERP$ is 
\begin{align*}
	\text{Adv}&_{\Pi, \advA, \simS}^{RoIElemRepP} (\dD){\coloneqq}\abs{\Pr[\Real(\advA, \dD)\,{=}\,1] {-} \Pr[\Ideal(\advA, \dD, \simS){=}1]} \\
	&\le \abs{\Pr[\Real(\advA, \dD) = 1] - \Pr[\gamePI(\advB) = 1 \mid c = 0]} \\
	&\quad + \abs{\Pr[\gamePI(\advB) = 1 \mid c = 0] - \Pr[\gamePI(\advB) = 1 \mid c = 1]} \\
	&\quad + \abs{\Pr[\gamePI(\advB) = 1 \mid c = 1]- \Pr[\Ideal(\advA, \dD, \simS) = 1]} \\	
	&\le \abs{\Pr[\Real(\advA, \dD) = 1] - \Pr[\gamePI(\advB) = 1 \mid c = 0]} \hspace{2.8em} (\star) \\
	&\quad + \varepsilon \\
	&\quad + \abs{\Pr[\gamePI(\advB) = 1 \mid c = 1]- \Pr[\Ideal(\advA, \dD, \simS) = 1]}. \hspace{.6em} (\dagger)
\end{align*}
Hence, to prove our result we proceed to bound $(\star)$ and $(\dagger)$.

We start with $(\star)$. Consider the $\varepsilon$-PI adversary $\advB$ defined in Fig.~\ref{fig:priv:ElemRep:PIadvB}.
\begin{figure}
	\centering
	\begin{pcvstack}[boxed]
		\procedure[linenumbering,headlinecmd={\vspace{.1em}\hrule\vspace{.3em}}]{PI adversary $\advB^{\scriptsize \repO, \pcbox{\upO, \qryO,} \revO}$}{%
			\setV \getsr \advA_1\\
			\repO(\setV) \\ 
			\pcreturn d' \getsr \dD(\advA_2^{\scriptsize \pcbox{\upO, \qryO,}\revO}, \setV)
		}
	\end{pcvstack}
	\caption{PI adversary $\advB$.}
	\label{fig:priv:ElemRep:PIadvB}
\end{figure}
By inspection, $\Real$ and the $\gamePI$ game when $c = 0$ are identical, except for $\Real$ sampling ``$K \getsr \Kspace;\ F \gets R_K$'' and $\gamePI$ sampling ``$F \getsr \Funcs[\domain,\range]$''. Then, by assumption of the PRF security of $R_K$, 
\begin{align*}
	(\star) = \abs{\Pr[\Real(\advA, \dD) = 1] - \Pr[\gamePI(\advB) = 1 \mid c = 0]} \le \varepsilon_{\text{PRF}}.
\end{align*}

Now we look at $(\dagger)$. We start by recalling the $\gamePI$ game when $c = 1$ in the context of this proof.
Here, $\advA = (\advA_1, \advA_2)$, simulated inside of $\advB$, interacts with $\Pi$, with every element queried to the $\repO$, $\upO$ and $\qryO$ oracles being permuted by a random permutation $\pi \colon \domain \rightarrow \domain$.
Letting $V$ be the set output by $\advA_1$, we can illustrate this as
\begin{center}
	\begin{tikzpicture}
		
		\node [draw,
		circle,
		minimum size =2.5cm,
		label={135:$\domain$}] (D) at (0,0){};
		
		\node [draw,
		circle,
		minimum size =.8cm] (V) at (-.5,.5){$V$};
		
		\node [draw,
		circle,
		minimum size =2.5cm,
		label={45:$\pi(\domain) = \domain$}] (piD) at (4,0){};
		
		\node [draw,
		circle,
		minimum size =.8cm] (piV) at (3.75,.5){$\pi(V)$};

		\draw[->] (D) -- (piD) node[midway,above] {$\pi$};
 		\draw[->,dashed] (V) -- (piV);
		
	\end{tikzpicture}
\end{center}

Now consider the simulator $\simS = \simS_\text{Elem-Rep}$ in Fig.~\ref{f:SElemRep}, to be used inside $\Ideal$.
In $\Ideal$, $\advA_1$ outputs $V$, but this is not passed to $\simS$. Instead, the simulator uses its $\repleakO$ oracle to sample a random set $Y \subset \domain$ with the same cardinality as $V$, and inserts it into an internally instantiated $\amq$ with the same public parameters as $\Pi$. 
$\simS$ then runs $\advA_2$. Let $W$ be the set of every element $x \in \domain$ queried to the $\qrysimO$ and $\upsimO$ oracles by $\advA_2$.
In order to provide $\advA_2$ with a similar view to the one inside the $\gamePI(\advB)$ game when $c = 1$, $\simS$ lazily samples an injective function $\pi' \colon W \rightarrow \domain$ using its $\PermuteO$ procedure, and stores it into a key-value store $P$. In particular, $\simS$ uses the $\elemleakO$ oracle to consistently map elements in $V$ into $Y$, and elements in $\domain \setminus V$ into $\domain \setminus Y$, in order to keep consistency with the sampling made inside of $\repsimO$.
We can illustrate the effect of $\pi'$ as
\begin{center}
	\begin{tikzpicture}
		
		\node [draw,
		circle,
		minimum size =2.5cm,
		label={135:$\domain$}] (D) at (0,0){};
		
		\node [draw,
		circle,
		minimum size =.8cm] (V) at (-.5,.5){$V$};
		
		\node [draw,
		circle,
		minimum size =1.5cm] (W) at (-.05,-.4){$W$};
		
		\node [draw,
		circle,
		minimum size =2.5cm,
		label={45:$\domain$}] (piD) at (4,0){};
		
		\node [draw,
		circle,
		minimum size =.8cm] (piV) at (3.75,.5){$Y$};

		\node [draw,
		circle,
		minimum size =1.5cm] (piW) at (4.05,-.4){$\pi'(W)$};

		\node (VW) at (-0.4,.25){};
		\node (piVW) at (3.95,.25){};
		
		
		\draw[->] (W) -- (piW) node[midway,above] {$\pi'$};
		\draw[->,dashed] (VW) -- (piVW) node[midway,above] {$\pi'$};
		
	\end{tikzpicture}
\end{center}	

From the point of view of $\advA$, $\pi'$ and $\pi$ are indistinguishable, and hence $\Ideal$ and the $\gamePI$ game when $c = 1$ are identical. Therefore,
\begin{align*}
	(\dagger) = \abs{\Pr[\gamePI(\advB) = 1 \mid c = 1]- \Pr[\Ideal(\advA, \dD, \simS) = 1]} = 0.
\end{align*}
	
Finally, we can collect the upper bounds and obtain the final result,
\begin{align*}
	\text{Adv}_{\Pi, \advA, \simS}^{RoIElemRepP} (\dD)&\,{\coloneqq}\,\abs{\Pr[\Real(\advA, \dD)\,{=}\,1] {-} \Pr[\Ideal(\advA, \dD, \simS)\,{=}\,1]} \\
	&\le (\star) + \varepsilon + (\dagger) \le \varepsilon_{\text{PRF}} + \varepsilon. \ACM{}{\queedee}
\end{align*}
\end{proof}

\subsection{Proof of Theorem \ref{th:er-to-r}}\label{app:proof-er-to-r}
\oldrepprivacythm*
\begin{proof}
Recall that $\Real$ (resp. $\Ideal$) denotes the $d = 0$ (resp. $d = 1$) version of $\gameRoIERP$ in Fig.~\ref{game:priv:ElemRep:r-o-i}, and $\overline{\Real}$ (resp. $\overline{\Ideal}$) denotes the $d = 0$ (resp. $d = 1$) version of $\gameRoIRP$ in Fig.~\ref{game:priv:Rep:r-o-i}.

	Let $\simS'$ be the same as $\simS$, but where every call to $\elemleakO$ is replaced with $\bot$. Then, the view that $\advA_2$ sees through $\simS$ and $\simS'$ are the same at least up until $\advA_2$ makes its first $\upO,\qryO$ query on something in $V$. Thus for every $t_d$-distinguisher $\dD$,
	
	\begin{align*}
		\left|
		\pr{\Ideal(\advA, \dD, \simS){=}1} 
		{-}\pr{\overline{\Ideal}(\advA, \dD,\simS'){=}1}
		\right| \nn
		\leq \Pr[{W \cap \setV \not= \emptyset}].
	\end{align*}
	Observing that $\Real, \overline{\Real}$ are identical and using the $\varepsilon$-Elem-Rep privacy of $\Pi$,
	\begin{align*}
		\text{Adv}&_{\Pi, \advA, \simS'}^{RoIRepP} (\dD)
		= \left|
		\pr{\overline{\Real}(\advA, \dD){=}1}
		{-}\pr{\overline{\Ideal}(\advA, \dD,\simS'){=}1}
		\right| \nn\\
		&= \left|
		\pr{\Real(\advA, \dD){=}1}
		{-}\pr{\overline{\Ideal}(\advA, \dD,\simS'){=}1}
		\right| \nn\\
		&\leq \left|
		\pr{\Real(\advA, \dD){=}1}
		{-}\pr{\Ideal(\advA, \dD,\simS){=}1}
		\right| \nn\\
		&\quad + \left|
		\pr{\Ideal(\advA, \dD, \simS){=}1}
		{-}\pr{\overline{\Ideal}(\advA, \dD,\simS'){=}1}
		\right| \nn\\
		&\leq \varepsilon + \Pr[{W \cap \setV \not= \emptyset}]. \ACM{}{\queedee}
	\end{align*}

\end{proof}

\section{More on secure instances}\label{app:secure-instantiations} 

We explain in detail how to use our results to instantiate \amq instances achieving provable guarantees, expanding on \S~\ref{sec:secure-instantiations}.
We emphasize that secure instantiation using our results requires three steps: first, either replacing hash functions with PRFs or wrapping the $\amq$ algorithms with a PRF evaluation; second, constraining the budget of queries allowed to the adversary; third, using our bounds to choose parameters for the $\amq$. The impact on efficiency of these measures comes mainly from the third step. For example, when considering an adversarial setting over a non-adversarial one, in order to achieve the same correctness guarantees the designer will have to choose bigger parameters, which will affect the space and runtime efficiency. Crucially, the use of PRFs will not imply a significant overhead over a construction using plain hash functions, since, as explained in the introduction, PRFs can be efficiently instantiated. In this section, we assume $\varepsilon$-secure PRFs with $\varepsilon = 2^{-256}$.

\subsection{Correctness}
Recall the guarantee from Theorem~\ref{thm:general-correctness}:
\begin{align}
	\text{Adv}_{\Pi, \mathcal{A, S}}^{{\text{RoI}}} (\dD) &=  \abs{\pr{\Real(\advA, \dD) {=} 1} {-} \pr{\Ideal(\advA, \dD,\simS) {=} 1}} \nonumber \\
	&= \abs{\pr{\dD(\advA) {=} 1} {-} \pr{\dD(\simS(\advA, pp)) {=} 1}} \nonumber \\
	& \leq \varepsilon +  2 q_t \cdot \pfp{n + q_u}, \label{eqn:thm1-sec-inst}
\end{align}
where $\varepsilon$ is a PRF distinguishing advantage, $n$ is the number of elements initially inserted into the \amq, and $q \coloneqq q_u + q_t$ is the total number of queries made by $\advA$ that influence her success probability.\footnote{Note that the number of $\revO$ queries $q_v$ does not play a role in bounding $\advA$'s success probability.} Crucially, usually $\pfp{n + q_u}$ can be estimated using well-established upper bounds, such as those provided by Lemmas~\ref{lem:bf-false-positive-probability} and~\ref{clm:cf-variants-fp-bound}.

This result allows us to establish an upper bound on the probability $\Pr[\dD(\advA)\,{=}\,1]$ of an adversary $\advA$ with the given query budget $q$ finding a sequence of queries to an \amq $\Pi$ that allows them to satisfy some desired predicate $P$ in the $\Real$ world, by relating this probability to that of $\advA$ satisfying the predicate in the $\Ideal$ world, where analysis of the probability $\Pr[\dD(\simS(\advA, pp))\,{=}\,1]$ is simplified by the extra guarantee of the state of $\Pi$ being non-adversarially-influenced at every point in time. For example, if the predicate is on the state $\sigma$ of the \amq after $\advA$'s run, this can be captured by having $\advA$ return $\sigma \gets \revO()$ at the end of her run, and $\dD$ return whether $P(\sigma)$ is satisfied or not.

Two example predicates for which designers may want to bound $\pr{\dD(\advA) {=} 1}$ could be $S_r \coloneqq [\advA$ observes $r$ false positives during her run$]$, the predicate investigated in~\cite{_CCS:ClaPatShr19}, and $H_w \coloneqq [$the state bitvector $\sigma$ of the Bloom filter $\Pi$ has Hamming weight exactly $w$ at some point during $\advA$'s run$]$.
Crucially, the best approach to analysing $\pr{\dD(\simS(\advA, pp)) {=} 1}$ may differ between different predicates.
In the case of $S_r$, the optimal strategy for an adversary in the $\Ideal$ world against a \emph{non-decreasing membership probability} \amq $\Pi$ (see Def.~\ref{def:bf-consistency-rules}) would be making all $n + q_u$ possible insertions into $\Pi$ first and then using all $q_t$ available $\qryO$ queries to try to produce $r$ false positive answers. As shown in~\cite{_CCS:ClaPatShr19}, the success probability of this adversary follows from a binomial distribution analysis using parameters $p =\pfp{n+q_u}$ and $q_t$.\footnote{Additionally, using the Chernoff bound to obtain a closed formula for the resulting probability.}
On the other hand, an adversary trying to satisfy predicate $H_w$ for some specific value of $w$ may achieve a higher success probability by not using all of her available $\upO$ queries, requiring a more careful derivation of the success probability in $\Ideal$.

As a practical example, we investigate the choice of adversarially correct parameters for Bloom and PRF-wrapped Cuckoo filters for three possible query budgets. Concretely, we bound the probability that after $n+q_u$ adversarial insertions, querying a random non-inserted element returns a false positive result. This is the adversarial variant of the commonly performed analysis for choosing parameters for Bloom filters, and can be captured by letting $\advA$ make $n+q_u$ $\upO$ queries on distinct elements first, and then output the result of a single $\qryO$ query on a non-inserted element. The output is passed to a trivial distinguisher $\dD$ which on input $b$ returns $b$. For different values of $k$ for Bloom filters and $s$ for Cuckoo filters, we plot the false positive probability against the bit-size of the data structure. We also plot the same curves when performing an honest-case analysis where the queries are considered non-adversarial.

\paragraph{Finding a false positive after $n+q_u$ insertions.}
We start by creating a query budget. This can be done by first estimating the total number of elements that are expected to be inserted by honest users into the \amq $\Pi$, and assigning this as the value of $n$. Then, one proceeds to estimate a maximum number of queries $q$ that $\advA$ is allowed to make. This could be determined based on the expected shelf-life of $\Pi$, and on rate-limiting done on the $\qryO$ and $\upO$ interfaces. We consider a situation where the total number of queries $q = q_u + q_t$ is bounded, but not the exact values of $q_u$ and~$q_t$.

We then proceed to analyse the set of candidate public parameters $pp$, and compute an upper bound to \eqref{eqn:thm1-sec-inst},\footnotemark
given $q_t = t$ and $q_u = q-t$ for every value of $t \in \{0,\dots,q\}$, storing the attained maximum as a function of $t$. This gives us an upper bound 
on the maximum adversarial advantage $\text{Adv}_{\Pi, \mathcal{A, S}}^{{\text{RoI}}} (\dD)$ for $q$ and $pp$. \footnotetext{Using the well-established formulae for $\overline{\pfpsym}$ in Lemmas~\ref{lem:bf-false-positive-probability}
	and~\ref{clm:cf-variants-fp-bound}.} 
Then, given our predicate, we upper bound  $\pr{\dD(\simS(\advA, pp)) {=} 1}$ as $\opfp{n+q_u} = \opfp{n+q-t}$, resulting in an upper bound to $\pr{\dD(\advA) {=} 1}$ in the $\Real$ world.

Finally, in Fig.~\ref{fig:correct-instances} we plot the upper bound on the maximum probability $\pr{\dD(\advA) {=} 1}$ attained by $\advA$ for any given value of $pp$, ordering these by their storage cost: $m$ for Bloom filters, $s \cdot 2^{\lambda_I} \cdot \lambda_T$ for Cuckoo filters.
In the case of Bloom filters, we produce various trade-offs for various values of $k$; we do similarly for Cuckoo filters based on values of $s$. We note that in this example we do not investigate other performance metrics, such as insertion or query runtime. 
We also plot trade-offs obtained from a non-adversarial analysis, where one considers $\pr{\dD(\advA) {=} 1} = \pr{\dD(\simS(\advA, pp)) {=} 1} = \pfp{n+q_u}$.
We stress that the difference in storage seems to grow by a constant factor ($\approx 2$ for Bloom filters, $\approx 3$ for Cuckoo filters) both in the case where the number of total insertions made is lower ($\approx 2^8$) and larger ($\approx 2^{30}$), hence across a broad range of practical parameters. We leave the analytical study of this growth factor to future work.

\begin{figure*}[t]
\Wider[10em]{
	\centering
	\begin{subfigure}{.45\textwidth}
		\centering
			\adjustbox{trim={0\width} {10pt} {0} {0}, clip}{
				\includegraphics[width=\textwidth]{\detokenize{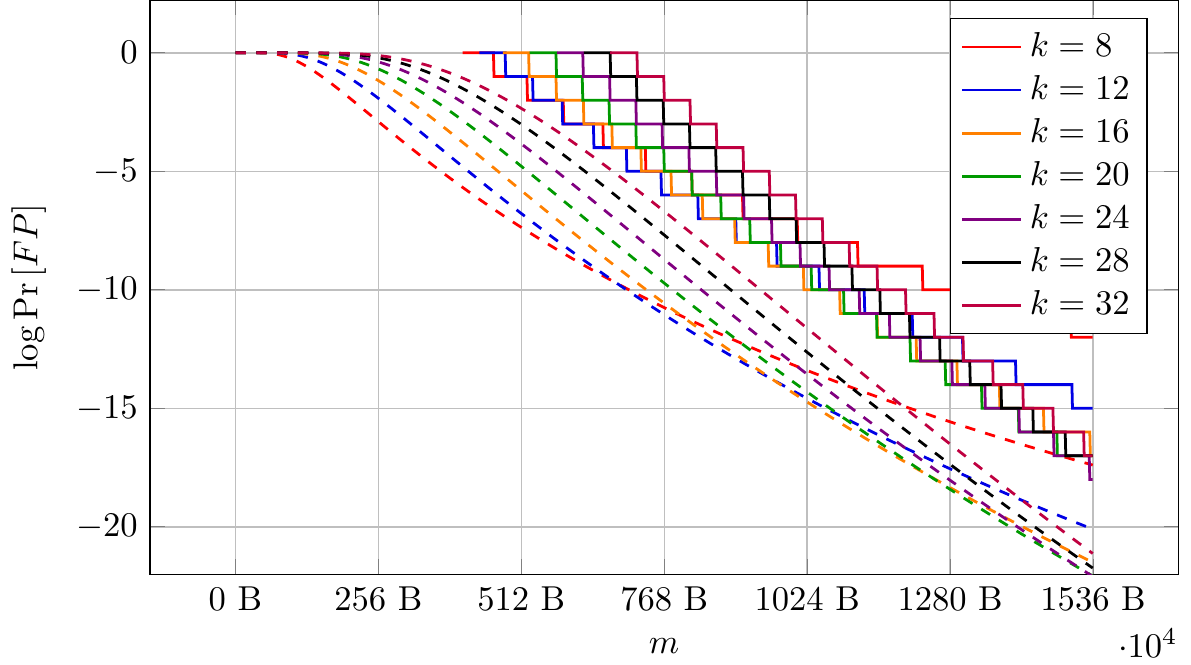}}}
		\caption{Bloom, $\log{n} = 7$, $\log{(q_u + q_t)} = 8$}\label{fig:bf-correct-instances-small-n-small-q}
	\end{subfigure}
	\begin{subfigure}{.45\textwidth}
	\centering
		\adjustbox{trim={0\width} {11pt} {0} {0}, clip}{
			\includegraphics[width=\textwidth]{\detokenize{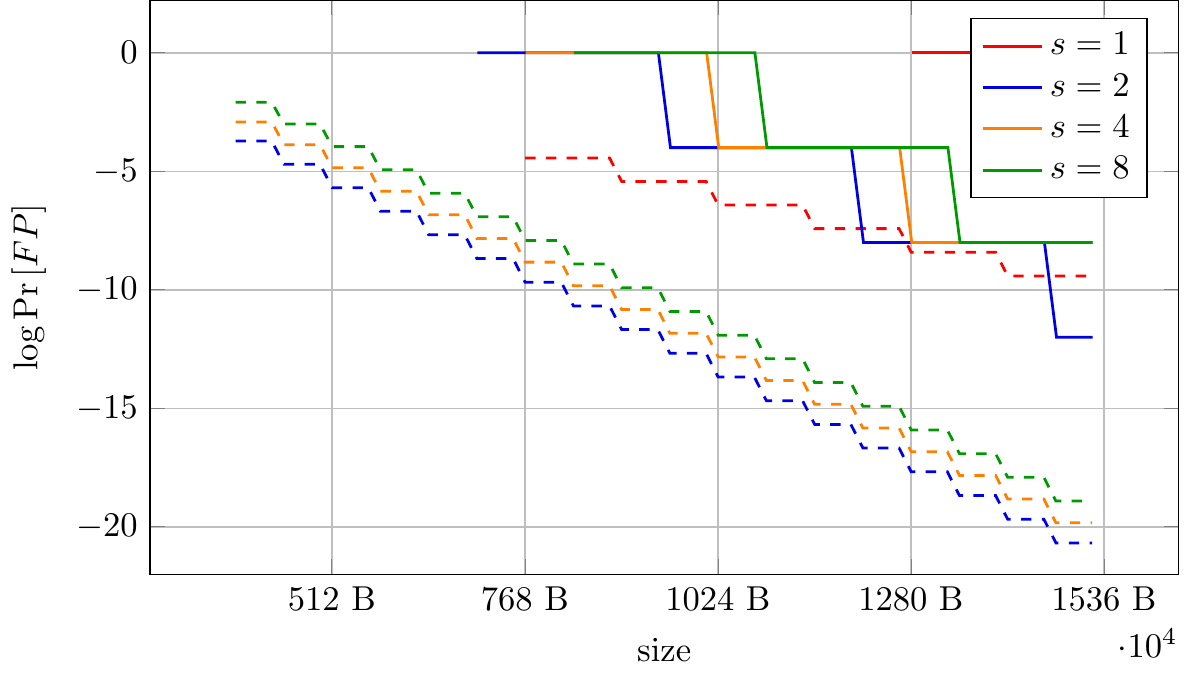}}}
	\caption{Cuckoo, $\log{n} = 7$, $\log{(q_u + q_t)} = 8$}\label{fig:cf-correct-instances-small-n-small-q}
	\end{subfigure}
	\begin{subfigure}{.45\textwidth}
		\centering
			\adjustbox{trim={0\width} {10pt} {0} {0}, clip}{
				\includegraphics[width=\textwidth]{\detokenize{plots/bf_log_pfp_vs_m_small_n_big_q.pdf}}}
		\caption{Bloom, $\log{n} = 7$, $\log{(q_u + q_t)} = 30$}\label{fig:bf-correct-instances-small-n-big-q}
	\end{subfigure}
	\begin{subfigure}{.45\textwidth}
	\centering
		\adjustbox{trim={0\width} {10pt} {0} {0}, clip}{
			\includegraphics[width=\textwidth]{\detokenize{plots/cf_log_pfp_vs_size_small_n_big_q.pdf}}}
	\caption{Cuckoo, $\log{n} = 7$, $\log{(q_u + q_t)} = 30$}\label{fig:cf-correct-instances-small-n-big-q}
	\end{subfigure}
	\begin{subfigure}{.45\textwidth}
		\centering
			\adjustbox{trim={0\width} {10pt} {0} {0}, clip}{
				\includegraphics[width=\textwidth]{\detokenize{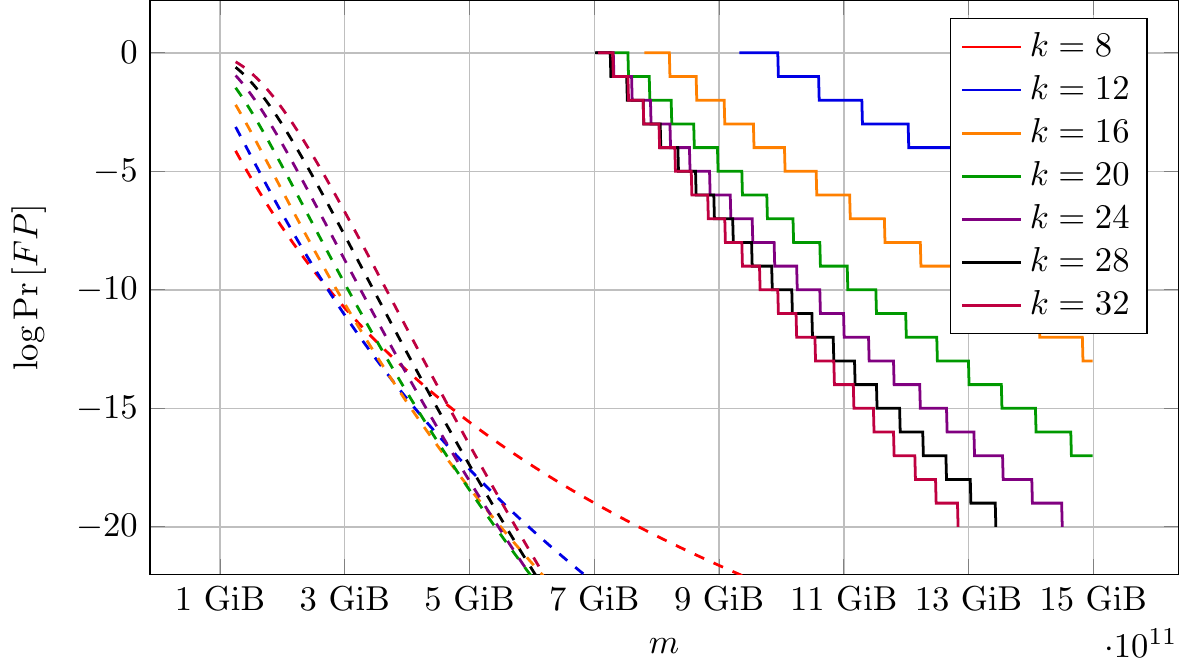}}}
		\caption{Bloom, $\log{n} = 29$, $\log{(q_u + q_t)} = 30$}\label{fig:bf-correct-instances-big-n-big-q}
	\end{subfigure}
	\begin{subfigure}{.45\textwidth}
		\centering
			\adjustbox{trim={0\width} {10pt} {0} {0}, clip}{
				\includegraphics[width=\textwidth]{\detokenize{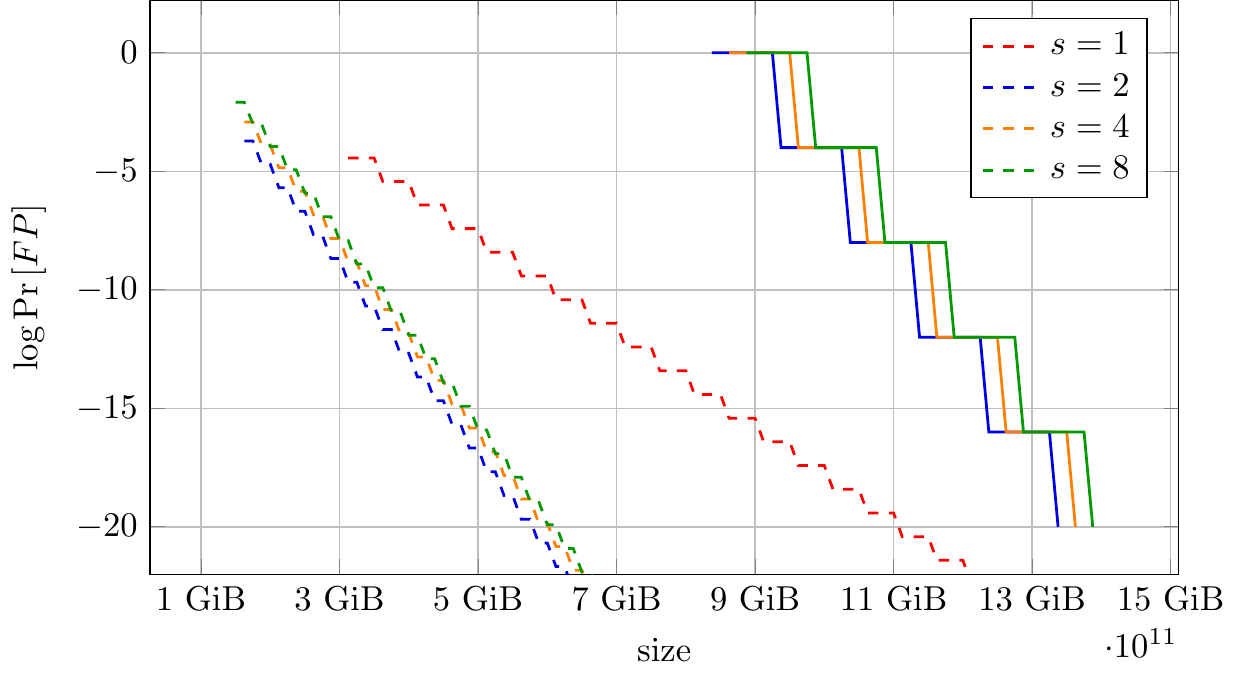}}}
		\caption{Cuckoo, $\log{n} = 29$, $\log{(q_u + q_t)} = 30$}\label{fig:cf-correct-instances-big-n-big-q}
	\end{subfigure}
}
\caption{Correctness guarantees vs.~size trade-offs for Bloom and PRF-wrapped insertion-only Cuckoo filters. Solid lines represent adversarial guarantees ($\log\Pr[FP] \ge \log \Pr[\dD(\advA) = 1]$). Dashed lines represent the values obtained assuming no adversarial influence ($\log\Pr[FP] = \log \opfp{n+q_u}$).}\label{fig:correct-instances}
\end{figure*}

\begin{remark}
	The topic of optimal parameter choice for Cuckoo filters has not been investigated as extensively as for Bloom filters.
	In particular, the procedure for choosing parameters used in~\cite{_CoNEXT:FAKM14} is based on an analysis of the probability of the $\upS$ algorithm being disabled. The authors provide a theoretical lower bound of this probability, and experimental evidence that for $s \ge 4$ and $\lambda_T \ge 6$ more than $95\%$ of the filter's slots will be filled before $\upS$ is disabled, whenever $num = 500$. This allows them to estimate $\lambda_I$ given an expected value for $n+q_u$. Unfortunately, the lower-bound analysis does not apply to the insertion-only variant of Cuckoo filters, and the experiments do not account for PRF wrapping. To produce our results in Fig.~\ref{fig:correct-instances}, we replicated their experimental analysis on PRF-wrapped insertion-only Cuckoo filters, assuming a PRF with $256$-bit outputs. Our results (see Fig.~\ref{fig:io-prf-cf-load-factor}) show that also in the case of PRF-wrapped insertion-only Cuckoo filters, for $s = 4,\,8$ at least $95\%$ of the slots are occupied when $num = 500$. Hence, we replicate the rest of their analysis to choose parameters. 
\end{remark}

\begin{figure*}[t]
	\Wider[10em]{
		\centering
		\begin{subfigure}{.45\textwidth}
			\centering
			\adjustbox{trim={0\width} {0pt} {0} {0}, clip}{
				\includegraphics[width=.9\textwidth]{\detokenize{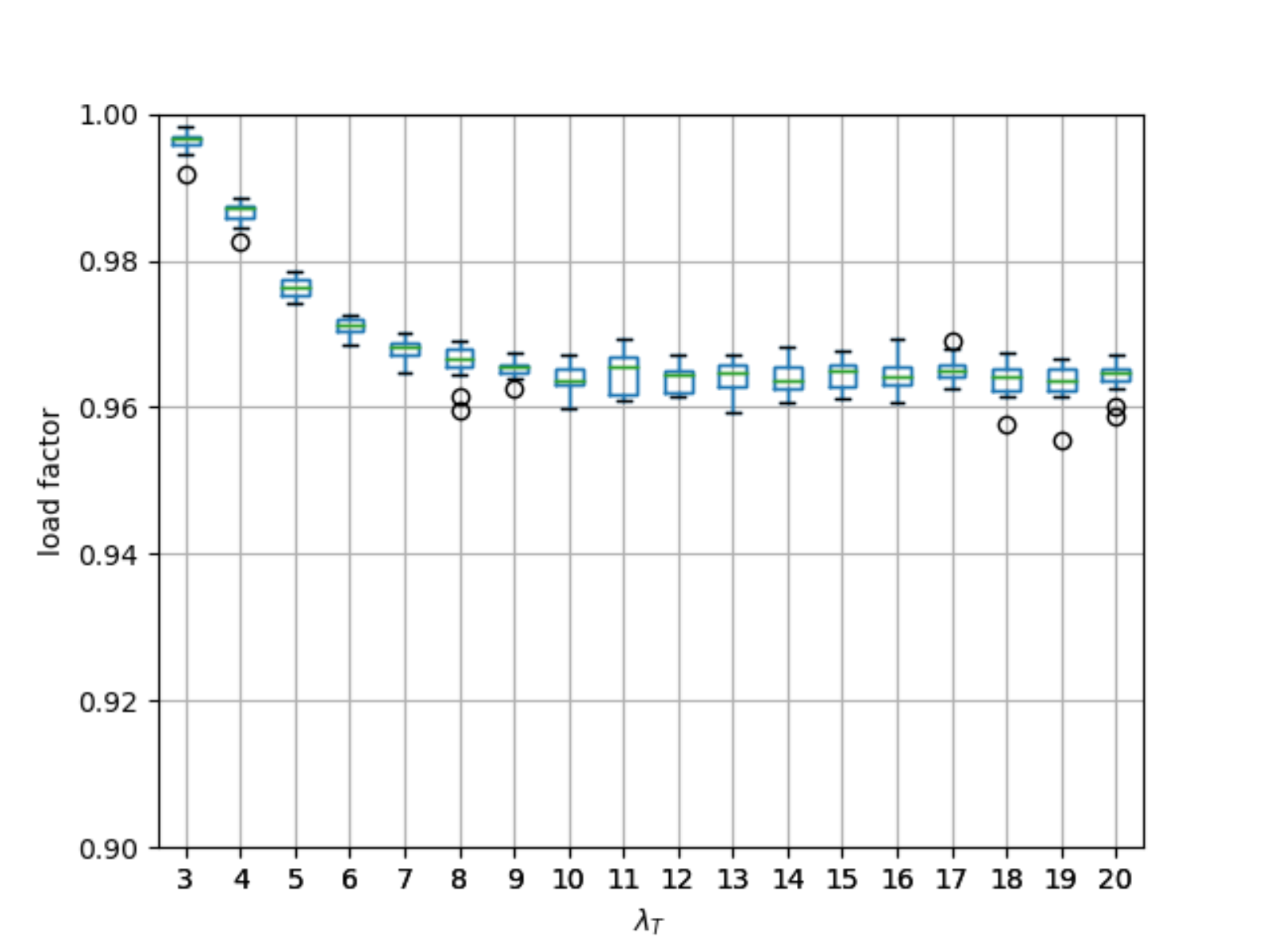}}}
		\caption{$\lambda_I = 15$, $s = 4$, $num = 500$}
		\end{subfigure}
		\begin{subfigure}{.45\textwidth}
		\centering
		\adjustbox{trim={0\width} {0pt} {0} {0}, clip}{
			\includegraphics[width=.9\textwidth]{\detokenize{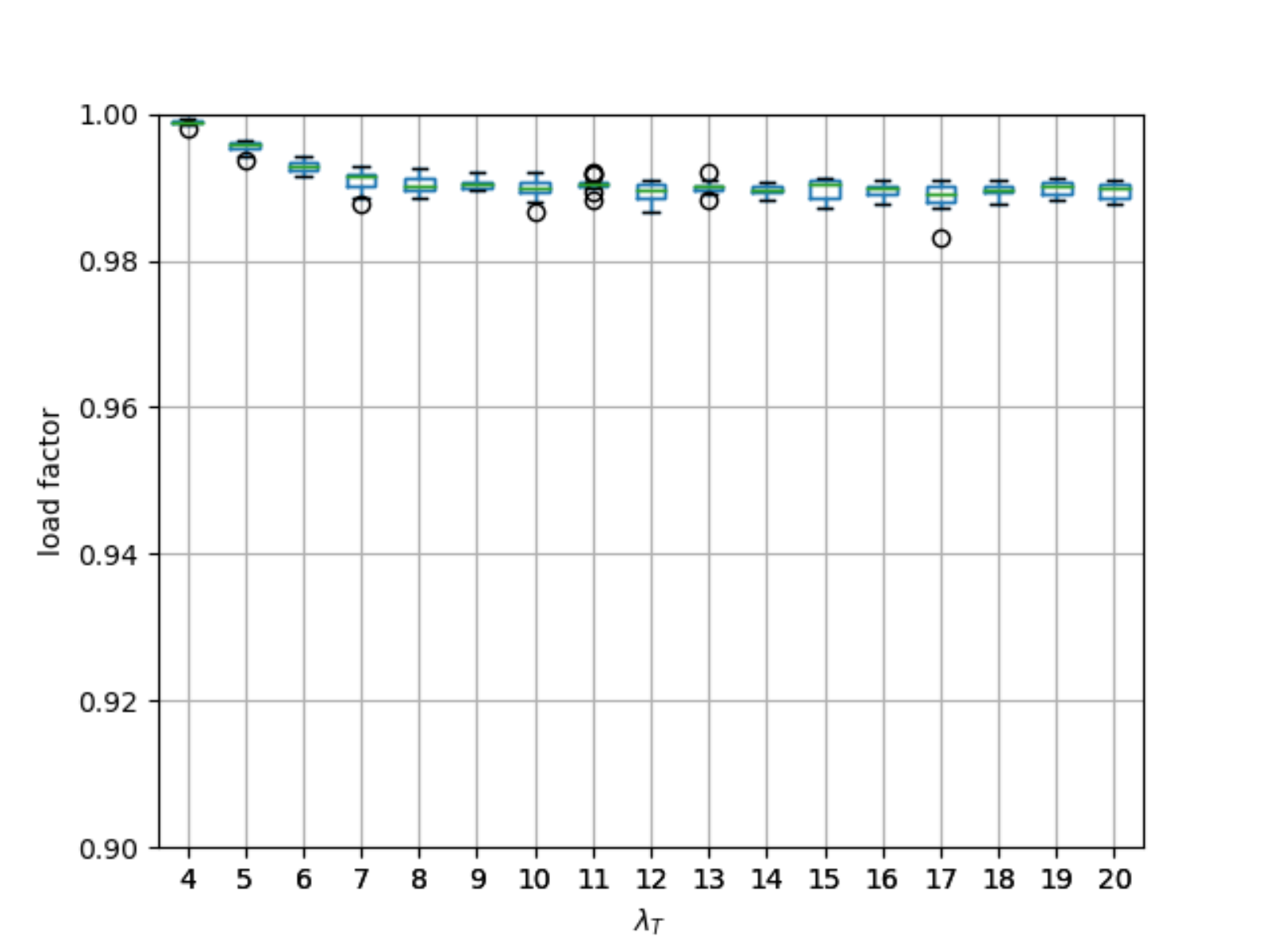}}}
		\caption{$\lambda_I = 15$, $s = 8$, $num = 500$}
		\end{subfigure}
		\begin{subfigure}{.45\textwidth}
		\centering
		\adjustbox{trim={0\width} {0pt} {0} {0}, clip}{
			\includegraphics[width=.9\textwidth]{\detokenize{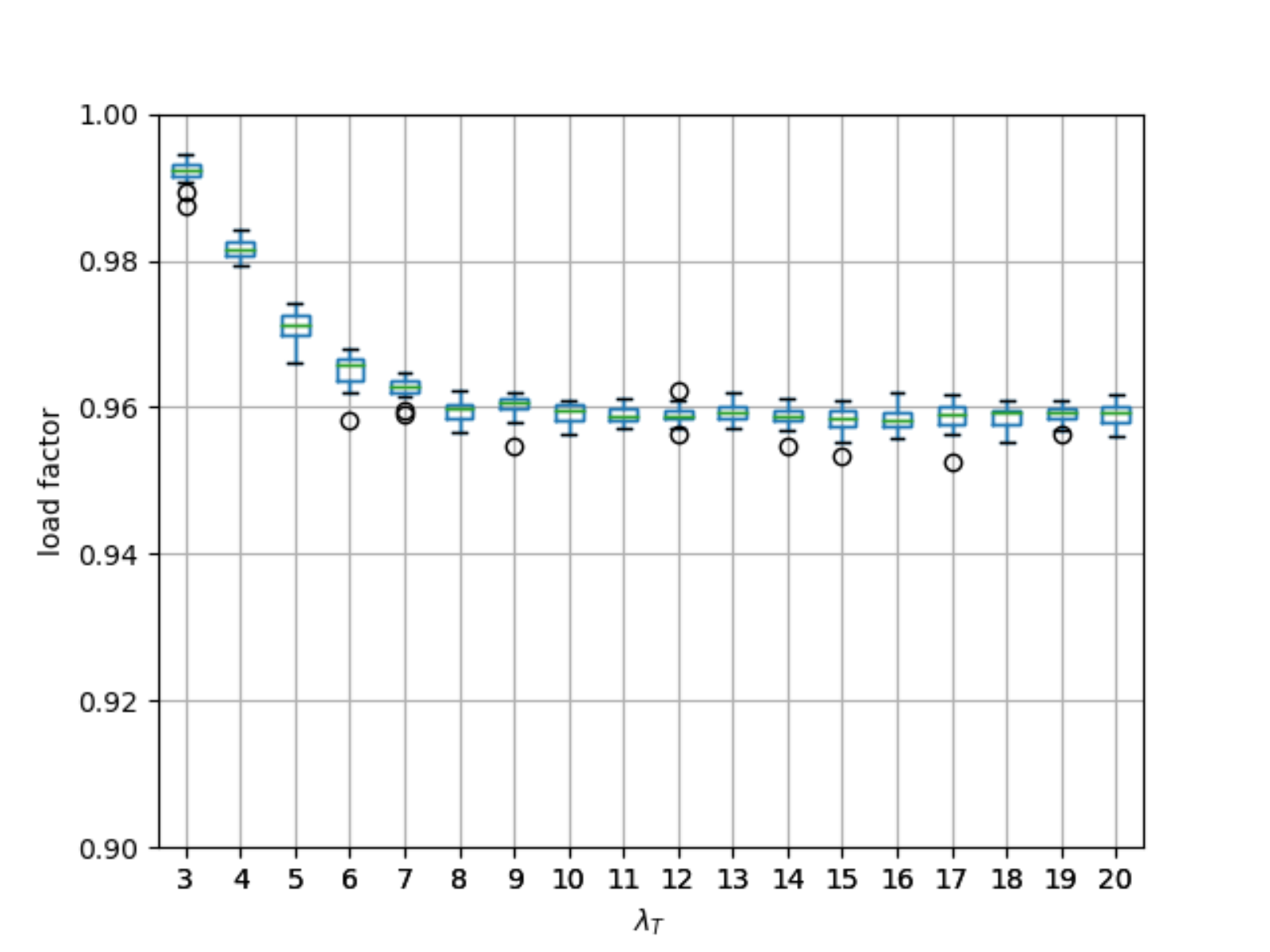}}}
		\caption{$\lambda_I = 20$, $s = 4$, $num = 500$}
		\end{subfigure}
		\begin{subfigure}{.45\textwidth}
			\centering
			\adjustbox{trim={0\width} {0pt} {0} {0}, clip}{
				\includegraphics[width=.9\textwidth]{\detokenize{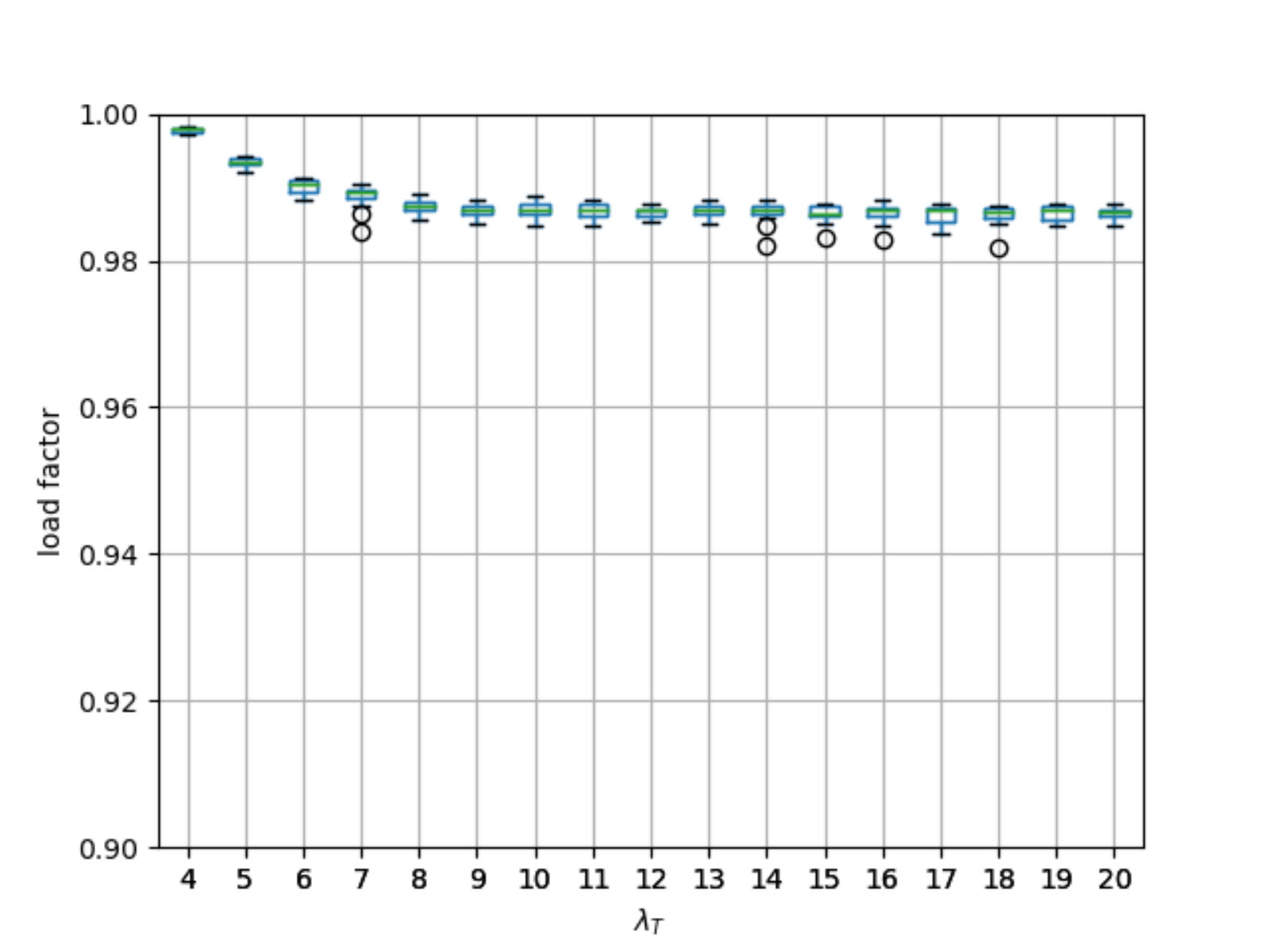}}}
			\caption{$\lambda_I = 20$, $s = 8$, $num = 500$}
		\end{subfigure}
	}
	\caption{Proportion of slots filed before $\upS$ is disabled (the ``load factor'' in~\cite{_CoNEXT:FAKM14}) for PRF-wrapped, insertion-only Cuckoo filters as a function of the tag length $\lambda_T$. Each plot is generated using data from 16 independent experiments.}\label{fig:io-prf-cf-load-factor}
\end{figure*}

\subsection{Privacy}

Contrary to our analysis of adversarial correctness, the privacy results we prove are independent of the \amq public parameters and instead depend on an estimation of the ``unguessability'' of the distribution output by $\advA_1$ (in this analysis, the honest user), given the query budget available to the adversary $\advA_2$.

It is likely that $\advA_1$'s distribution will be inherently application-specific, meaning that the only freedom from the designer's perspective is the query budget that can be allowed to the adversary. Similarly to the case of correctness, query budget constraints could be imposed by resource constraints, eg.~by limiting the life-span of the \amq, or by communication constraints, eg.~by imposing rate limiting of the $\qryO$ and $\upO$ oracles.

Unlike our correctness analysis, our privacy guarantees do not come from bounding the probability of $\advA$ achieving some predicate in an $\Ideal$ world. Rather, using Theorem~\ref{th:er-to-r}, we compute an upper bound on $\advA$'s ability to distinguish a \emph{function-decomposable} \amq from a simulation of it that knows at most only the cardinality of the set $V$ output by $\advA_1$. This gap is bounded by $\Pr[{W \cap \setV \ne \emptyset}]$, where $W$ is the set of elements queried by $\advA_2$ with a query budget of $q = q_t + q_u$ in the $\Ideal$ world. Whenever the gap is  small, a corresponding degree of indistinguishability, and hence privacy, is achieved.

\paragraph{An example computation.}
To compute a concrete privacy bound, a potentially useful tool is the min-entropy of the distribution output by $\advA_1$. In more detail, $\advA_1$ outputs a list $V = [x_1,\dots,x_n]$ from some distribution $D_{\advA_1}$ over $\plist(\domain)$.
We first define the conditional min-entropy of $D_{\advA_1}$ given $|V|$ to be
\begin{align*}
	H_\infty(D_{\advA_1} \mid |V| = \nu) \coloneqq -\log \max_{x \in \domain} \Pr[x \in V \mid |V| = \nu].
\end{align*}
Let $W$ be a random variable with support $\supp{W}$ describing the set of elements queried by $\advA_2$ to the $\qryO$ and $\upO$ oracles.
The best possible strategy for $\advA_2$ should be querying the elements most likely to be in $V$. $\advA_2$ could either directly test membership of these elements using $\qryO$, or in the case of \emph{reinsertion-invariant} \amq (see Def.~\ref{def:bf-consistency-rules}), indirectly test membership by attempting to insert likely elements and testing whether the state of the $\amq$ changed after the attempted insertion via the $\revO$ oracle.

To bound $\Pr[{W \cap \setV \ne \emptyset}]$ we start by bounding the conditional probability $\Pr[W \cap V \neq \emptyset \mid |V| = \nu]$ as
\clearpage
\begin{align*}
	&\Pr[W \cap V \neq \emptyset \mid |V| = \nu]\\
	&= \sum_{w \in \supp{W}} \Pr[W \cap V \neq \emptyset \wedge W = w \mid |V| = \nu] \\
	&= \sum_{w \in \supp{W}}\Pr\left[\bigvee_{i=1}^{|w|} (y_i \in V) \wedge W {=} w \ \vline\  |V| {=} \nu\right] 
	\begin{matrix*}[l]
		\enspace \text{where }\\ 
		 \enspace w = \{y_1,\dots,y_{|w|}\} 
	\end{matrix*} \\
	&\le \sum_{w \in \supp{W}} |w| \cdot \max_{y \in w} \Pr\left[y \in V \wedge W = w \ \vline\  |V| = \nu\right]  \\
	&= \sum_{w \in \supp{W}} |w| \cdot \max_{y \in w} \left \{
	\begin{matrix*}
		\Pr\left[y \in V \ \vline\ W = w, \,  |V| = \nu\right] \\
		\cdot \Pr\left[W = w \ \vline\  |V| = \nu\right]
	\end{matrix*}
	\right\} \\
	&= \hspace{-1em}\sum_{w \in \supp{W}} \hspace{-1em}|w| \cdot \max_{y \in w} \left \{
	\begin{matrix*}
		\hspace{-0.8em}\Pr\left[y \in V \ \vline\  |V|\,{=}\,\nu\right] \\
		\, \cdot \Pr\left[W\,{=}\,w \ \vline\  |V|\,{=}\,\nu\right] \,
	\end{matrix*}
	\right\}  \hspace{-.2em}\begin{matrix*}[l]
		\enspace \text{\small by $V$ and $W$}\\
		\enspace \text{\small being} \\
		\enspace \text{\small independent} \\
		\enspace \text{\small given $|V|$, as} \\
		\enspace \text{\small the simulator}\\
		\enspace \text{\small has only}\\ 
		\enspace \text{\small access to $|V|$}
		\end{matrix*} \\
	&= \sum_{w \in \supp{W}} |w| \cdot \Pr\left[W = w \ \vline\  |V| = \nu\right] \cdot \max_{y \in w} \Pr\left[y \in V \ \vline\  |V| = \nu\right] \\
	&\le (q_u + q_t) \sum_{w \in \supp{W}} 
	\begin{matrix*}
	\hspace{-4em}\Pr\left[W = w \ \vline\  |V| = \nu\right]\\
	\,\cdot\max_{y \in \domain} \Pr\left[y \in V \ \vline\  |V| = \nu\right]
	\end{matrix*} \\
	&\le (q_u + q_t) \cdot \max_{y \in \domain} \Pr\left[y \in V \ \vline\  |V| = \nu\right] \\
	&= (q_u + q_t) \cdot 2^{-H_\infty(D_{\advA_1} \mid |V| = \nu)}.
\end{align*}

We then conclude by observing
\begin{align*}
	\Pr&[{W \cap \setV \ne \emptyset}] = \sum_{\nu=1}^{n} \Pr[W \cap V \neq \emptyset \mid |V| = \nu] \Pr[|V| = \nu]\\
	&\le \sum_{\nu=1}^{n} (q_u + q_t) \cdot 2^{-H_\infty(D_{\advA_1} \mid |V| = \nu)} \Pr[|V| = \nu]\\
	&\le (q_u + q_t) \cdot \max_{\nu \in [n]} 2^{-H_\infty(D_{\advA_1} \mid |V| = \nu)} \sum_{\nu=1}^{n}  \Pr[|V| = \nu]\\
	&= (q_u + q_t) \cdot \max_{\nu \in [n]} 2^{-H_\infty(D_{\advA_1} \mid |V| = \nu)},
\end{align*}
hence obtaining a closed formula for a bound on the distance between the $\Real$ and $\Ideal$ worlds. High min-entropy settings may imply indistinguishability of the worlds for significantly large values of $q_u + q_t$. Lower min-entropy settings may require strict query budget constraining to attain indistinguishability.

\section{Hiding set cardinality}\label{app:leaking-nothing}

In \S~\ref{sec:privacy-other-pds} we noted that we assumed so far that in the context of achieving privacy, leaking the cardinality of the set being stored in an \amq is acceptable. 
If this were not to be an acceptable assumption, different designs may need to be investigated.
For example, one may attempt to achieve a smaller leakage profile by using some \amq with query runtime independent of $|V|$ by removing access to $\revO$ from the adversary. If the server's storage was untrusted, the PDS could be kept on a cryptographic enclave, or encrypted under a homomorphic encryption scheme and operated upon in encrypted form, with the scheme's secret key kept in the enclave. While the frequency of false positive $\qryO$ answers may still allow estimation of $|V|$, it may be harder in practice than it is by simply querying $\revO$.
}{}

\end{document}